\documentclass[a4paper,UKenglish,cleveref, autoref, thm-restate,authorcolumns]{lipics-v2019}

\usepackage{microtype}

\usepackage{hyperref}
\usepackage{amssymb}
\usepackage{amsfonts}
\usepackage{amsmath}
\usepackage{mathrsfs}
\usepackage{latexsym}
\usepackage{graphicx}
\usepackage[utf8]{inputenc}
\usepackage{verbatim}
\usepackage{enumerate}
\usepackage{xspace,cite}
\usepackage[normalem]{ulem}

\newtheorem{property}{Property}

 \renewcommand{\O}{\mathcal{O}}
 \renewcommand{\P}{{\sf P}}
  \renewcommand{\H}{\mathcal{H}}

 \newcommand{\F}{\mathcal{F}}
 \newcommand{\A}{\mathcal{A}}
 \newcommand{\D}{\mathcal{D}}
 \newcommand{\Y}{B} 

 \newcommand{\B}{\mathcal{B}}
 \renewcommand{\L}{\mathcal{L}}
 \renewcommand{\S}{\mathcal{S}}
 \newcommand{\E}{\mathcal{E}}
  \newcommand{\R}{\mathcal{R}}
 \newcommand{\I}{\mathcal{I}}

 \newcommand{\fixmeperso}[1]{}

 \usepackage{amsmath,amsfonts,amssymb,epsfig,color,xspace}

 \newcommand{\mbvfinale}[1]{}
 \definecolor{jacolor}{cmyk}{0.7, 0.1, 0.4, 0.1}

\newcommand{\bs}{{\sf bs}\xspace} 
\newcommand{\mbs}{{\sf mbs}\xspace} 
\newcommand{\is}{{\sf is}\xspace} 
\newcommand{\mis}{{\sf mis}\xspace} %

\newcommand{\MMBS}{{\sc MMBS}\xspace} 
\newcommand{\MMBSEG}{{\sc MMBS}$^{=}$\xspace} 
\newcommand{\MMBSINFEG}{{\sc MMBS}$^{\le}$\xspace} 
\newcommand{\MMVC}{{\sc MMVC}\xspace} 
\newcommand{\MMHS}{{\sc MMHS}\xspace} 
\newcommand{\SIMPLEEXTMMHS}{{\sc Simple-Ext-MMHS}\xspace} 
\newcommand{\EXTMMHS}{{\sc Ext-MMHS}\xspace} 
\newcommand{\MIS}{{\sc IS}\xspace} 
\newcommand{\VC}{{\sc VC}\xspace}
\newcommand{\MCIS}{{\sc CIS}\xspace} 
\newcommand{\UPDOM}{{\sc Up-Dom}\xspace} 

\newcommand{\VCdistoF}{{\sc VC}\xspace/{\sc dist}-{\sc to}-$\F$\xspace}
\newcommand{\LMMBS}{{\sc Maximum Minimal Blocking Set}\xspace} 
\newcommand{\LMMVC}{{\sc Maximum Minimal Vertex Cover}\xspace} 
\newcommand{\LMMHS}{{\sc Maximum Minimal Hitting Set}\xspace} 
\newcommand{\LEXTMMHS}{{\sc Extension-MMHS}\xspace} 
\newcommand{\LVC}{{\sc Vertex Cover}\xspace}
\newcommand{\LUPDOM}{{\sc Upper Dominating Set}\xspace} 

\newcommand{\tw}{{\sf tw}\xspace}

\newcommand{\mmbs}{{\sf mmbs}\xspace}
\newcommand{\mmhs}{{\sf mmhs}\xspace}

\newcommand{\FPT}{{\sf FPT}\xspace}
\newcommand{\NP}{{\sf NP}\xspace}
\newcommand{\RP}{{\sf RP}\xspace}
\newcommand{\NPh}{{\sf NP}-hard\xspace}
\newcommand{\XP}{{\sf XP}\xspace}

\newcommand{\Woneh}{{\sf W[1]}-hard\xspace}
\newcommand{\Wonehness}{{\sf W[1]}-hardness\xspace}
\newcommand{\Won}{{\sf W[1]}\xspace}
\newcommand{\ETH}{{\sf ETH}\xspace}
\newcommand{\ParaNPH}{para-{\sf NP}-hard\xspace}

\newcommand{\yes}{{\sf yes}\xspace}
\newcommand{\no}{{\sf no}\xspace}

\newcommand{\defproblem}[3]{\par
 \vspace{3mm}
\noindent\fbox{
 \begin{minipage}{0.96\textwidth}
 \begin{tabular*}{\textwidth}{@{\extracolsep{\fill}}lr} #1 &  \vspace{1mm} \\ \end{tabular*}
{\textbf{Input:}} #2
  \vspace{1mm}\\%
 {\textbf{Question:}} #3
 \end{minipage}
 }
 \vspace{3mm}\par
}

\renewenvironment{proof}[1][]{\par \noindent {\bf Proof:#1}\ }{\hfill$\Box$\medskip}

\newenvironment{proofclaim}[1][]{\par \noindent {\em Proof of the claim:#1}\ }{\hfill$\diamond$\bigskip}

\definecolor{dark-red}{rgb}{0.4,0.15,0.15}
\definecolor{dark-blue}{rgb}{0.15,0.15,0.4}
\definecolor{medium-blue}{rgb}{0,0,0.5}
\definecolor{gray}{rgb}{0.5,0.5,0.5}
\definecolor{color-Ig}{rgb}{0.15,0.7,0.15}

\hypersetup{
    colorlinks, linkcolor={dark-red},
    citecolor={dark-blue}, urlcolor={medium-blue}
}

\hypersetup{
colorlinks = true,
linkcolor = black!30!blue,
citecolor = black!30!green
}

\theoremstyle{plain}

\title{Parameterized complexity of computing maximum minimal blocking and hitting sets}

\author{J\'ulio Ara\'ujo}{Departamento de Matem\'atica, Universidade Federal do Cear\'a, Fortaleza, Brazil}{julio@mat.ufc.br}{https://orcid.org/0000-0001-7074-2753}{CNPq-Pq 304478/2018-0, CAPES-PrInt 88887.466468/2019-00 and CAPES-STIC-AmSud 88881.569474/2020-01.}

\author{Marin Bougeret}{LIRMM, Universit\'e de Montpellier, CNRS, Montpellier, France}{marin.bougeret@lirmm.fr}{https://orcid.org/0000-0002-9910-4656}{}

\author{Victor A. Campos}{Departamento de Computa\c c\~ao, Universidade Federal do Cear\'a, Fortaleza, Brazil}{victoitor@ufc.br}{https://orcid.org/0000-0002-2730-4640}{FUNCAP - PNE-011200061.01.00/16.}

\author{Ignasi Sau}{LIRMM, Universit\'e de Montpellier, CNRS, Montpellier, France}{ignasi.sau@lirmm.fr}{https://orcid.org/0000-0002-8981-9287}{DEMOGRAPH (ANR-16-CE40-0028), ESIGMA (ANR-17-CE23-0010) and ELIT (ANR-20-CE48-0008-01).}

\authorrunning{J\'ulio Ara\'ujo, Marin Bougeret, Victor A. Campos, and Ignasi Sau} 

\Copyright{J\'ulio Ara\'ujo, Marin Bougeret, Victor A. Campos, and Ignasi Sau} 

\ccsdesc{Mathematics of computing~Graph algorithms.}

\keywords{maximum minimal blocking set, maximum minimal hitting set, parameterized complexity, treewidth, kernelization, vertex cover, upper domination.\vspace{-.1cm}}

\category{} 

\relatedversion{} 

\supplement{}

\acknowledgements{}

\nolinenumbers 

\hideLIPIcs  

\EventNoEds{2}
\EventLongTitle{}
\EventShortTitle{arXiv preprint}
\EventAcronym{}
\EventYear{2020}
\EventDate{August 24--28, 2020}
\EventLocation{Prague, Czech Republic}
\EventLogo{}
\SeriesVolume{}
\ArticleNo{1}

\begin{document}

\maketitle

\begin{abstract}
A \emph{blocking set} in  a graph $G$ is a subset of vertices that intersects every maximum independent set of $G$.
Let $\mmbs(G)$ be the size of a maximum (inclusion-wise) minimal blocking set of $G$.
  This parameter has recently played an important role in the kernelization of \LVC parameterized by the distance to a graph class $\F$.
  Indeed, it turns out that the existence of a polynomial kernel for this problem is closely related to the property that $\mmbs(\F)=\sup_{G \in \F}\mmbs(G)$ is bounded by a constant,
  and thus several recent results focused on determining $\mmbs(\F)$ for different classes~$\F$. We consider the parameterized complexity of computing $\mmbs$ under various parameterizations, such as the size of a maximum independent set of the input graph and the natural parameter. We provide
  a panorama of the complexity of computing both $\mmbs$ and $\mmhs$, which is the size of a maximum minimal hitting set of a hypergraph, a closely related parameter.
  Finally, we consider the problem of computing $\mmbs$ parameterized by treewidth, especially relevant in the context of kernelization.
  Given the ``counting'' nature of $\mmbs$, it does not seem to be expressible in monadic second-order logic, hence its tractability does not follow from Courcelle's theorem. Our main technical contribution is a fixed-parameter tractable algorithm for this problem.

\end{abstract}
\section{Introduction}


%
%

Given a graph $G$, we denote by $\alpha(G)$ the maximum size of an \emph{independent set} of $G$, that is, of a set of pairwise non-adjacent vertices. 
For the sake of conciseness, we abbreviate ``independent set''  as \is, and ``maximum independent set'' as \mis.
 A set $\Y \subseteq V(G)$ is a \emph{blocking set}, abbreviated as \bs, of $G$ if  $\alpha(G \setminus \Y) < \alpha(G)$, where $G \setminus \Y = G[V(G) \setminus \Y]$. 
Equivalently, $\Y$ is a blocking set of $G$ if for every \mis $I^* \subseteq V(G)$, $I^* \cap \Y \neq \emptyset$.
In this work we are interested in (inclusion-wise) minimal blocking sets, which we abbreviate as \mbs.
We denote by $\mmbs(G)$ the maximum size of an \mbs of $G$, and by \LMMBS (\MMBS for short) the problem where, given a graph $G$ and an integer $\beta$, the objective is to decide whether $\mmbs(G) \ge \beta$.
The main objective of this paper is to study the parameterized complexity of \MMBS. As discussed below, this problem is strongly related to the \LMMHS (\MMHS) problem, for which we also present several results.



\bigskip
\noindent\textbf{Role of maximum minimal blocking sets in kernelization.}
Given a graph $G$, a set of vertices $S \subseteq V(G)$ is a \emph{vertex cover} if it contains at least one endpoint of every edge.
The \LVC (\VC for short) problem consists in, given a graph $G$ and an integer $k$, decide if there is a vertex cover $S$ of $G$ such that $|S| \le k$. 
For a fixed graph class $\F$, the \VCdistoF parameterized problem is defined as follows. The input is a triple $(G,X,k)$ where $G$ is a graph, $X \subseteq V(G)$, and $G \setminus X$
belongs to $\F$. Set $X$ is often referred to as a \emph{modulator} to $\F$, and $|X|$ as the \emph{distance} of $G$ to $\F$.
The objective of the problem is to decide whether $G$ admits a vertex cover of size at most $k$, and the parameter is $|X|$.
A \emph{kernel} of vertex size $f$ for this problem is a polynomial-time algorithm that, given an input $(G,X,k)$, outputs an equivalent instance $(G',X',k')$ with $|V(G')| \le f(|X|)$.
Informally, such a kernel compresses the input graph $G$ to a smaller graph $G'$ whose size is bounded by a function $f$ depending only on $|X|$. If $f$ is a polynomial (resp. linear) function, we speak of a \emph{polynomial} (resp. \emph{linear}) kernel. We refer the reader to Section~\ref{sec:prelim} for formal definitions.
The \VCdistoF problem has been defined by Jansen and Bodlaender~\cite{JansenB13} for $\F$ being the class of forests as a way to improve the linear kernel for  {\sc Vertex Cover} parameterized by the standard parameter $k$, and the main result of~\cite{JansenB13} is a polynomial kernel for \VCdistoF (for $\F$ being the forests).

This result triggered a long line of follow-up research,
which aimed to find the most general graph families $\F$ such that \VCdistoF admits a polynomial kernelization~\cite{FellowsJKRW18}.
Several results where proved for specific families $\F$ such as those of degree at most two, of bounded treedepth, pseudo-forests (see~\cite{FellowsJKRW18,DBLP:conf/icalp/BougeretJS20} for a complete
list of references), and a major open question in this area
is to find a characterization of the families $\F$ for which \VCdistoF\ admits a polynomial kernel~\cite{DBLP:conf/icalp/BougeretJS20}. This is where parameter $\beta$ comes into play.

Kernelization algorithms for \VCdistoF\ usually proceed in two steps. In step~1, they reduce the number of connected components of $G \setminus X$ to a polynomial in $|X|$,
and in  step~2 they reduce the size of each connected component of $G \setminus X$ to a polynomial in $|X|$ as well.
Minimal blocking sets have been introduced in the seminal paper of Jansen and Bodlaender~\cite{JansenB13} for the case of $\F$ being the class of forests as a handy tool to achieve step~1. After that, this notion
has been generalized and reused for example in~\cite{DBLP:journals/algorithmica/BougeretS19,DBLP:conf/stacs/HolsKP20,DBLP:conf/icalp/BougeretJS20}, finally leading to the following black box tool for step~$1$, where $\mmbs(\F)=\sup_{G \in \F}\mmbs(G)$.
\begin{theorem}[Hols et al.~\cite{DBLP:conf/stacs/HolsKP20}]\label{thm:hols2}
Let $\F$ be a hereditary graph class on which \VC can be solved in polynomial time.
There is a polynomial-time algorithm that, given an instance $(G,X,k)$ of \VCdistoF, returns an equivalent instance $(G_0,X,k_0)$
of \VCdistoF such that $G_0 \setminus X \in \F$ and has at most $\O(|X|^{\mmbs(\F)})$ connected components.
\end{theorem}
Informally,  Theorem~\ref{thm:hols2} states that, if $\mmbs(\F)$ is bounded, then ``half'' of the kernelization algorithm can be done automatically. 
Moreover, it has been shown that $\mmbs(\F)$ being bounded by a constant is necessary in order to obtain a polynomial kernel:
\begin{theorem}[Hols et al.~\cite{DBLP:conf/stacs/HolsKP20}]\label{thm:hols}
Unless \NP $\subseteq$ {\sf coNP/poly}, \VCdistoF does not admit a kernel of size $\O(|X|^{\mmbs(\F)-\varepsilon})$ for any $\varepsilon > 0$.
 \end{theorem}

These two theorems might suggest that $\mmbs$ might be the right candidate to characterize graph classes $\F$ for which \VCdistoF admits a polynomial kernel.
However, it turns out that there exists a class $\F$ where $\mmbs(\F)$ is bounded but for which there is no polynomial kernel for \VCdistoF under standard complexity assumptions~\cite{DBLP:conf/stacs/HolsKP20}.
Nevertheless, for minor-closed families\footnote{A graph class is \emph{minor-closed} if any minor of a member in the class also belongs to it, and a graph $H$ is a \emph{minor} of a graph $G$ is $H$ can be obtained from a subgraph of $G$ by contracting edges.}, the following theorem shows that $\mmbs$ is indeed the correct parameter in order to characterize the existence of polynomial kernels for \VCdistoF.
\begin{theorem}[Bougeret et al.~\cite{DBLP:conf/icalp/BougeretJS20}]
If $\F$ is a minor-closed graph class, then \VCdistoF admits a polynomial kernel if and only if $\mmbs(\F)$ is bounded by a constant.
\end{theorem}

To summarize this discussion, for general graph classes $\F$, having bounded $\mmbs(\F)$ is necessary but not sufficient, although having bounded $\mmbs(\F)$ yields ``half'' of the kernel. For minor-closed classes $\F$, having bounded $\mmbs(\F)$ is indeed the correct characterization.
These results explain the recent interest in computing $\mmbs(\F)$ for different classes $\F$~\cite{DBLP:conf/stacs/HolsKP20}, and thus our motivation
to study the complexity of the \MMBS problem.

Let us also mention that computing $\mmbs$ can in addition be useful when implementing any of the previously mentioned kernels.
Indeed, given an instance $(G,X,k)$ of \VCdistoF with $|V(G)|=n$, the algorithm behind Theorem~\ref{thm:hols2}
takes as additional input the value $\mmbs(\F)$ and outputs
the claimed equivalent instance in time $n^{\mmbs(\F) +   \O(1) }$.
However, when implementing this algorithm, we can rather first compute $\mmbs(G\setminus X)$, and use the algorithm
of Theorem~\ref{thm:hols2} with additional input $\mmbs(G\setminus X)$ instead of $\mmbs(\F)$ (note that $\mmbs(G\setminus X) \le \mmbs(\F)$, potentially much smaller),
and thus obtain a running time $n^{\mmbs(G \setminus X) +   \O(1) }$, and an equivalent instance $(G_0,X,k_0)$, where $G_0 \setminus  X \in \F$, with at most $\O(|X|^{\mmbs(G\setminus X)})$ connected components.

 
\medskip
\noindent\textbf{Contribution and related work.} In what follows we present our contribution and relate it to previous work, by considering each parameterization of the considered problems separately. 

\medskip
\noindent\emph{Choice of the parameters.}
As the \VCdistoF problem has only been considered for graph classes $\F$ where {\sc Maximum Independent Set} (\MIS for short) can be solved in polynomial time,
we also incorporate  this assumption in this work, hence motivating the parameterization of \MMBS by combinations of
 $\alpha$ (i.e., the  size of a \mis of the input graph) and the threshold $\beta$ (the natural parameter).
Moreover, since in all the previously mentioned cases where \VCdistoF has a polynomial kernel~\cite{FellowsJKRW18,DBLP:conf/icalp/BougeretJS20} the graphs in the class $\F$ have bounded \emph{treewidth}, we also consider the \MMBS problem parameterized by the treewidth of the input graph $G$.

\medskip
\noindent\emph{Problems related to computing $\mmbs$.} 
We denote by \LMMHS (\MMHS for short) the problem where, given a hypergraph $\H$ and an integer $\beta$, the objective
is to decide whether $\mmhs(\H) \ge \beta$, where $\mmhs(\H)$ is the size of a largest  minimal hitting set of $\H$, that is, an inclusion-wise minimal set of vertices of $\H$ containing at least one vertex of every hyperedge. A \emph{dominating set} in a graph $G$ is a subset of vertices $S \subseteq V(G)$ such that every vertex in $V(G) \setminus S$ has a neighbor in $S$.
We denote by \LUPDOM (\UPDOM for short) the problem of computing a maximum inclusion-wise minimal dominating set in an input graph $G$.
As pointed out by Bazgan et al.~\cite{BAZGAN20182}, \UPDOM is a special case of \MMHS, as we can create a hyperedge for each closed neighborhood of the vertices in $G$, implying that the negative results
stated below for \UPDOM transfer directly to \MMHS.
As mentioned before, we only consider graph classes $\F$ where \MIS can be solved in polynomial time. A natural special case of such classes is when $\alpha$ is constant,
and in this case \MMBS reduces to \MMHS by simply generating in time $n^{\alpha + \O(1)}$ a hyperedge for each \mis of $G$.
Let us now define parameterizations for \MMBS and \MMHS. A parameterization is, in a nutshell, a function mapping instances of a parameterized problem to non-negative integers (see Section~\ref{sec:prelim}
 for the details).
For \MMBS, we define, with slight abuse of notation, the parameters $\alpha$ as $\alpha(G,\beta)=\alpha(G)$ and $\beta$ as $\beta(G,\beta)=\beta$.
Similarly, for \MMHS, we define the parameters $\alpha$ as $\alpha(\H,\beta)=\max_{H \in E(\H)}|H|$ and $\beta$ as $\beta(\H,\beta)=\beta$.

\medskip
The first objective of this paper is to obtain a complete landscape of the parameterized complexity of \MMBS and \MMHS under different combinations of $\alpha$ and $\beta$ to compare
the behavior of these two problems.
To the best of our knowledge, the parameterized complexity of \MMBS has not been considered before in the literature.
On the other hand, the \MMHS problem has received considerable attention, especially concerning the enumeration of minimal hitting sets, as it allows to construct
the so-called \emph{dual} hypergraph, that is, the hypergraph having a hyperedge for every minimal hitting set of the original hypergraph.
Let us now present our contributions together with the related work about the parameterized complexity of \MMHS, for each different parameterization that we consider.

\medskip
\noindent\emph{Parameterization by $\alpha$ or $\beta$ separately.}
When parameterizing by $\alpha$ only, both \MMBS and \MMHS are \ParaNPH, meaning \NP-hard for fixed values of the parameter. Indeed, the particular case $\alpha=2$ of \MMHS corresponds to the \LMMVC (\MMVC for short) problem, which is \NPh~\cite{BORIA201562}.
When parameterizing by $\beta$ only, we show in Proposition~\ref{prop:hardness} that \MMBS is \ParaNPH, whereas \MMHS is \Woneh~\cite{BAZGAN20182} and \XP~\cite{blasius2019efficiently}.
As discussed in Section~\ref{section:alphabeta}, the \Wonehness proof of~\cite{BAZGAN20182} implies that, unless the Exponential Time Hypothesis (\ETH for short) fails, \UPDOM cannot be solved in time $f(k) \cdot n^{o(\sqrt{k})}$ on $n$-vertex graphs for any computable function $f$.
We improve this lower bound by showing in Theorem~\ref{thm:updom} that \UPDOM cannot be solved in time $f(k) \cdot n^{o(k)}$ for any computable function $f$, implying the same result
(replacing $k$ by $\beta$) for \MMHS.
We point out that very recently and independently from our work, this improved lower bound for \UPDOM  has also been proved by Dublois et al.~\cite{dublois2021upper}.

\medskip
\noindent\emph{Parameterization by one parameter while fixing the other.}
When fixing $\alpha$ and parameterizing by $\beta$, \MMBS reduces to \MMHS, which was known to be \FPT~\cite{DBLP:journals/disopt/Damaschke11},
and for which we provide in Proposition~\ref{prop:fixedalpha_beta_kernel} a polynomial kernel with $\O(\beta^\alpha)$ vertices,
generalizing the known quadratic kernel for \MMVC~\cite{FernauHDR,BORIA201562}.
When fixing $\beta$ and parameterizing by $\alpha$, we show in Proposition~\ref{prop:hardness} that \MMBS is \Woneh, whereas \MMHS is \FPT (at it is even \FPT parameterized by the sum as explained
in the next paragraph).

\medskip
\noindent\emph{Parameterization by the sum.}
Finally, when parameterizing by $\alpha+\beta$, the hardness result given in Proposition~\ref{prop:hardness} (i.e., parameterizing by $\alpha$ for fixed $\beta$) implies that \MMBS is \Woneh, whereas \MMHS is \FPT
for the following reasons.
We first provide in Corollary~\ref{cor:mmhs_xp_fpt} a simple \FPT algorithm for \MMHS that reduces to an extension problem considered by Bläsius et al.~\cite{blasius2019efficiently}, and then design in Theorem~\ref{thm:fptv2} a more involved
an ad-hoc algorithm to improve the running time to $\O^*(2^{\alpha \beta})$, where the $\O^*$-notation hides multiplicative polynomial terms (see Section~\ref{sec:prelim}).  


\medskip
Our results considering parameters $\alpha$  and $\beta$ are summarized in Table~\ref{fig:array}, where $\Pi/\kappa$ denotes problem $\Pi$ parameterized by function $\kappa$ (see again Section~\ref{sec:prelim}).

\begin{figure}
\begin{tabular}{p{0.2\textwidth}|p{0.3\textwidth}|p{0.4\textwidth}}
Parameter & \MMBS & \MMHS \\
\hline
$\beta$ &  \ParaNPH (Prop.~\ref{prop:hardness}) & \XP (\!\!\cite{blasius2019efficiently}, Corollary~\ref{cor:mmhs_xp_fpt}) \newline \Woneh~\cite{BAZGAN20182} \newline  $\nexists$ $f(\beta)\cdot (|V(\H)|+|E(\H)|)^{o(\beta)}$ (Cor.~\ref{cor:mmhs}) \\
\hline
$\alpha$ &  \ParaNPH (Prop.~\ref{prop:hardness}) &  \ParaNPH (\MMVC for $\alpha=2$) \\
\hline
$\beta$ ($\alpha$ fixed) & Reducible to \MMHS  in time $n^{\alpha + \O(1)}$ & Kernel with $\O(\beta^{\alpha})$ vertices (Prop.~\ref{prop:fixedalpha_beta_kernel}) \newline \FPT in $\O^*(\alpha^{\beta})$~\cite{DBLP:journals/disopt/Damaschke11} ($\O^*$ hides $n^{f(\alpha)}$) \\
\hline
$\alpha$ ($\beta$ fixed) & \Woneh (Prop.~\ref{prop:hardness}) \newline \XP (trivial) & \FPT (as \FPT by $\alpha+\beta$) \\
\hline
$\alpha+\beta$ & \Woneh (Prop.~\ref{prop:hardness}) &  \FPT (Cor.~\ref{cor:mmhs_xp_fpt} and Thm.~\ref{thm:fptv2}) \\
\hline
\end{tabular}
\caption{Parameterized complexity of \MMBS and \MMHS with parameters $\alpha$ and $\beta$. The main differences between \MMBS and \MMHS show up in the parameterizations by $\beta$ and $\alpha+\beta$.}
\label{fig:array}
\end{figure}



\medskip
\noindent\emph{Parameterization by treewidth.}
Let us now turn to the second objective of this paper, which is the parameterization by treewidth.
It is known that both \UPDOM~\cite{BAZGAN20182} and \MMVC~\cite{BORIA201562} are \FPT parameterized by treewidth, but none of these results implies the same result for \MMBS.
We also mention that that the problem of finding a maximum minimal set intersecting all maximum cliques of a graph is \FPT parameterized by treewidth~\cite{lee2011weighted},
implying that \MMBS is \FPT parameterized by treewidth of the complement of the input graph.
We prove in Theorem~\ref{thm:tw} that \MMBS/\tw is \FPT, which is the main technical result of this paper. Let us mention that \MMBS does not seem to be (at least, easily) expressible in monadic second-order logic, due to the fact that the a blocking set in a graph $G$ is defined so that it intersects every {\sl maximum-sized} independent set of $G$, and properties involving counting the {\sl sizes} of arbitrarily large sets are typically non-expressible in monadic second-order logic~\cite{BonnetEPT15,daglib/0030804}. This is why, in order to deduce that \MMBS/\tw is \FPT, we cannot directly apply Courcelle's theorem~\cite{Courcelle90}, and we need to design an ad-hoc algorithm that is quite involved, needing a number of technical lemmas. In \label{subsec:prelim_tw} we discuss the list of parameters used in the tables of our dynamic programming algorithm, along with several examples to illustrate the difficulties that motivate their choice.


\medskip
\noindent\textbf{Organization.}
In Section~\ref{sec:prelim} we provide preliminaries about graphs, parameterized complexity, the formal statement of the considered problems, and treewidth, and we state in Lemma~\ref{lem:prelim} several useful properties of minimal blocking and hitting sets.
Section~\ref{section:alphabeta} is devoted to parameterizations by $\alpha$ and $\beta$, and Section~\ref{sec:tw} to the algorithm for \MMBS parameterized  by treewidth. We conclude the article in Section~\ref{section:conclusion}
 with some directions for further research.

\section{Preliminaries}
\label{sec:prelim}

\noindent\textbf{Graphs and functions.} We only provide here basic definitions and refer the reader to~\cite{Die10} for any missing definitions about graphs. We only consider finite simple graphs with no loops nor multiple edges. 
For a graph $G$ and a vertex $v \in V(G)$, we denote by $N_G(v)$ the set of vertices of $G$ adjacent to $v$ and, for a subset $S \subseteq V(G)$, we let $N_S(v) = N_G(v) \cap S$. When the graph $G$ is clear from the context, we may omit the subscript.
Given $B \subseteq V(G)$, we denote $G \setminus B = G[V(G)\setminus B]$ where $G[X]$ denotes the graph induced by $X \subseteq V(G)$.
We denote a \emph{triangle}, that is, a complete graph on three vertices, on vertices $u,v,w$ by $(u,v,w)$.
Given a graph $G$, we say that $X \subseteq V(G)$ is a \emph{vertex cover} if for any edge $e \in E(G)$, $e \cap X \neq \emptyset$,
and that $X$ is a \emph{dominating set} if for every $v \in V(G)\setminus X$, there exists $u \in X$ such that $\{u,v\} \in E$.
Given a hypergraph $\H$, we say that $I \subseteq V(\H)$ is an \emph{independent set} if for every $H \in E(\H)$, $H \nsubseteq I$,
and that $X \subseteq V(\H)$ is a \emph{hitting set} if for every $S \in E(\H)$, $S \cap X \neq \emptyset$. A graph class is \emph{hereditary} if it is closed under induced subgraphs.

If a set $A$ is partitioned into pairwise disjoint subsets $A_1, \ldots, A_k$, we denote it by $A = A_1 \uplus \dots \uplus A_k$.
If $A$ is a set, we denote by $2^A$ the collection containing all the subsets of $A$.
Given a function $f: A \to B$ and a subset $A' \subseteq A$, we denote by $f_{|A'}$ the restriction of $f$ to $A'$.
For a positive integer $k$, we let $[k]$ be the set containing every integer $i$ such that $1 \leq i \leq k$.

For a function $f$ mapping graphs to integers (such as the parameterizations $\alpha$ and $\beta$ discussed before) and a class of graphs $\F$, we define $f(\F)=\sup_{G \in \F}f(G)$.
Given two functions $f_1$ and $f_2$, mapping instances $I$ of a problem to $\mathbb{N}$), if there exists a polynomial $p$ such that
for every instance $I$, $f_1(I) \le f_2(I) \cdot p(|I|)$, where $|I|$ is the size of $I$, then we write $f_1 = \O^*(f_2)$.

\medskip
\noindent
\textbf{Parameterized complexity.} We refer the reader to~\cite{DF13,CyganFKLMPPS15} for basic background on parameterized complexity, and we recall here only some basic definitions. A \emph{parameterized problem} is a language $L \subseteq \Sigma^* \times \mathbb{N}$, where $\Sigma$ is some fixed alphabet.  For an instance $I=(x,k) \in \Sigma^* \times \mathbb{N}$, $k$ is called the \emph{parameter}.
Given a classical (non-parameterized) decision problem $L_{c} \subseteq \Sigma^*$ and a function $\kappa: \Sigma^* \rightarrow \mathbb{N}$, we denote by
$L_{c}/\kappa = \{(x,\kappa(x)\} \mid x \in L_{c}\}$ the associated parameterized problem.

A parameterized problem $L$ is \emph{fixed-parameter tractable} (\textsf{FPT}) if there exists an algorithm $\mathcal{A}$, a computable function $f$, and a constant $c$ such that given an instance $I=(x,k)$, $\mathcal{A}$   (called an \textsf{FPT} \emph{algorithm}) correctly decides whether $I \in L$ in time bounded by $f(k) \cdot |I|^c$. For instance, the \textsc{Vertex Cover} problem parameterized by the size of the solution is \textsf{FPT}.
	
A parameterized problem $L$ is \textsf{XP} if there exists an algorithm $\mathcal{A}$ (called an \emph{\textsf{XP} algorithm}) and two computable functions $f$ and $g$ such that given an instance $I=(x,k)$, $\mathcal{A}$  (called an \textsf{XP} \emph{algorithm}) correctly decides whether $I \in L$ in time bounded by $f(k) \cdot |I|^{g(k)}$. For instance,  the \textsc{Independent Set} problem parameterized by the size of the solution is \textsf{XP}.

Within parameterized problems, the \textsf{W}-hierarchy may be seen as the parameterized equivalent to the class \textsf{NP} of classical decision problems. Without entering into details (see~\cite{DF13,CyganFKLMPPS15} for the formal definitions), a parameterized problem being \textsf{W}[1]-\emph{hard} can be seen as a strong evidence that this problem is \textsl{not} \textsf{FPT}.
The canonical example of \textsf{W}[1]-hard problem is \textsc{Independent Set}  parameterized by the size of the solution.

The most common way to transfer \Wonehness is via parameterized reductions.
A \emph{parameterized reduction} from a parameterized problem $L_1$ to a parameterized problem $L_2$ is an algorithm that, given an instance $(x,k)$ of $L_1$, outputs an instance $(x',k')$ of $L_2$
such that
\begin{itemize}
\item $(x,k)$ is a yes-instance of $L_1$ if and only if $(x',k')$ is a yes-instance of $L_2$,
\item $k' \le g(k)$ for some computable function $g$, and
\item the running time is bounded by $f(k)\cdot|x|^{\O(1)}$ for some computable function $f$
\end{itemize}
If $L_1$ is \textsf{W}[1]-hard and  there is a parameterized reduction from $L_1$ to $L_2$, then $L_2$ is \textsf{W}[1]-hard as well. 

A parameterized problem is  \emph{para-\NP-hard} if it is \NPh for a fixed value of the parameter, implying in particular that the problem cannot be in \XP unless $\P=\NP$.

A \emph{kernelization algorithm} for a parameterized problem $L$ is an algorithm $\mathcal{A}$ that, given an instance $(x, k)$ of $L$, generates in polynomial time an equivalent instance $(x', k')$ of $Q$ such that $|x'| + k' \leq f(k)$, for some computable function $f : \mathbb{N} \rightarrow \mathbb{N}$. If $f(k)$ is bounded from above by a polynomial function, we say that $L$ admits a \emph{polynomial} kernel. In particular, if $f(k)$ is bounded by a linear (resp. quadratic) function, then we say that $L$ admits a \emph{linear} (resp. \emph{quadratic}) kernel.

The Exponential Time Hypothesis (\ETH for short) of Impagliazzo et al.~\cite{ImpagliazzoPZ01} is a complexity assumption implying that the $3$-SAT problem cannot be solved in time $2^{o(n)}$ restricted to formulas with $n$ variables.

\medskip
\noindent\textbf{List of considered problems.}
We denote by \MIS the {\sc Maximum Independent Set} problem where, given a graph $G$ and an integer $k$, the objective is to decide whether $\alpha(G) \ge k$,
and by \MCIS the {\sc Multicolored Independent Set} problem, where given graph $G$ and an integer $k$ such that $V(G)$ is partitioned into $k$ cliques $\{V_i \mid i \in [k]\}$,
the goal is to decide whether $\alpha(G) \ge k$.

In the two following problems, recall that $\mmbs(G)$ (resp. $\mmhs(\H)$) denotes the size of a largest minimal blocking set of $G$ (resp. largest minimal hitting set of $\H$).
\defproblem
{\LMMBS (\MMBS)}
{A graph $G$ and a positive integer $\beta$.}
{$\mmbs(G) \ge \beta$?}

\defproblem
{\LMMHS (\MMHS)}
{A hypergraph $\H$ and a positive integer $\beta$.}
{$\mmhs(\H) \ge \beta$?}

The problem {\sc Maximum Minimal Vertex Cover} (\MMVC) corresponds to the restriction  of \MMHS to instances where all hyperedges have size two, i.e., graphs.
For a fixed positive integer $\alpha$, we also define $\alpha$-\MMHS as the \MMHS problem restricted to instances whose hypergraph $\H$ is such that $|H| \le \alpha$ for every $H \in E(\H)$,
and $\alpha$-\MMBS as the \MMBS problem restricted to instances whose graph $G$ is such that $\alpha(G) \le \alpha$.
Notice that in any \FPT or kernel algorithm for $\alpha$-\MMHS or for $\alpha$-\MMBS, as $\alpha$ is fixed, the running time given using the $\O^*$-notation might typically hide a term $n^{f(\alpha)}$, where $n$ is the number of vertices of the graph or hypergraph under consideration.
Finally, we define \MMBSEG (resp. \MMBSINFEG) as the \MMBS problem where the objective is to decide whether $\mmbs(G)=\beta$ (resp. $\mmbs(G) \le \beta$).

\defproblem
{\LEXTMMHS (\EXTMMHS)}
{A hypergraph $\H$ and a two subsets $X,Y \subseteq V(\H)$ such that $X \cap Y = \emptyset$.}
{Does there exist a minimal hitting set $S$ of $\H$ such that $X \subseteq S \subseteq V(\H) \setminus Y$?}

Problem \EXTMMHS was defined by Bläsius et al.~\cite{blasius2019efficiently}. We also define \SIMPLEEXTMMHS as the special case of the \EXTMMHS where $Y = \emptyset$. 

The last problem we define here is the ``max-min'' version of \textsc{Dominating Set}.

\defproblem
{\LUPDOM (\UPDOM)}
{A graph $G$ and an integer $k$.}
{Does $G$ contain a minimal dominating set of size at least $k$?}

\noindent\textbf{Tree decompositions and treewidth.}\label{def:tw} A \emph{tree decomposition} of a graph $G$ is a pair ${\cal D}=(T,{\cal B})$, where $T$ is a tree
and ${\cal B}=\{X^{w}\mid w\in V(T)\}$ is a collection of subsets of $V(G)$, called \emph{bags},
such that:
\begin{itemize}
\item $\bigcup_{w \in V(T)} X^w = V(G)$,
\item for every edge $\{u,v\} \in E$, there is a $w \in V(T)$ such that $\{u, v\} \subseteq X^w$, and
\item for every $\{x,y,z\} \subseteq V(T)$ such that $z$ lies on the unique path between $x$ and $y$ in $T$,  $X^x \cap X^y \subseteq X^z$.
\end{itemize}
We call the vertices of $T$ {\em nodes} of ${\cal D}$ and the sets in ${\cal B}$ {\em bags} of ${\cal D}$. The \emph{width} of a  tree decomposition ${\cal D}=(T,{\cal B})$ is $\max_{w \in V(T)} |X^w| - 1$.
The \emph{treewidth} of a graph $G$, denoted by $\tw(G)$, is the smallest integer $t$ such that there exists a tree decomposition of $G$ of width at most $t$.
We need to introduce nice tree decompositions, which will make the presentation of the algorithm of Section~\ref{sec:tw} much simpler.

\medskip
\noindent
\textbf{Nice tree decompositions.} \label{def:nice-tw}Let ${\cal D}=(T,{\cal B})$
be a rooted tree decomposition of $G$ (meaning that $T$ has a special vertex $r$ called the \emph{root}).
As $T$ is rooted, we naturally define an ancestor relation among bags, and say that $X^{w'}$ is a \emph{descendant} of $X^w$ if the vertex set of the unique simple path in $T$ from $r$ to $w'$ contains
$w$. In particular, every node $w$ is a descendant of itself.
For every $w \in V(T)$, we define $G_{X^w}=G[\bigcup\{X^{w'}\mid X^{w'}\text{ is a descendant of }X^w\text{ in }T\}]$.

Such a rooted decomposition is called a \emph{nice tree decomposition} of $G$ if the following conditions hold:
\begin{itemize}

\item $X^{r} = \emptyset$,
\item every node of $T$ has at most two children in $T$,
\item for every leaf $\ell \in V(T)$, $X^{\ell} = \emptyset$. Each such a node $\ell$ is called a {\em leaf node},
\item if $w \in V(T)$ has exactly one child $w'$, then either
  \begin{itemize}
  \item $X^w = X^{w'}\cup \{v\}$ for some $v \not \in X^{w'}$.
    Each such a node is called an \emph{introduce node},
  \item $X^w = X^{w'} \setminus \{v\}$ for some $v \in X^{w'}$.
    Each uch a node is called a \emph{forget node}, and
  \end{itemize}
\item if $w \in V(T)$ has exactly two children $w_L$ and $w_R$, then $X^{w} = X^{w_L} = X^{w_R}$. 
Notice that there is no edge in $G_{X^w}$ between $V(G_{X^{w_L}})\setminus X^w$ and $V(G_{X^{w_R}})\setminus X^w$. Each such a node is called a \emph{join node}.
\end{itemize}

Given a tree decomposition of a graph $G$, it is possible to transform it in polynomial time into a {\sl nice} one of the same width~\cite{Klo94}.

For the sake of simplicity of the (already quite heavy) notation used in the dynamic programming algorithm of Section~\ref{sec:tw}, we will drop the vertices of $V(T)$ from the notation of bags defined above. Therefore,
in the case of an introduce or forget node, the bag $X^w$ and its child $X^{w'}$ will be denoted $X$ and $X^C$, respectively,
and in the case of a join node, the bag $X^w$ and its children $X^{w_L}$ and $X^{w_R}$ will be denoted $X,X^L,$ and $X^R$ respectively.

\medskip
\medskip
\noindent
\textbf{Basic properties.} We now state some basic properties of the considered problems that will be used later.
The ones concerning minimal blocking sets have been already (explicitly or implicitly) observed  in~\cite{DBLP:conf/stacs/HolsKP20}, 
but for the sake of completeness we prove all of them here.

\begin{lemma}\label{lem:prelim}
The following properties hold.
\begin{enumerate}
\item\label{prop1} For every graph $G$, $\Y \subseteq V(G)$ is an \mbs of $G$ if and only if $\Y$ is a \bs of $G$ and, for every $v \in \Y$, there is a \mis $I_v$ of $G$ such that $I_v \cap \Y = \{v\}$.
\item\label{prop1bis} For every hypergraph $\H$, $\Y \subseteq V(\H)$ is an minimal hitting set of $\H$ if and only if $\Y$ is a hitting set of $\H$ and, for every $v \in \Y$, there is a hyperedge $H_v$ of $\H$ such that $H_v \cap \Y = \{v\}$.
\item\label{prop2} For every graph $G$, there exists a unique \mis in $G$ if and only if $\mmbs(G)=1$. 
\item\label{prop3} Given an instance  $(\H,\beta)$  of \MMHS, we can obtain in polynomial time an equivalent instance $(\H',\beta)$
such that $V(\H') = V(\H)$, $E(\H') \subseteq E(\H)$, and no hyperedge of $\H'$ is contained in another hyperedge of $\H'$.
\item\label{prop4} Given an instance $(\H,\beta)$  of \MMHS, we can obtain in polynomial time an equivalent instance $(\H',\beta)$
such that $V(\H') \subseteq V(\H)$, and every vertex of $\H'$ belongs to at least one hyperedge of $\H'$.
\end{enumerate}
\end{lemma}

\begin{proof}
\emph{Property~\ref{prop1}}.
For the forward implication, consider an \mbs $\Y$ of $G$. As $\Y$ is minimal, for every $v \in \Y$, there exists a \mis $I_v$ such that $I_v \cap \{\Y \setminus \{v\}\} = \emptyset$.
As $\Y$ is a \bs, $I_v \cap \Y \neq \emptyset$, implying $I_v \cap \Y = \{v\}$.
The backward implication is immediate. The proof of Property~\ref{prop1bis} is almost the same.

\emph{Property~\ref{prop2}}. For the forward implication, let $I$ be the unique \mis of $G$, $\Y$ be a \bs of $G$, and $v \in \Y \cap I$. If $|\Y| \ge 2$ then $\Y$ is not minimal as $\{v\}$ is still a \bs.
Let us now prove the contrapositive of the backward implication. Suppose that $G$ contains two distinct \mis $I_1$ and $I_2$.
Let us consider the reduction mapping $G$ to a hypergraph $\H$ on the same vertex set, and having one hyperedge for each \mis of $G$.
Let $H_1$ and $H_2$ be the hyperedges corresponding to $I_1$ and $I_2$, respectively. Observe that $\{H_1,H_2\}$ is a sunflower (see Definition~\ref{def:sunflower} in page~\pageref{def:sunflower}).
By Lemma~\ref{lemma:sunflower} (page~\pageref{lemma:sunflower}), we get that $\mmhs(\H) \ge 2$, and as $\mmhs(\H) = \mmbs(G)$, we obtain the desired result.

\emph{Property~\ref{prop3}}.
Suppose that there exist two distinct hyperedges $H_1$ and $H_2$ of $\H$ such that $H_1 \subsetneq H_2$.
Let $\tilde{\H}$ be the hypergraph obtained from $\H$ by removing $H_2$.
Let us prove that for every $\Y \subseteq V(\H)$, $\Y$ is a minimal hitting set in $\H$ if and only if $\Y$ is a minimal hitting set in $\tilde{\H}$.
The equivalence of being a hitting set is immediate, and the minimality in $\tilde{\H}$ directly implies the minimality in $\H$.
Finally, suppose that $\Y$ if minimal in $\H$. Consider an arbitrary vertex $v \in \Y$ and let $H_v \in E(\H)$ such that $\Y \cap H_v = \{v\}$, which exists by Property~\ref{prop1bis}.
Let $H'$ be an inclusion-wise minimal hyperedge of $\H$ such that $v \in H' \subseteq H_v$. We get that $H' \in E(\tilde{\H})$, and $H' \cap B = \{v\}$. 
We now repeat this operation of removing a hyperedge containing another until no hyperedge is included in another, and define $\H'$ as the obtained hypergraph. By Property~\ref{prop1}, it follows that $B$ is a minimal hitting set in $\tilde{\H}$.

\emph{Property~\ref{prop4}}.
Suppose that there exists $v \in V(\H)$ such that no hyperedge of $\H$ contains $v$.
Let $\tilde{\H}$ be the hypergraph obtained from $\H$ by removing $v$.
We immediately have that for every $\Y \subseteq V(\H)$, $\Y$ is a minimal hitting set in $\H$ if and only if $\Y$ is a minimal hitting set in $\tilde{\H}$, as
no minimal hitting set in $\tilde{\H}$ can contain vertex $v$.
We now repeat this operation until we get the claimed hypergraph $\H'$.
\end{proof}


\section{Parameterizations by $\alpha$ and $\beta$}\label{section:alphabeta}
\label{section:alphabeta}
In this section we establish the results summarized in Table~\ref{fig:array}
about the parameterized complexity of \MMBS and \MMHS under several parameterizations depending on $\alpha$ and $\beta$. We present the negative and the positive results in Section~\ref{sec:negative} and Section~\ref{sec:positive}, respectively.

\subsection{Hardness results}
\label{sec:negative}

It is natural to ask, for a graph $G$, whether computing $\alpha(G)$ can help toward computing $\mmbs(G)$, and vice-versa. In fact, the parameters $\alpha$ and $\mmbs$ are linked by the duality relation discussed in what follows.

Given a ground set $S$, a \emph{clutter} is a family $\A$ of subsets of $S$ such that no set $A_1 \in \A$ contains another set $A_2 \in \A$.
Given a clutter $\A$, the family of \emph{blocking sets} of $\A$, denoted by $b(\A)$, is the set of minimal subsets $B$ of $S$ such that $B$ intersects every set $A \in \A$.
Notice that $b(\A)$ is a clutter, and thus $b(b(\A))$ is well-defined. The following theorem provides a duality relation and can be found, for instance, in~\cite{berge1984hypergraphs}.

\begin{theorem}\label{thm:duality}
$b(b(\A))=\A$.
\end{theorem}



If we apply Theorem~\ref{thm:duality} to our setting, namely with $\A$ being the set of all \mis of a graph $G$, we get that $b(\A)$ is the set of all minimal blocking sets,
and that the set of minimal sets intersecting all the sets in $b(\A)$ is the set of all \mis.
Even if this theorem gives a relation between \mbs and \mis, it seems, to the best of our knowledge, that it does not provide
a way to compute $\alpha(G)$ from $\mmbs(G)$, or $\mmbs(G)$ from $\alpha(G)$.

Let us start with the easy direction.
\begin{property}
  Let $\F$ be a hereditary graph class.
  If the problem of computing an \mbs (of any size) is polynomial on $\F$, then \MIS is polynomial on $\F$.
  This implies that if \MMBS is polynomial on $\F$, then \MIS is polynomial on $\F$.
\end{property}


\begin{proof}
  Suppose that we have an algorithm that, given a graph $G \in \F$, outputs in polynomial time an \mbs $\Y$ of $G$.
  According to Lemma~\ref{lem:prelim} (Property~\ref{prop1}), for every $v \in \Y$ there exists a \mis  $I_v$ such that $I_v \cap \Y = \{v\}$,
  implying that $\alpha(G \setminus \Y) = \alpha(G)-1$. As $G \setminus \Y \in \F$ because $\F$ is hereditary, we can repeat the same argument to $G \setminus \Y$, stopping when  we obtain an empty graph. It follows that $\alpha(G)$ is equal to the number of iterations of this procedure.
\end{proof}

Let us now show that there is no hope to get the same kind of property in the backward direction.
We point out a related result in~\cite{DBLP:conf/stacs/HolsKP20} showing that there is a graph class $\F$ where $\mmbs(\F)=1$ (as there is a unique \mis for any $G \in \F$, see Property~\ref{prop2} of Lemma~\ref{lem:prelim}),
but \MIS is not polynomial unless $\NP=\RP$.
This result is obtained through a reduction from \textsc{Unique-SAT} and guarantees that if the original instance
is a \yes-instance, then there is a unique \mis of size $k$, and otherwise a unique \mis of size $k-1$.
In the following property, the situation is different as we target a complexity result for \MMBS,  and not \MIS. For a graph class $\F$, let $\alpha(\F)=\sup_{G \in \F}\alpha(G)$.


\begin{proposition}\label{prop:hardness} There exist
  \begin{itemize}
  \item a (hereditary) graph class $\F$ where $\alpha(\F) \le 2$ and on which \MMBS  is \NPh (implying that $2$-\MMBS is \NPh), and
  \item a graph class $\F$ where \MIS is polynomial, and \MMBS/$\alpha$
is \Woneh, even the particular case of deciding, given an input graph $G$, whether $\mmbs(G) > 1$. This implies that \MMBS/$\beta$ is \ParaNPH, and that \MMBS/($\alpha+\beta$) is \Woneh.
  \end{itemize}
\end{proposition}

\begin{proof}
We consider the reduction of Boria et al.~\cite[Theorem 1]{BORIA201562} from \MIS to \MMVC showing an inapproximability result for \MMVC.
  Given an input graph $G_0$ of the \MIS problem, this reduction produces a graph $G$ by starting from $G_0$, and adding for every $v \in V(G_0)$ a new private vertex $v_p$
  with $N_{G}(v_p) = \{v\}$. As it is known that \MIS remains \NPh on triangle-free graphs (just by subdividing every edge twice), 
  and as this reduction does not create triangles, we get that \MMVC is \NPh on triangle-free graphs.

  To prove our first statement, we reduce from \MMVC on triangle-free graphs, and given an input $G$ of \MMVC, we define our input $G'$ of \MMBS as the complement of $G$
  (that is, the graph obtained from $G$ by swapping edges and non-edges). We may also assume that $G$ contains at least one edge.
  Observe that as $G$ is triangle-free and $G$ contains at least one edge, $\alpha(G') = 2$. Moreover, there is a bijection between edges of $G$ and \mis of $G'$.
  This implies that for every subset $\Y \subseteq V(G)$, $\Y$ is a vertex cover of $G$ if and only if $\Y$ is a \bs of $G'$, and thus that $\Y$ is a minimal vertex cover of $G$ if and only if $\Y$ is an \mbs of $G'$.

  To prove our second statement, we reduce from the \textsc{Multicolored Independent Set} (\MCIS) problem.
  It is known~\cite{CyganFKLMPPS15} that \MCIS/$k$ is \Woneh. Given an input $(G,k)$ of \MCIS, where $V(G)$ is partitioned into $k$ cliques $\{V_i \mid i \in [k]\}$,  let $G_1$ be a copy of $G$ and
  let $G_2$ be the graph composed of an \is of size $k$. We define $G'$ as the graph obtained by taking the disjoint union of $G_1$ and $G_2$, and adding all edges between $V(G_1)$ and $V(G_2)$.
  Observe that $|V(G')|=|V(G)|+k$, $\alpha(G') \le k$ as $\alpha(G_i) \le k$, and $\alpha(G')=k$ as $V(G_2)$ is an \is.

  If $(G,k)$ is a \yes-instance, there are  two distinct (even disjoint) \mis $I_i$ in
  $G'$: we can define $I_1 \subseteq V(G_1)$ as an \is of size $k$ in $G$, and $I_2 = V(G_2)$. This implies by Lemma~\ref{lem:prelim} (Property~\ref{prop2})
  that $\mmbs(G') \ge 2$.

  Conversely, if $(G,k)$ is a \no-instance, the unique \mis of $G'$ is $V(G_2)$, implying by Lemma~\ref{lem:prelim} (Property~\ref{prop2}) that $\mmbs(G')=1$.
  Finally, this is a parameterized reduction as $\alpha(G') \le k$, and \MIS is clearly polynomial restricted to the family of graphs produced by the reduction.
\end{proof}

Let us now turn to lower bounds for \MMHS/$\beta$.
It is known that \MMHS/$\beta$ is \Woneh~\cite{BAZGAN20182}
and that, unless the \ETH fails, \SIMPLEEXTMMHS cannot be solved in time $f(|X|) \cdot (n+m)^{o(|X|)}$ for any computable function $f: \Bbb{N} \to \Bbb{N}$, where $n$ and $m$ are the number of vertices and hyperedges of the input hypergraph, respectively~\cite{blasius2019efficiently}.
In our next theorem we prove the same lower bound for \UPDOM, transferring the result to \MMHS as well (recall that \UPDOM is a special case of \MMHS).

Let us mention that the reduction for \MMHS/$\beta$ of Bazgan et al.~\cite{BAZGAN20182} is  a reduction from \textsc{Multicolored Independent Set} parameterized by $k$, showing that, in fact,  \UPDOM is \Woneh parameterized by the solution size,
where the parameter of the \UPDOM instance is $\O(k^2)$. While being indeed a
parameterized reduction, it only implies
that, unless the \ETH fails, \UPDOM cannot be solved in time $f(k) \cdot (n+m)^{o(\sqrt{k})}$ for any computable function $f: \Bbb{N} \to \Bbb{N}$. We also mention that very recently and independently from our work,
Theorem~\ref{thm:updom} has also been proved by Dublois et al.~\cite{dublois2021upper}, by using a reduction quite similar to ours.



\begin{theorem}\label{thm:updom}
Unless the \ETH fails, the \UPDOM problem  cannot be solved in time $f(k) \cdot |V(G)|^{o(k)}$ for any computable function $f: \Bbb{N} \to \Bbb{N}$.
\end{theorem}
\begin{proof}
Chen et al.~\cite{CHEN20061346} proved that, unless the \ETH fails, the {\sc $k$-Clique} problem (where we have to decide if a given graph have a clique of size at least $k$) cannot be solved in time $f(k) \cdot n^{o(k)}$ on $n$-vertex graphs for any computable function $f: \Bbb{N} \to \Bbb{N}$.
Since there is a simple parameterized reduction from {\sc $k$-Clique} to te \textsc{Multicolored Independent Set} problem parameterized by $k$, namely \MCIS/$k$,
with linear dependency on the parameter (see for instance~\cite{CyganFKLMPPS15}), the result of Chen et al.~\cite{CHEN20061346} implies that
the \MCIS cannot be solved in time $f(k) \cdot n^{o(k)}$ on $n$-vertex graphs for any computable function $f: \Bbb{N} \to \Bbb{N}$.

We present a parameterized reduction from \MCIS to \UPDOM such that, given an instance $(G,k)$ of \MCIS, creates in polynomial time a graph $G'$ that contains a minimal dominating of size at least $3k$ if and only if $G$ contains a multicolored \is of size $k$. By the above discussion, such a  reduction concludes the proof of the theorem.

Given $G$, with $V(G)= V_1 \uplus \cdots \uplus V_k$, for every $i \in [k]$ we add to $G'$ three copies $A_i,B_i,C_i$ of $V_i$, and let $U_i = A_i \cup B_i \cup C_i$ be their union. We denote $A = \bigcup_{i \in [k]}A_i$, $B = \bigcup_{i \in [k]}B_i$, $C = \bigcup_{i \in [k]}C_i$, and, for a vertex $v \in V(G)$, we denote by $v_A,v_B,v_C$ its corresponding copy in $A,B,C$, respectively. For every $i \in [k]$, the set
 $U_i$ induces a clique minus the triangles $\{ (v_A,v_B,v_C) \mid v \in V_i\}$. That is, within the same color $i$, every vertex is adjacent to all other vertices except for its two other copies. For every edge $\{u,v\} \in E(G)$ such that $u \in V_i$ and $v \in V_j$ with $i \neq j$, we add to $G'$ the edges $\{u_A,v_B\}$ and $\{u_B,v_A\}$. This concludes the construction of $G'$. We claim that $G$ contains a multicolored \is of size $k$ if and only if $G'$ that contains a minimal dominating of size at least $3k$.

Let first $S \subseteq V(G)$ be a multicolored \is of size $k$. Let $D \subseteq V(G')$ contain, for every vertex $v \in S$, its three copies $v_A,v_B,v_C$. Note that $|D|=3k$. We claim that $D$ is a minimal dominating set of $G'$. Since $D$ contains a vertex in each of the $3k$ cliques into which $V(G')$ is partitioned, $D$ is clearly a dominating set. Consider a vertex $v_A \in D \cap A$ (the case $v_B \in D \cap B$ is symmetric). Then $D \setminus \{v_A\}$ is not a dominating set, since by the hypothesis that $S$ is an \is, no vertex in $D \setminus \{v_A\}$ is adjacent to $v_A$. Consider now a vertex $v_C \in D \cap C$, with $v \in V_i$. Then $D \setminus \{v_C\}$ is not a dominating set either, as $D \cap (A_i \cup B_i) = \{v_A,v_B\}$, and none of $v_A$ and $v_B$ is adjacent to $v_C$. Hence, $D$ is a minimal dominating set of $G'$ and we are done.

Conversely, let $D \subseteq V(G')$ be a minimal dominating set with $|D| \geq 3k$.

\begin{claim}\label{claim:-at-most-one-each}
For every $i \in [k]$, $|D \cap A_i| \leq 1$ and $|D \cap B_i| \leq 1$.
\end{claim}
\begin{proofclaim}
We say that an index $i \in [k]$ is \emph{abnormal} if $|D \cap A_i| \geq 2$ or $|D \cap B_i| \geq 2$ (or both), and \emph{normal} otherwise. We will construct a set $D' \subseteq V(G')$ with $|D'| = |D|$ such that if $i$ is normal, $|D' \cap U_i | \leq 3$, and if $i$ is abnormal, $|D' \cap U_i | \leq 2$. Hence, if there exists an abnormal index, it holds that
$|D| = |D'| < 3k$, contradicting the hypothesis that $|D| \geq 3k$. We now proceed to the construction of  $D'$, which is not required to be a dominating set of $G'$. We start with $D' = D$, and we update $D'$ as follows.

Let $i$ be an abnormal index. Since $|D \cap A_i| \geq 2$ or $|D \cap B_i| \geq 2$, by construction of $G'$ we have that $D \cap (A_i \cup B_i)$ dominates $U_i$, and since $N(C_i) = A_i \cup B_i$, necessarily  $|D \cap C_i| = 0$, as otherwise $D$ would not be minimal. If $|D \cap U_i | \leq 2$ we do nothing, as we already have that $|D' \cap U_i | = |D \cap U_i | \leq 2$. Assume henceforth that $|D \cap U_i | = |D \cap (A_i \cup B_i)| \geq 3$.

To simplify the presentation, suppose that $|D \cap A_i| \geq |D \cap B_i|$ (the other case is symmetric), so we have that $|D \cap A_i| \geq 2$. Note that any two vertices in the set $D \cap A_i$, say $u_A$ and $v_A$, dominate the whole set $U_i$. Hence, since $D$ is a minimal dominating set of $G'$, for every other vertex $w_A \in (D \cap A_i) \setminus \{u_A,v_A\}$ (resp. $w'_B \in D \cap B_i$) there must exist an index $j \neq i$ (resp. $\ell \neq i$) and a vertex $z_B \in B_j$ (resp. $z'_A \in A_{\ell}$) not in $D$ and dominated only by $w_A$ (resp. $w'_B$), that is, with $z_B \notin D$ (resp. $z'_A \notin D$) and $N_D(z_B) = \{w_A\}$ (resp. $N_D(z'_A) = \{w'_B\}$); see Figure~\ref{fig:fig-UD}(a) for an illustration, where the vertices in $D$ are depicted in red. Note that such an index $j$ (resp. $\ell$) is necessarily normal, as otherwise vertex $z_B$ (resp. $z'_A$) would be already dominated within $U_j$ (resp. $U_{\ell}$). Note also that, for the same reason, $D \cap B_j = \emptyset$ (resp. $D \cap A_{\ell} = \emptyset$). For each such a vertex $w_A \in D$ (resp. $w'_B \in D$), we remove vertex $w_A$ (resp. $w'_B$) from $D'$ and we add vertex $z_B$ (resp. $z'_A$) to $D'$;  see Figure~\ref{fig:fig-UD}(b) for an illustration, where the vertices in $D'$ are depicted in red. We say that vertex $z_B \in B_j$ (resp. $z'_A \in A_{\ell}$) is a \emph{sink}. This concludes the construction of $D'$. It just remains to verify that the claimed properties of $D'$ are satisfied.

\begin{center}
    \begin{figure}[htb]
      \begin{center}
        \includegraphics[width=.82\textwidth]{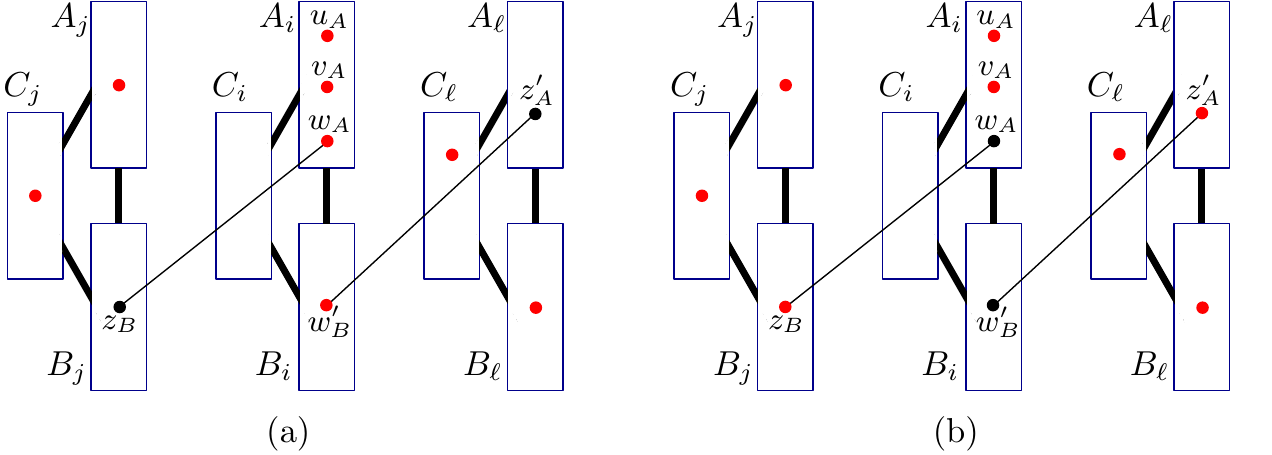}
      \end{center}
      \caption{Configuration in the proof of Claim~\ref{claim:-at-most-one-each}. Index $i$ is abnormal, while indices $j$ and $\ell$ are normal. Vertices $z_B$ and $z'_A$ are sinks. (a) The vertices in $D$ are depicted in red. (b) The vertices in $D'$ are depicted in red.}
      \label{fig:fig-UD}
    \end{figure}
  \end{center}

By construction, we clearly have that $|D'| = |D|$. Note that, if $i$ is an abnormal index as in the above paragraph, it cannot contain any sink since all the vertices of $U_i$ are already dominated by $D \cap U_i$. Hence, no vertex is added to $D' \cap U_i$ and it holds that $D' \cap U_i = \{u_A, v_A\}$, so we indeed have that $|D' \cap U_i| \leq 2$.

It remains to verify that, if $j$ is a normal index, then $|D' \cap U_j| \leq 3$.
Since the vertices in $C_j$ have neighbors only in $U_j$, necessarily $|D \cap U_j| \geq 1$. Hence, at most one vertex in $A_j$ and at most one vertex in $B_j$ are not dominated by the vertices in $D \cap U_j$. Thus, each of $A_j$ and $B_j$ contains at most one sink. We distinguish three cases according to the number of sinks in $A_j \cup B_j$.

Suppose first that $A_j \cup B_j$ contains no sink. Then $|D' \cap U_j| = |D \cap U_j| \leq 3$, where the last inequality follows easily by using that, since $j$ is normal, $|D \cap A_j| \leq 1$ and $|D \cap B_j| \leq 1$.

Suppose now that $A_j \cup B_j$ contains exactly one sink, so we have that $|D' \cap U_i| = |D \cap U_i| +1$. Suppose without loss of generality that the sink is a vertex $z_A \in A_j$, so we have $|D \cap A_j| = 0$. Since $z_A$ is not dominated by $D \cap U_i$, necessarily $|D \cap B_j| \leq 1$ and $|D \cap C_j| \leq 1$, so $|D \cap U_i| \leq 2$ and $|D' \cap U_i| \leq 3$.

Finally, suppose that $A_j \cup B_j$ contains two sinks $z_A \in A_j$ and $z'_B \in B_j$, so we have that $|D' \cap U_i| = |D \cap U_i| +2$, $|D \cap A_j| = 0$, and $|D \cap B_j| = 0$. Also, since none of $z_A$ and $z'_B$ can be dominated by $D \cap U_j$, necessarily $z=z'$ and $D \cap C_j = \{z_C\}$, so $|D \cap C_j| =1$. Thus, $|D \cap U_i| \leq 1$ and $|D' \cap U_i| \leq 3$, and the claim follows.
\end{proofclaim}

\begin{claim}\label{claim:at-most-three}
For every $i \in [k]$, $|D \cap U_i| \leq 3$.
\end{claim}
\begin{proofclaim}
Suppose for contradiction that there exists $i \in [k]$ such that $|D \cap U_i| \geq 4$. By Claim~\ref{claim:-at-most-one-each}, necessarily $|D \cap C_i| \geq 2$. If $|D \cap C_i| \geq 3$, then deleting all but any two vertices in $D \cap C_i$ results in a proper subset of $D$ that is still a dominating set of $G'$, contradicting the minimality of $D$. Hence $|D \cap C_i| =2$, $|D \cap A_i| =1$, and $|D \cap B_i| =1$. Let $u_A \in D \cap A_i$ and let $v_C,w_C \in D \cap C_i$. At least one among $v$ and $w$, say $v$, is not equal to $u$. Then the set $D \setminus \{w_C\}$ is still a dominating set of $G'$, contradicting again the minimality of $D$.
\end{proofclaim}

\begin{claim}\label{claim:one-each}
For every $i \in [k]$, $|D \cap A_i| = |D \cap B_i| = |D \cap C_i| = 1$.
\end{claim}
\begin{proofclaim}
Since by hypothesis we have that $|D| \geq 3k$ and by Claim~\ref{claim:at-most-three} it holds that $|D \cap U_i| \leq 3$ for every $i \in [k]$, necessarily $|D \cap U_i| = 3$ for every $i \in [k]$. Since by Claim~\ref{claim:-at-most-one-each} we have that
 $|D \cap A_i| \leq 1$ and $|D \cap B_i| \leq 1$ for every $i \in [k]$, we conclude that
 $|D \cap A_i| = |D \cap B_i| = |D \cap C_i| = 1$ for every $i \in [k]$.
\end{proofclaim}

We proceed to define from $D$ a multicolored \is $S \subseteq V(G)$ with $|S|= k$. Consider an arbitrary index $i \in [k]$. By Claim~\ref{claim:one-each}, $D \cap A_i = \{u_A\}$, $D \cap B_i = \{v_B\}$, and $D \cap C_i = \{w_C\}$. Note that if $u \neq v$, then $u_A$ and $v_B$ would dominate the whole set $U_i$ and $w_C$ could be removed from $D$, contradicting its minimality. Thus, we have that $u = v$, which in turn implies that $w=u$ as well. We define $S \cap V_i = \{v\}$. It remains to verify that $S$ is indeed an \is of $G$. Consider $u,v \in S$ with $u \in V_i$ and $v \in V_j$. If $\{u,v\} \in E(G)$ then $D \setminus \{u_A\}$ would still be a dominating set of $G'$. Indeed, vertex $u_A$ would be dominated by $v_B$, and the other vertices in $A_i$ would still be dominated by $u_B$. Thus, $\{u,v\} \notin E(G)$ and $S$ is an \is in $G$.
\end{proof}

Theorem~\ref{thm:updom} immediately yields the following corollary for \MMHS.

\begin{corollary}\label{cor:mmhs}
Unless the \ETH fails, \MMHS cannot be solved in time $f(\beta) \cdot (|V(\H)|+|E(\H)|)^{o(\beta)}$ for any computable function  $f: \Bbb{N} \to \Bbb{N}$.
\end{corollary}

\subsection{Positive results}
\label{sec:positive}

Let us now turn to positive results, and consider a class where \MIS is polynomial.
As according to Proposition~\ref{prop:hardness} we cannot hope for solving \MMBS in polynomial time, we consider the parameterized complexity of the \MMBS problem.
The first results show the crucial difference between the problems of, given a graph $G$ and a positive integer $\beta$, deciding whether $\mmbs(G) = \beta$ and deciding whether $\mmbs(G) \ge \beta$.
In any maximization problem, the first property implies the second one, but the backward implication is not always true.
In particular, for \MMBS, $\mmbs(G) \ge \beta$ does {\sl not} imply that $G$ contains an \mbs of size {\sl exactly} $\beta$,
and this is informally what makes the inequality version harder.

As observed in~\cite{FernauHDR} or in~\cite[Proposition 1]{DBLP:journals/disopt/Damaschke11}, deciding whether there exists a minimal hitting set of size exactly  $\beta$, or  at most $\beta$, in a hypergraph with hyperedges of size at most $\alpha$ can be trivially decided by a search-tree in time $\O^*(\alpha^\beta)$. However, we cannot use directly this result, as a reduction to \MMHS would require time $n^{\alpha(G)}$ to generate the hyperedges, and thus we have to define an ad-hoc algorithm, which is also based on branching.

\mbvfinale{mention running time for MMHS/alpha+beta obtained avec kernel puis enum : ca fait n to alpha (a verif) imes alpha to the beta, qui est incomparable avec 2 to the alpha beta}

\begin{proposition}\label{prop:compute_beta_equal}
  Let $\F$ be a hereditary graph class on which \MIS is polynomial. Then
  \MMBSEG/$(\alpha+\beta)$ and \MMBSINFEG/$(\alpha+\beta)$ are \FPT restricted to input graphs in $\F$. More precisely, they can both be solved in time $\O^*(\alpha(G)^\beta)$.
\end{proposition}
\begin{proof}
 We only prove the result for \MMBSEG, as it directly implies the result for
 \MMBSINFEG. Consider an input graph $G \in \F$.
Let us define an algorithm $A(X)$ that, given a set $X \subseteq V(G)$ with $|X| \le \beta$, answers ``\yes'' if and only if there exists an \mbs $\Y$ of $G$ such that $X \subseteq \Y$ and $|\Y|=\beta$,
  in which case we say that $X$ is a \emph{\yes-set}.
  The algorithm starts with $X= \emptyset$, and calls itself recursively for a larger set $X$ obtained from branching on vertices of a \mis of $G \setminus X$, as detailed below.
  Note that if $|X|=\beta$, then $A(X)$ answers ``\yes'' if and only if $X$ is an \mbs.
  This can be checked in polynomial time, as it is equivalent to the properties that $\alpha(G \setminus X) < \alpha(G)$ and $\alpha(G \setminus (X \setminus \{v\})) = \alpha(G)$ for every $v \in X$.
  Let us now consider the cases where $0 \leq |X| < \beta$. If $\alpha(G \setminus X) < \alpha(G)$, then we answer ``\no'' as $X$ is already a \bs, and thus no superset $X' \supsetneq X$ can be an \mbs.
  Otherwise, we have that $\alpha(G \setminus X) = \alpha(G)$. Since $G \setminus X \in \F$ as $\F$ is hereditary and  \MIS is polynomial on $\F$, we can compute in polynomial time a \mis $I$ of $G \setminus X$
  (notice that beeing able to solve \MIS implies that we can construct such a \mis). Observe that
  $I$ is also a \mis of $G$, and that $I \cap X = \emptyset$. In this case, $A(X)$ returns $\bigvee_{v \in I}A(X \cup \{v\})$.

Let us now prove the correctness of this latter case.
  Suppose first that $X$ is a \yes-set, and let $\Y$ be an \mbs in $G$ of size $\beta$ such that $X \subseteq \Y$. As $\Y$ is a \bs, there exists $v \in I \cap \Y$.
  As $I \cap X = \emptyset$, we get $v \notin X$, and thus $X \cup \{v\}$ is a \yes-set and $A(X)$ returns ``\yes''.
  We prove the other direction by reverse induction on $|X|$, the case $|X|=\beta$ being correct as discussed above. Consider a set $X$ with $|X|<\beta$, and suppose inductively that the claimed property is correct for sets of size $|X|+1$. Thus, if $A(X \cup \{v\})$ returns ``\yes'' for some $v \in I$, then by induction $(X \cup \{v\})$ is a \yes-set, implying by definition of $A$ that $X$ is also a \yes-set.

  Let us finally discuss the running time of the algorithm, given an input graph $G$. Starting with $X= \emptyset$, for every set $X$ the algorithm performs a polynomial number of operations,
  and then branches on a set $I$ of size at most $\alpha(G)$, as such a  set $I$ is always a \mis of a subgraph of $G$). As the depth of the branching tree corresponding to the algorithm has depth at most $\beta$, the claimed running time follows.
\end{proof}

Observe that,  unless \FPT = \Won, we cannot obtain results similar to Proposition~\ref{prop:compute_beta_equal} to decide whether $\mmbs(G) \ge \beta$, and even in time $f(\alpha(G),\beta)$ for any computable function $f$, as it would imply that \MMBS/$(\alpha+\beta)$ is \FPT, contradicting the fact that \MMBS/$\beta$ is \ParaNPH by Proposition~\ref{prop:hardness}.
Thus, we need consider a stronger assumption than assuming that \MIS is polynomial on $\F$. Namely, in what follows we consider the $\alpha$-\MMBS problem, that is, the case where $\alpha$ is fixed.
Recall that, according to Proposition~\ref{prop:hardness}, even $2$-\MMBS remains \NPh, motivating the study of the parameterized complexity of $\alpha$-\MMBS.

\begin{proposition}\label{prop:fixedalpha_beta_fpt}
  For every fixed positive integer $\alpha$,
  $\alpha$-\MMBS/$\beta$ and $\alpha$-\MMHS/$\beta$ are \FPT. More precisely,
  both problems can be solved in time $\O^*(\alpha^{\beta})$.
\end{proposition}
\begin{proof}
  Given an input graph $G$ of $\alpha$-\MMBS, we compute in time $\O^*(n^{\alpha})$ the hypergraph $\H$ where $V(\H)=V(G)$, and $H$ is a hyperedge in $\H$ if and only if $H$ is a \mis in $G$.
  By definition of $\alpha$-\MMBS, all hyperedges of $\H$ have size exactly $\alpha$, and for every $\Y \subseteq V(G)$, $\Y$ is an \mbs in $G$ if and only if $\Y$ is a minimal hitting set in $\H$.
  Then, according to~\cite[Lemma 6]{DBLP:journals/disopt/Damaschke11}, as $\alpha$ is fixed we can decide whether there is a minimal hitting set of $\H$ of size at least $\beta$ in time $\O^*(\alpha^{\beta})$.
\end{proof}

We point out that in both Proposition~\ref{prop:fixedalpha_beta_fpt} and~\cite[Lemma 6]{DBLP:journals/disopt/Damaschke11}, in order to decide whether there is a minimal hitting set of $\H$ of size at least $\beta$, there is term $|V(\H)|^{f(\alpha)}$ hidden inside the $\O^*$-notation.
This means that~\cite[Lemma 6]{DBLP:journals/disopt/Damaschke11} does not imply that \MMHS/$(\alpha+\beta$) is \FPT (recall that function $\alpha$ in the parameterization of \MMHS denotes
the size of a largest hyperedge of $\H$).
However, according to the two following propositions, it turns out that \MMHS/$(\alpha+\beta$) is indeed \FPT. This highlights a difference between
 \MMHS and \MMBS, as according to Proposition~\ref{prop:hardness} \MMBS/$(\alpha+\beta$) is unlikely to be \FPT.

Let us start with a kernelization result, using the well-known notion of sunflower.
\begin{definition}\label{def:sunflower}
Let $\beta \in \Bbb{N}$. Given a hypergraph $\H$, a \emph{sunflower} in $\H$ with $\beta$ \emph{petals} and \emph{core} $C \subseteq V(\H)$ is a collection of $\beta$ hyperedges  $H_1,\dots,H_\beta$ of $\H$ such that
$H_i \cap H_j = C$ for all $i \neq j$, and for every $i \in [\beta]$ the so-called petal $H_i \setminus C$ is not empty.
We say that a function $s: \Bbb{N}^2 \to \Bbb{N}$ is a \emph{sunflower function} if, for every hypergraph $\H$ whose hyperedges have size at most $\alpha$,
if $|E(\H)| > s(\alpha,\beta)$ then $\H$ admits a sunflower with $\beta$ petals.
\end{definition}

It is known (see for instance~\cite{CyganFKLMPPS15}) that $s_1(\alpha,\beta):=(\alpha^2) !\alpha (\beta-1)^\alpha$ is a sunflower function.
Even if this is not relevant for our next proposition, where $\alpha$ is fixed, we point out that
this bound has been recently been improved by Rao~\cite{rao2020coding} to
\begin{equation}\label{eq:sunflower}
 s_2(\alpha,\beta):=(c\beta \cdot \log(\alpha \beta))^{\alpha}.
\end{equation}
\vspace{-.6cm}
\begin{lemma}\label{lemma:sunflower}
Let $\beta \in \Bbb{N}$, and let $\H$ be a hypergraph such that no hyperedge is included in another hyperedge.
If $\H$ has a sunflower with $\beta$ petals, then $\mmhs(\H) \ge \beta$.
\end{lemma}
\begin{proof}
Let $\{H_i \mid i \in [\beta]\}$ be a sunflower of $\H$ with $\beta$ petals.
  Let $C = \bigcap_{i \in [\beta]}H_i$ and $S = V(\H) \setminus C$. As there is no $H \in E(\H)$ such
  that $H \subseteq C$, it follows that $S$ is a hitting set of $\H$. Let $S' \subseteq S$ be a minimal hitting set of $\H$.
  For every $i \in [\beta]$, $H_i \setminus C \neq \emptyset$ by definition of a sunflower, hence $S'$ must contain at least one vertex in $H_i \setminus C$. This implies that $|S'| \ge \beta$,
  and thus that $\mmhs(H) \ge \beta$.
\end{proof}

\begin{proposition}\label{prop:fixedalpha_beta_kernel}
Let $s(\alpha,\beta)$ be a sunflower function that is polynomial in $\beta$ for fixed $\alpha$.
For every fixed integer $\alpha$,
  \begin{itemize}
  \item $\alpha$-\MMBS/$\beta$ admits a polynomial kernel with at most $\alpha \cdot s(\alpha,\beta)$ vertices, which can be constructed in time $\O^*(|V(G)|^\alpha)$, and
  \item $\alpha$-\MMHS/$\beta$ admits a polynomial kernel with at most $\alpha \cdot s(\alpha,\beta)$ vertices, which can be constructed in time $\O^*(|V(\H)|)$ (that is, not depending on $\alpha$).
  \end{itemize}
\end{proposition}
\begin{proof}
  Let us start with $\alpha$-\MMHS. Consider an instance $(\H,\beta)$ of $\alpha$-\MMHS. By Lemma~\ref{lem:prelim} (Property~\ref{prop4}), we can compute in polynomial time
  an equivalent instance $(\H',\beta)$ where every vertex belongs to a hyperedge. If $|E(\H')|\le s(\alpha,\beta)$,
  then we get $|V(\H')| \le \alpha \cdot |E(\H')| \leq \alpha \cdot  s(\alpha,\beta)$ and we are done. Otherwise, as $s$ is a sunflower function, it follows  that $\H'$ contains a sunflower with $\beta$ petals.
   According to Lemma~\ref{lemma:sunflower}, we get that $(\H,\beta)$ is a yes-instance.

  The result for $\alpha$-\MMBS is now straightforward, but we provide the details as we cannot directly say that any kernel for $\alpha$-\MMHS implies a kernel for \MMBS. For example,
  removing a hyperedge in  $\alpha$-\MMHS cannot necessarily be translated to  $\alpha$-\MMBS.
  Let $(G,\beta)$ be an instance of $\alpha$-\MMBS. In time $\O^*(|V(G)|^\alpha)$ we can compute an equivalent instance $(\H,\beta)$ of $\alpha$-\MMHS
  by creating a hyperedge for every \mis of $G$.
  As all hyperedges have size exactly $\alpha(G)$, no hyperedge can be included in another hyperedge.
  Now, if the previous kernel detects a \yes-instance, we are done, and otherwise we output $(G',\beta)$ where $G'=G[V(\H')]$.
  As all vertices of $V(G) \setminus V(G')$ are vertices that do not belong to any \mis of $G$, by the arguments of the proof of Lemma~\ref{lem:prelim} (Property~\ref{prop4}),
  we get an equivalent instance.
\end{proof}



Even if Proposition~\ref{prop:fixedalpha_beta_kernel} implies that \MMHS/$(\alpha+\beta$) is \FPT, the running obtaining by applying brute force to the kernelized instance
is $\O^*(2^{\alpha \cdot s(\alpha,\beta)})$, and thus doubly exponential in $\alpha$.
This motivates the question of obtaining a faster \FPT algorithm for \MMHS/$(\alpha+\beta$).
We point out that trying to improve the running time by considering separated parameters, instead of the aggregated parameter $\alpha+\beta$,
is not possible as \MMHS/$\beta$ is \Woneh~\cite{BAZGAN20182}, and \MMHS/$\alpha$ is already \NPh for $\alpha=2$, as it corresponds to \MMVC.
A first way to get a faster \FPT algorithm is to reduce to the \SIMPLEEXTMMHS problem.



Bläsius et al.~\cite{blasius2019efficiently} proved that \EXTMMHS (and thus \SIMPLEEXTMMHS) can be solved in time $\O^*(\lambda^{|X|})$, where $\lambda=\min\left(\frac{|E(\H)|}{|X|},\Delta(\H)\right)$
and $\Delta(\H)= \max_{v \in V(\H)}$ $|\{H \in E(\H) \mid v \in H\}|$ is the maximum degree of $\H$. Informally, this algorithm, in the simplified setting of \SIMPLEEXTMMHS, just
guesses for each $x \in X$ its ``private'' hyperedge $H_x$ such that $H_x \cap X = \{x\}$, and checks that there is no $H \in E(\H)$ such that $H \subseteq (\bigcup_{x \in X}H_x) \setminus X$.
Thus, guessing a hyperedge for every $x \in X$ yields the claimed running time. In the next proposition we formalize these ideas, using ideas similar to the proof
of Proposition~\ref{prop:compute_beta_equal}.

\begin{proposition}\label{prop:algommhs}
We can decide an instance $(\H,\beta)$ of \MMHS in time $\O^*((\alpha(\H) \cdot \lambda)^\beta)$, where
$\lambda=\min\left(\frac{|E(\H)|}{\beta},\Delta(\H)\right)$ and $\Delta(\H)= \max_{v \in V(\H)}$ $|\{H \in E(\H) \mid v \in H\}|$.
\end{proposition}

\begin{proof}
Let $(\H,\beta)$ be an instance of \MMHS.
Let us define an algorithm $A(\H,\beta,X)$ such that, given a set $X \subseteq V$ such that $|X| \le \beta$, decides whether there exists
a minimal hitting set $S$ of $\H$ such that $X \subseteq S$ and $|S| \ge \beta$. As in the proof of Proposition~\ref{prop:compute_beta_equal}, the algorithm starts with $X=\emptyset$.
If $|X|=\beta$ then we return $B(\H,X)$, where $B$ is the algorithm of Bläsius et al.~\cite{blasius2019efficiently}.
Let us now consider the cases where $0 \leq |X| < \beta$. If there is no $H \in E(\H)$ such that $H \cap X = \emptyset$ then we answer ``\no''.
Otherwise, let $H \in E(\H)$ such that $H \cap X \neq \emptyset$, and in this case $A(\H,\beta,X)$ returns $\bigvee_{v \in H}A(\H,\beta,X \cup \{v\})$.

Let us prove that $A$ is correct by induction on $\beta-|X|$. The only non-trivial case is when $A$ returns  $\bigvee_{v \in H}A(\H,\beta,X \cup \{v\})$. If $A(\H,\beta,X)$ returns ``\yes''
then there exists $v \in H$ such that $A(\H,\beta,X \cup \{v\})$ returns ``\yes'', implying by induction that $(\H,\beta,X \cup \{v\})$ is a \yes-instance,
and thus that $(\H,\beta,X)$ is a \yes-instance as well.
Conversely, if $(\H,\beta,X)$ is a \yes-instance certified by a solution $S^*$, then there exists $v^* \in S^* \cap H$ since $S^*$ is a hitting set, implying
that $(\H,\beta,X \cup \{v^*\})$ is a \yes-instance.

As at each step we branch on all vertices of a hyperedge $\H$, the running time is bounded by $\O^*(\alpha^\beta \cdot f(\alpha,\beta))$ where
$\O^*(f(\alpha,\beta))$ is the running time of algorithm $B$, implying $f(\alpha,\beta)=\O^*(\lambda^{\beta})$ and $\lambda=\min\left(\frac{|E(\H)|}{\beta},\Delta(\H)\right)$.
\end{proof}

Proposition~\ref{prop:algommhs} implies the following corollary, where the \XP algorithm follows directly from the algorithm of Proposition~\ref{prop:algommhs}. On the other hand,  the \FPT algorithm
is obtained by first applying the kernel for \MMHS of Proposition~\ref{prop:fixedalpha_beta_kernel} to ensure that $\Delta(\H) \le |E(\H)| \le s_2(\alpha,\beta)$, where
$s_2$ is the sunflower function of Equation~(\ref{eq:sunflower}).

\begin{corollary}\label{cor:mmhs_xp_fpt} The following claims hold:
\begin{itemize}
\item \MMHS/$\beta$ is \XP.
\item \MMHS/$(\alpha+\beta)$ is \FPT. More precisely, it can be solved in time $\O^*(\alpha^\beta(c \beta \cdot \log(\alpha \beta))^{\alpha \beta})$, where $c$ is the constant in the $s_2$ sunflower function of Equation~(\ref{eq:sunflower}).
\end{itemize}
\end{corollary}


Even if the algorithm of Proposition~\ref{prop:algommhs} gives a running time matching the lower bound of Corollary~\ref{cor:mmhs} for the dependency on $|E(\H)|$,
we can get a faster \FPT algorithm for \MMHS/$(\alpha+\beta)$ using an ad-hoc algorithm that  does not reduce to the extension problem. Namely, we present in Theorem~\ref{thm:fptv2} an algorithm for
\MMHS/$(\alpha+\beta)$ running in time $\O^*(2^{\alpha\beta})$.
We first need some preliminaries.

\mbvfinale{idee algo final : imaginer hypergraph avec alpha=3.
collect S un mhs. si plus petit que beta:
guess l'intersection d'un opt avec S,  appellons là X.
Soit EX les hyperaretes qui touchent pas X.
Soit S2 un mhs de EX.
Point clef 1, si S2 plus grand que beta, la ausssi on peut dire oui. Donc supposons S2 plsu petit beta.
La aussi on gess X2 la trace de l'opt dans S2.
Et ben là, il reste plus rien à decider, quand on a X et X2 (et les hyperaretes de taille 1), faut just verif que c'est bien un minimal hitting set,
car anyway on a pas le droit de rajouter encore des sommets ailleurs (car on a déjà un hitting set ).
On evite de se retrouver avec un X de taille beta à etendre}

\begin{definition}
  Let $\H$ be a hypergraph, let $I \subseteq V(\H)$ be an \is in $\H$, and let $X \subseteq V(\H)$. Let
  \begin{itemize}
    \item $\H_{I}$ such that $V(\H_{I})=V(\H) \setminus I$, and $E(\H_{I})=\{H \setminus I \mid H \in E(\H)\}$,
    \item $E^{\bar{X}} =\{ H \in E(\H) \mid H \cap X= \emptyset \}$, and
    \item $\H^{\bar{X}}$ such that $V(\H^{\bar{X}})=V(\H)$ and $E(\H^{\bar{X}})= E^{\bar{X}}$.
  \end{itemize}
\end{definition}

\begin{lemma}\label{lemma:fptv2:1}
  Let $\H$ be a hypergraph and let $I \subseteq V(\H)$ be an \is in $\H$.
  \begin{enumerate}
    \item For every minimal hitting set $S$ of $\H_I$, $S$ is also a minimal hitting set of $\H$. This implies $\mmhs(\H) \ge \mmhs(\H_{I})$.
    \item For every minimal hitting set $S^*$ of $\H$ such that $S^* \cap I = \emptyset$, $S^*$ is also a minimal hitting set of $\H_I$.
  \end{enumerate}
  \end{lemma}
\begin{proof}
  For the first property,  let $S$ be a minimal hitting set of $\H_I$.
  Let us first prove that $S$ is a hitting set of $\H$. Consider an arbitrary hyperedge $H \in E(\H)$. As $I$ is an \is, $H \setminus I \neq \emptyset$, and as $S$ is a hitting set of $\H_I$ and $H \setminus I \in E(\H_I)$, we get $S \cap (H \setminus I) \neq \emptyset$.
  Let us now prove that $S$ is minimal. Consider an arbitrary vertex $v \in S$. By the minimality in $\H_I$, there exists $H \in E(\H_I)$ such that $(S \setminus \{v\}) \cap H = \emptyset$,
  implying that $(S \setminus \{v\}) \cap (H \cup I) = \emptyset$ as $S \cap I = \emptyset$, where $H \cup I \in E(\H)$.

  For the second property, let $H' \in E(\H_I)$, where $H' = H \setminus I$, $H \in E(\H)$.
  As $S^*$ is a hitting set of $\H$, $S^* \cap H \neq \emptyset$, and as $S^* \cap I = \emptyset$, we get $S^* \cap H' \neq \emptyset$.
  Let us now verify the minimality. Consider an arbitrary vertex $v \in S^*$. As $S^*$ is minimal in $\H$, there exists $H \in E(\H)$ such that $(S^* \setminus \{v\}) \cap H = \emptyset$,
  implying $(S^* \setminus \{v\}) \cap (H \setminus I) = \emptyset$.
\end{proof}

\begin{lemma}\label{lemma:fptv2:2}
  Let $\H$ be a hypergraph, let $X \subseteq V(\H)$, and let $S'$ be a minimal hitting set of $\H^{\bar{X}}$. There exists a minimal hitting set $S$ of $\H$ such that $S' \subseteq S$.
\end{lemma}
\begin{proof}
  Let $S = S' \cup X$. Observe that $S$ is a hitting set of $\H$.
  Now, as far as there exists $v \in S \cap X$ such that $S \setminus \{v\}$ is still a hitting set of $\H$, remove $v$ from $S$.
  Let $S^*$ be the obtained set, which satisfies $S' \subseteq S^* \subseteq S$, and let us verify that $S^*$ is minimal. For every $v \in S^* \cap X$,
  by definition of $S^*$ we have that $S^* \setminus \{v\}$ is not a hitting set of $\H$. For every $v \in S^* \cap S'$, as $S'$ is minimal in
  $\H^{\bar{X}}$, it follows that there exists $H \in E^{\bar{X}}$ such that $(S' \setminus \{v\}) \cap H = \emptyset$. As $H \cap X = \emptyset$, we get
  $(S^* \setminus \{v\}) \cap H = \emptyset$ as well.
\end{proof}

We are now ready to present our \FPT algorithm.

\begin{definition}\label{def:improved-algo}
  For a positive integer $\beta$, we define algorithm $A^{\beta}(\H,X)$, where $\H$ is a hypergraph without empty hyperedges and $X \subseteq V(\H)$, as follows:
  \begin{itemize}
  \item If $E^{\bar{X}} = \emptyset$,
  \begin{itemize}
    \item if $|X| \ge \beta$ and $X$ is minimal hitting set of $\H$, return ``\yes''.
    \item Otherwise,  return ``\no''.
  \end{itemize}
  \item Otherwise,  let $S$ be a minimal hitting set of $\H^{\bar{X}}$.
    \begin{itemize}
      \item If $|S| \ge \beta$, return ``\yes''.
      \item Otherwise, return $\bigvee_{S_1 \in \mathcal{L}} A^{\beta}(\H_{S_1},X \cup (S \setminus S_1))$, where $\mathcal{L}$ $= \{S_1 \subseteq S \mid$ $\mbox{$S_1$ is an \is of $\H^{\bar{X}}$}\}$.
    \end{itemize}
  \end{itemize}
\end{definition}

\noindent In order to analyze the algorithm, given an input $(\H,X)$ of $A^{\beta}$, we define the measure
$$
  m(\H,X) =
    \begin{cases}
      \max\{|H| \mid H \in E(\H^{\bar{X}})\} & \text{, if } E(\H^{\bar{X}}) \neq \emptyset\\
      0 & \text{, otherwise.}
    \end{cases}
$$
Observe that, using the notation of Definition~\ref{def:improved-algo}, as $S$ is a hitting set of $\H^{\bar{X}}$, for every $S_1 \in \mathcal{L}$ we have $m(\H_{S_1},X \cup (S \setminus S_1))<m(\H,X)$.
Indeed, a hyperedge of $\H^{\bar{X}}$ either intersects $S \setminus S_1$ and is not taken into account in the `$\max$', or intersects $S_1$ and thus its corresponding hyperedge in the hypergraph $\H_{S_1}$  has smaller size. Observe also that as $\H$ does not contain an empty hyperedge, $m(\H,X)=0$ is equivalent to $E(\H^{\bar{X}}) = \emptyset$.

\begin{lemma}\label{lem:two-properties-algo}
  The following statement hold:
  \begin{enumerate}
  \item If $A^{\beta}(\H,X)$ returns ``\yes'' then $\mmhs(\H) \ge \beta$.
  \item If there exists a minimal hitting set $S^*$ of $\H$ such that $X \subseteq S^*$ and $|S^*| \ge \beta$, then $A^{\beta}(\H,X)$ returns ``\yes''.
  \end{enumerate}
The above properties imply that, given  an instance $(\H,\beta)$ of \MMHS, $A^{\beta}(\H,\emptyset)$ returns ``\yes'' if and only if $\mmhs(\H) \ge \beta$.
\end{lemma}
\begin{proof}
We use the notation introduced in Definition~\ref{def:improved-algo}.
Let us first argue that the prerequisite that the input hypergraph $\H$ does not contain an empty hyperedge is always satisfied.
To that end, let us consider such an input $\H$, and we shall prove that for any $S_1 \in \mathcal{L}$, $\H_{S_1}$ does not contain an empty hyperedge either.
Observe that as $S$ is minimal and hyperedges of $\H^{\bar{X}}$ do not intersect $X$, we have $S \cap X = \emptyset$, implying $S_1 \cap X = \emptyset$.
Together with the fact that $S_1$ is an \is of $\H^{\bar{X}}$, this implies that $S_1$ is an \is of $\H$.
Thus, $\H_{S_1}$ does not contain an empty hyperedge.

We prove both properties by induction on $m(\H,X)$.
  Let us start with the first property.
  If $m(\H,X)=0$, then $E(\H^{\bar{X}}) = \emptyset$, and the claimed property is true. Let is now assume that $m(\H,X) > 0$.
Suppose that $A^{\beta}(\H,X)$ returns ``\yes''. As $E(\H^{\bar{X}}) \neq \emptyset$, the algorithm goes to the second case and chooses $S$.
  If $|S| \ge \beta$, then by Lemma~\ref{lemma:fptv2:2} we get that $\mmhs(\H) \ge |S| \ge \beta$.
  Otherwise, there exists an \is $S_1$ of $\H^{\bar{X}}$ such that $A^{\beta}(\H_{S_1},X \cup (S \setminus S_1))$ returns ``\yes''.
As $m(\H_{S_1},X \cup (S \setminus S_1))<m(\H,X)$, by induction we get $\mmhs(\H_{S_1}) \ge \beta$, implying by Lemma~\ref{lemma:fptv2:1} that $\mmhs(\H) \ge \beta$.

  Let us now turn to the second property, and assume that there exists a minimal hitting set $S^*$ of $\H$ such that $X \subseteq S^*$ and $|S^*| \ge \beta$.
  Suppose first that $m(\H,X)=0$, implying that $E(\H^{\bar{X}}) = \emptyset$.
  In this case, we have that $X$ is already a hitting set of $\H$, and thus, as $S^*$ is minimal and $X \subseteq S^*$, we get that $S^*=X$,
  implying that $|X| \ge \beta$ and that the algorithm returns ``\yes''.
  Suppose now that $m(\H,X) > 0$, implying $E(\H^{\bar{X}}) \neq \emptyset$, and thus that the algorithm goes to the second case and chooses $S$.
If $|S| \ge \beta$ then we are done.
  Otherwise, let $S_2^* = S \cap S^*$ and $S_1^* = S \setminus S^*$. As $S^*$ is a hitting set of $\H$, there is no $H \in E(\H)$ such
  that $H \subseteq S_1^*$, implying that $S_1^*$ is an \is and thus that $S_1^* \in \mathcal{L}$.
  As by Lemma~\ref{lemma:fptv2:1}, $S^*$ is also a minimal hitting set of $\H_{S_1^*}$, as $ X \cup (S \setminus S_1^*) \subseteq S^*$,
  and as $m(\H_{S_1^*},X \cup (S \setminus S_1^*))< m(\H,X)$, by induction we get that $A^{\beta}(\H_{S_1^*},X \cup (S \setminus S_1^*))$ returns ``\yes'', and thus that $A(\H,X)$ returns ``\yes'' as well.
\end{proof}

\begin{theorem}\label{thm:fptv2}
\MMHS/$(\alpha+\beta)$ can be solved in time $\O^*(2^{\alpha\beta})$.
\end{theorem}
\begin{proof}
Given an instance $(\H,\beta)$ of \MMHS, we simply call $A^{\beta}(\H,\emptyset)$.
According to Lemma~\ref{lem:two-properties-algo}, this algorithm correctly decides whether $\mmhs(\H) \ge \beta$. Let us now analyze the running time.
  Let $f(\beta,\alpha,n)$ we the worst case running time of the algorithm $A^\beta(\H,X)$ when $m(\H,X) \le \alpha$ and $|V(\H)|=n$.
  We get that there exists a polynomial $p$ such that $f(\beta,0,n) \le p(n)$ (as when $m(\H,X) = 0$ we have $E(\H^{\bar{X}}) = \emptyset$ and
  the algorithm only checks that $X$ is a minimal hitting set of size at least $\beta$),
  and $f(\beta,\alpha,n) \le p(n)+2^{\beta-1}f(\beta,\alpha-1,n)$.
  To simplify the notation, let $b = 2^{\beta-1}$.
  Using a straightforward induction on $\alpha$ it follows that $f(\beta,\alpha,n) \le p(n)\cdot \left(b^{\alpha-1}+\frac{b^{\alpha-1}-1}{b-1}\right)$, implying the claimed running time.
  \end{proof}



\newpage
\section{\MMBS parameterized by treewidth}
\label{sec:tw}

\label{sec:tw}




\newcommand{\pbmmbstw}{\textsc{MMBS/\tw}\xspace}
In this section we prove that \pbmmbstw is \FPT.
The algorithm requires a long case analysis. As one may expect, we present a dynamic programming (DP) algorithm using a nice tree decomposition of the input graph $G$. In Section~\ref{subsec:prelim_tw}, we present the notation we need and we provide some intuition about the parameters that we store in the tables of the algorithm. In Sections~\ref{sec:join},~\ref{sec:intro}, and~\ref{sec:forget} we present how to compute the table entries for a join, introduce, and forget node, respectively.
In Section~\ref{sec:together} we combine the previous ingredients to complete the algorithm.

  \subsection{Preliminaries}
  \label{subsec:prelim_tw}


  Consider a graph $G$ and subsets $X \subseteq V(G)$, $\Y \subseteq V(G)$, and $Z \subseteq X$ such that $Z$ is an \is of $G$.
  We say that a set $I \subseteq V(G)$ is an \emph{$(X,Z)$-\is} if $I$ is an \is of $G$ such that $I \cap X = Z$.
  We denote by
  \begin{itemize}
    \item \emph{$\alpha_{(X,Z)}(G)$} the size of a largest $(X,Z)$-\is in $G$, and by
    \item \emph{$\alpha_{(X,Z)}^\Y(G)$} the size of a largest $(X,Z)$-\is $I$ in $G$ such that $I \cap \Y = \emptyset$.
  \end{itemize}
  In both cases, if such a set does not exist, we set the corresponding parameter to $-\infty$. We say that a set $\Y \subseteq V(G)$ is an \emph{$(X,Z)$-\bs} in $G$ if $\alpha_{(X,Z)}^\Y(G) < \alpha_{(X,Z)}(G)$, and we say that $(X,Z)$ is \emph{blocked} by $\Y$ in $G$.
  These concepts are illustrated in Figure~\ref{fig:Z_is}.
  Observe that if $\Y \cap Z \neq \emptyset$ then $\Y$ is an $(X,Z)$-\bs
  (as $\alpha_{(X,Z)}^\Y(G) = -\infty$), but the backward implication is not necessarily true
  as $\Y$ may contain one vertex in $X\setminus Z$ of each maximum $(X,Z)$-is of $G$.
  Observe also that an $(\emptyset,\emptyset)$-\is of $G$ is simply an \is of $G$, implying that $\alpha_{(\emptyset,\emptyset)}(G)=\alpha(G)$. Similarly, an $(\emptyset,\emptyset)$-\bs is a \bs of $G$.
  \begin{center}
    \begin{figure}[!h]
      \begin{center}
        \includegraphics[width=0.24\textwidth]{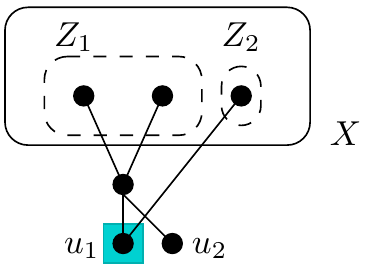}
      \end{center}
      \caption{In this example there is only one maximum $(X,Z_1)$-\is which is $I = Z_1 \cup \{u_1,u_2\}$. Note that $\Y = \{u_1\}$ is an $(X,Z_1)$-\bs, but $\Y$ is not
      an $(X,Z_2)$-\bs.}
      \label{fig:Z_is}
    \end{figure}
  \end{center}

  In what follows we assume that we are given a nice tree decomposition $\D=(T,{\cal B})$ of the input graph $G$
  as defined in Section~\ref{sec:prelim}.
  In particular, recall that
  \begin{itemize}
    \item every node of $T$ has at most two children,
    \item if a bag $X$ corresponds to a node of $T$ having two children with bags $X^L$ and $X^R$, then $X=X^L=X^R$ (the node corresponding to $X$ is a \emph{join} node);
    \item if a bag $X$ corresponds to a node of $T$ having one children with bag $X^C$, then
    \begin{itemize}
      \item either $X \subsetneq X^C$ and $|X^C|=|X|+1$ (the node corresponding to $X$ is a \emph{forget} node), or
      \item $X^C \subsetneq X$ and $|X|=|X^C|+1$ (the node corresponding to $X$ is an \emph{introduce} node).
    \end{itemize}
  \end{itemize}


  \paragraph*{Discussion on the list of parameters used in the DP algorithm}
  As usual, our dynamic programming algorithm performs a leaf-to-root traversal of a nice tree decomposition $\D=(T,{\cal B})$ of an input graph $G$ computing, for each node of the corresponding tree $T$, a set of tuples from the corresponding tuples of its children. Let us first explain the intuition behind each parameter of such tuples we shall compute and why they are needed. The formal details are presented in  Definition~\ref{def:tchack} (page~\pageref{def:tchack}).

To simplify the presentation, we call a bag of the tree decomposition \emph{join bag} (resp. \emph{forget bag}, \emph{introduce bag}) its corresponding node is a join node (resp. forget node, introduce node). We also speak about the \emph{children} of a bag, meaning the bags corresponding to the children of the considered node.

  Consider a join bag $X$ with children $X^L=X^R=X$, and suppose we look for a maximum \mbs $\Y$ of $G_X$.
  First, finding separately a maximum \mbs $\Y^L$ in $G_{X^L}$ and $\Y^R$ in $G_{X^R}$ will not guarantee that the size of $\Y^L \cup \Y^R$ (assuming
  $\Y^L \cup \Y^R$ is an \mbs in $G_X$) is maximum, and thus we introduce a parameter $\Y_0 \subseteq X$ and look for an \mbs $\Y$ of the graph $G_X$ such that $\Y \cap X = \Y_0$.

  Let $I$ be a \mis of $G_X$, $I^L = I \cap V(G_{X^L})$, and $I^R = I \cap V(G_{X^R})$. Observe that $I^L$ (resp. $I^R$)
  is not necessarily a \mis of $G_{X^L}$ (resp. $G_{X^R}$), and thus it is pointless to find an \mbs $\Y^L$ (resp. $\Y^R$) in $G_{X^L}$ (resp. $G_{X^R}$),
  as blocking maximum independent sets of $G_{X^L}$  and $G_{X^R}$ may not imply that we block maximum independent sets of $G_{X}$.
  This motivates the above notion of $(X,Z)$-\is in $G_{X}$.
  More precisely, let $\L = \{Z \subseteq X \mid$ there exists a \mis $I$ of $G_{X}$ such that $I \cap X = Z \}$
  .
  Then, $\Y$ is an \mbs of $G_X$ if and only if:
  \begin{enumerate}
    \item\label{ex:1} (blocking condition) for every $Z \in \L$, $\Y$ is an $(X,Z)$-\bs in $G_{X}$, and
    \item\label{ex:2} (minimality condition) for every $v \in \Y$, there must exist $Z \in \L$ such that $\Y \setminus \{v\}$ is not an $(X,Z)$-\bs in $G_{X}$.
  \end{enumerate}
  \mbvfinale{can someone check it is really true, even with the weird case where a $(X,Z)$-\bs $\Y$ contains a vertex of $Z$ ?}
  This explains why we have, in our list of parameters of our dynamic programming (and also the input of our auxiliary problem in Definition~\ref{def:auxpb}), a list $\L$ of subsets of $X$,
  in addition to the set $\Y_0$. Notice that there may exist $Z \in \L$ with $Z = \emptyset$.
  Toward the correct notion of the operator `$\vdash$' given in Definition~\ref{def:tchack}, let us introduce some intermediate ones that we denote by `$\vdash_0$', `$\vdash_1$', `$\vdash_2$' and whose scope is limited to this preliminary discussion (as they
  will not be used in the eventual DP algorithm).
  Given $(X,\Y_0,\L)$ and a set $\Y$, we say that $\Y \vdash_0 (X,\Y_0,\L)$ if and only if
  \begin{itemize}
    \item $\Y \cap X = \Y_0$, and
    \item $\Y$ satisfies the two properties above (blocking and minimality conditions).
  \end{itemize}
  Such a set $\Y$ will be called a \emph{solution} to $(X,\Y_0,\L)$ (instead of \bs).

  Let us now argue that these three parameters are still not sufficient to design our algorithm, by exhibiting two situations where we suppose that we computed ``small solutions'',
  but extending these small solutions to the current bag creates a solution which no longer respects the minimality condition.
  \begin{center}
    \begin{figure}[!h]
      \begin{center}
        \includegraphics[width=0.29\textwidth]{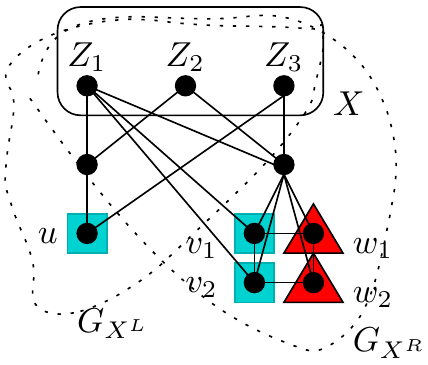}
      \end{center}
      \caption{$\Y \vdash_0 (X,\emptyset,\L)$ where $\L = \{Z_1,Z_2,Z_3\}$, $\Y=\{u,v_1,v_2\}$,
      and $\Y'=\{u,w_1,w_2\}$.}
      \label{fig:intuition_DP_S}
    \end{figure}
  \end{center}


  Let us start with the first situation.
  Suppose first that we have a solution $\Y \vdash_0 (X,\Y_0,\L)$ for some $\L = \{Z_1,Z_2,Z_3\}$, as depicted in Figure~\ref{fig:intuition_DP_S}.
  Recall that $X=X^L=X^R$ and let $\Y^L = \Y \cap V(G_{X^L})$ and $\Y^R = \Y \cap V(G_{X^R})$.
  We prove in Lemma~\ref{prop:p2} (page~\pageref{prop:p2}) that, for every $Z \in \L$, $\Y$ is an $(X,Z)$-\bs in $G_X$ if and only if $\Y^L$ is an $(X^L,Z)$-\bs in $G_{X^L}$
  or $\Y^R$ is an $(X^R,Z)$-\bs in $G_{X^R}$. Thus, it may be the case, as in Figure~\ref{fig:intuition_DP_S}, that $\Y^L=\{u\}$ is an $(X^L,Z)$-\bs in $G_{X^L}$ for $Z \in \L^L = \{Z_1,Z_2\}$,
  and $\Y^R=\{v_1,v_2\}$ is an $(X^R,Z)$-\bs in $G_{X^R}$ for $Z \in \L^R = \{Z_2,Z_3\}$. Suppose now that we compute $\Y^{'L}$ and
  $\Y^{'R}$ such that $\Y^{'L} \vdash_0 (X^L,\Y_0,\L^L)$ and $\Y^{'R} \vdash_0 (X^R,\Y_0,\L^R)$, and let $\Y' = \Y^{'L} \cup \Y^{'R}$.
  It may be the case that $\Y'$ does not verify the previous minimality condition~\ref{ex:2}.
  Indeed, let $u \in \Y^{'L} \setminus X$ and suppose that $\Y^{'L} \setminus \{u\}$ is not an $(X^L,Z_1)$-\bs in $G_{X^L}$.
  Unfortunately, if $\Y^{'R}$ is an $(X^R,Z_1)$-\bs in $G_{X^R}$ (even if $Z_1 \notin \L^R$), we will have that $\Y' \setminus \{u\}$ is still
  a $(X,Z_1)$-\bs in $G_X$, and thus maybe still an $(X,Z)$-\bs for any $Z \in \L$. We overcome this problem by  forcing $\Y^{'R}$ {\sl not} to be an $(X^R,Z_1)$-\bs in $G_{X^R}$.
  This explains why we have in the input a list $\S$ of subsets of $X$, and we now impose that for any $Z \in \S$, $\Y$ must not be an $(X,Z)$-\bs in $G_X$.



  Thus, now we denote by $\Y \vdash_1 (X,\Y_0,\L,\S)$ the property that $\Y\vdash_0 (X,\Y_0,\L)$ and, for any $Z \in \S$, $\Y$ is not an $(X,Z)$-\bs in $G_X$.

  \begin{center}
    \begin{figure}[!h]
      \begin{center}
        \includegraphics[width=0.35\textwidth]{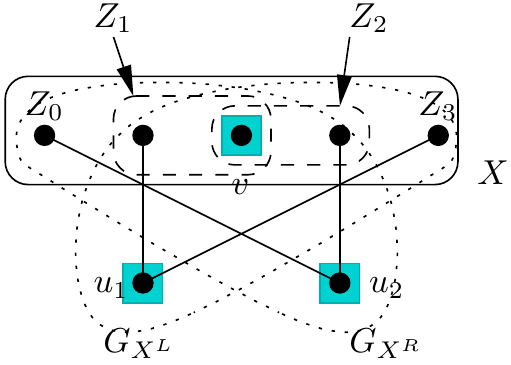}
      \end{center}
      \caption{$\Y_0 = \{v\}$, $\Y^{'L}=\{u_1,v\}$, and $\Y^{'R}=\{u_2,v\}$.}
      \label{fig:intuition_DP_f}
    \end{figure}
  \end{center}

  Let us now turn to the second situation, which is depicted Figure~\ref{fig:intuition_DP_f}, where $\L = \{Z_0,Z_1,Z_2,Z_3\}$ and $\Y_0 = \{v\}$.
  Suppose that we compute $\Y^{'L}$ and $\Y^{'R}$ such that $\Y^{'L} \vdash_1 (X^L,\Y_0,\L^L,\S^L)$ where $\L^L = \{Z_0,Z_1,Z_2\}$, $\S^L = \{Z_3\}$ and
  $\Y^{'R} \vdash_1 (X^R,\Y_0,\L^R,\S^R)$ where $\L^R = \{Z_1,Z_2,Z_3\}$ and $\S^R = \{Z_0\}$. Let $\Y' = \Y^{'L} \cup \Y^{'R}$.
  Let $v \in \Y_0$. By minimality condition~\ref{ex:2}, there exists $Z \in \L^L$ such that $\Y^{'L} \setminus \{v\}$ is not an $(X^L,Z)$-\bs in $G_{X^L}$
  (where $Z = Z_1$ in Figure~\ref{fig:intuition_DP_f}). In the same way, there exists $Z' \in \L^R$ such that $\Y^{'R} \setminus \{v\}$ is not an $(X^R,Z')$-\bs in $G_{X^R}$
  (where $Z'= Z_2$ in Figure~\ref{fig:intuition_DP_f}).
  If $Z \neq Z'$, we may not be able to conclude that there exists a $Z'' \in \L$ such that $\Y' \setminus \{v\}$ is not an $(X,Z'')$-\bs in $G_{X}$.
  In the example depicted in Figure~\ref{fig:intuition_DP_f}, $\Y' \setminus \{v\}$ is still an $(X,Z)$-\bs in $G_{X}$ for every $Z \in \L$.
  Thus, for $v \in \Y_0$, we will keep control of the minimality condition in a more precise way by
  \begin{itemize}
    \item[$\bullet$]introducing another list $\L_2$ of subsets of $X$, and still ask that $\Y$ is an $(X,Z)$-\bs in $G_X$ for any $Z \in \L_2$,
    \item[$\bullet$]introducing a function $f: \Y_0 \to \L_2$, and
    \item[$\bullet$] (minimality condition in $\Y_0$) requiring that for every $v \in \Y_0$, $\Y \setminus \{v\}$ is not an $(X,f(v))$-\bs in $G_X$.
  \end{itemize}
  In the previous example, we would have to set either $f^L(v)=f^R(v)=Z_1$ or $f^L(v)=f^R(v)=Z_2$, but none of these choices leads to a feasible solution on both the left and the right hand sides. This is not surprising, as in fact there is no set $\Y$ such that $\Y \vdash_1 (X,\Y_0,\L,\S)$ where $\L = \{Z_0,Z_1,Z_2,Z_3\}$, $\Y_0 = \{v\}$, and $\S = \emptyset$.
  Indeed, being an $(X,Z_0)$-\bs in $G_{X}$ forces any $\Y$ to contain $u_1$ (as $\Y$ cannot contain the vertex of $Z_0$), and
  being a $(X,Z_3)$-\bs in $G_{X}$ forces any $\Y$ to contain $u_2$. This means that we necessarily have $\{u_1,u_2,v\} \subseteq \Y$.
  Then, observe that $\Y$ is not minimal as $\{u_1,u_2\}$ is still a $(X,Z)$-\bs in $G_{X}$ for any $Z \in \L$. The conclusion is that in this situation, forcing $v$ to be in any solution
  leads to an infeasible instance.

  Finally, even when using function $f$, and defining accordingly $\Y \vdash_2 (X,\Y_0,\L_1,\L_2,f,\S)$ if $\Y \vdash_1 (X,\Y_0,\L_1,\S)$ and $Y$ respects
  the previous minimality condition in $\Y_0$, there is a last important detail. Suppose
  $\Y^{'L} \vdash_2 (X^L,\Y_0,\L_1^L,\L_2^L,f^L,\S^L)$ where for example that $\L_1^L = \{Z_1\}$, $\L_2^L = \{Z_2\}$, and consider $v \in \Y^{'L} \setminus \Y_0$.
  We know that there exists $Z$ such that $\Y^{'L} \setminus \{v\}$ is not an $(X^L,Z)$-\bs in $G_{X^L}$, but we must even impose that
  $Z \in \L_1^L$, as otherwise if $Z = Z_2$ then $\Y' \setminus \{v\}$ would still be an $(X,Z)$-\bs in $G_X$.
  Thus, the minimality condition is finally as follows:
  \begin{enumerate}
    \item (minimality condition outside $\Y_0$, forcing $Z \in \L_1$) $\forall v \in \Y \setminus \Y_0$, $\exists Z \in \L_1$ such that $\Y \setminus \{v\}$ is not an $(X,Z)$-\bs in $G_X$.
    \item (minimality condition in $\Y_0$) $\forall v \in \Y_0$, $\Y \setminus \{v\}$ is not an $(X,f(v))$-\bs in $G_X$, where $f(v) \in \L_2$.
  \end{enumerate}

  Even if we only discussed here the case where $X$ is a join node, it turns out that this list of parameters is also enough for the introduce and forget nodes.

  \paragraph*{Defining the auxiliary problem}
  Let us now define the auxiliary problem that will be solved by our DP algorithm.

  \begin{definition}
    Let $G$ be a graph and let $\D=(T,\B)$ be a nice tree decomposition of $G$.
    Let $\E(G,\D)$ be the set containing all tuples  $(X,\Y_0,\L_1,\L_2,f,\S)$ such that:
    \begin{itemize}
      \item[$\bullet$] $X \in \B$,
      \item[$\bullet$] $\Y_0 \subseteq X$,
      \item[$\bullet$] $\L_1,\L_2,\S \subseteq 2^X$
      such that for every $Z \in \L_1 \cup \L_2 \cup \S$, $Z$ is an \is of $G$, and
      \item[$\bullet$] $f: \Y_0 \to \L_2$. 
    \end{itemize}
    \end{definition}

\begin{definition}\label{def:tchack}
Let $G$ be a graph and let $\D=(T,\B)$ be a nice tree decomposition of $G$.
For every $(X,\Y_0,\L_1,\L_2,f,\S) \in \E(G,\D)$ and $\Y \subseteq V(G_X)$, we write $\Y \vdash (X,\Y_0,\L_1,\L_2,f,\S)$ if and only if
    \begin{enumerate}[i)]
      \item\label{def:Y0} $\Y \cap X = \Y_0$,
      \item\label{def:L1L2} $\forall Z \in \L_1 \cup \L_2$, $\Y$ is an $(X,Z)$-\bs in $G_X$,
      \item\label{def:S} $\forall Z \in \S$, $\Y$ is not an $(X,Z)$-\bs in $G_X$,
      \item\label{def:min} and the following two minimality conditions are satisfied:
      \begin{enumerate}[a)]
        \item\label{def:mina} $\forall v \in \Y \setminus \Y_0$, $\exists Z \in \L_1$ such that $\Y \setminus \{v\}$ is not an $(X,Z)$-\bs in $G_X$, and
        \item\label{def:minb} $\forall v \in \Y_0$, $\Y \setminus \{v\}$ is not an $(X,f(v))$-\bs in $G_X$.
      \end{enumerate}
    \end{enumerate}
  \end{definition}


      Let us point out that there may exist $Z \in \L_1 \cup \L_2 \cup \S$ with $Z = \emptyset$,
      and that if $\exists Z \in \S$ such that $\Y_0 \cap Z \neq \emptyset$, then there is no solution (because of Property~\ref{def:S}).

      \begin{definition}\label{def:auxpb}
        We define the optimization problem $\Pi$ as follows, where we consider that the input graph $G$ and a nice tree decomposition $\D$ of $G$ are fixed:
        \begin{itemize}
          \item[] \emph{\textbf{Input}: A tuple $(X,\Y_0,\L_1,\L_2,f,\S) \in \E(G,\D)$.}
          \item[] \emph{\textbf{Output}: A set $\Y \subseteq V(G_X)$ such that $\Y \vdash (X,\Y_0,\L_1,\L_2,f,\S)$.}
          \item[] \emph{\textbf{Objective}: Maximize $|\Y|$.}
        \end{itemize}
      \end{definition}

      We say that an instance $I$ of $\Pi$ is \emph{feasible} if there exists a set $\Y$ such that $\Y \vdash (X,\Y_0,\L_1,\L_2,f,\S)$.
      Let us now show that being able to solve optimally problem $\Pi$ is sufficient for computing the parameter $\mmbs(G)$, for a given graph $G$.

      \begin{proposition}\label{prop:mbsVspb1}
        Let $G$ be a graph and $\D=(T,\B)$ be a nice tree decomposition of $G$ such that $T$ is rooted at $X_0=\{\emptyset\}$.
        For every $\Y \subseteq V(G)$,
        $$\Y \vdash (\emptyset,\emptyset,\{\emptyset\},\emptyset,\emptyset,\emptyset) \ \text{ if and only if }\ \Y \mbox{ is an \mbs of $G$}.$$ \label{prop:mbsVspb1}
       \vspace{-.6cm}
      \end{proposition}
      \begin{proof}
            Let $(X,\Y_0,\L_1,\L_2,f,\S) = (\emptyset,\emptyset,\{\emptyset\},\emptyset,\emptyset,\emptyset)$.
            Recall that being an $(\emptyset,\emptyset)$-\bs  in $G_{X_0}$ is equivalent to being a $\bs$ in $G$.

            Suppose first that $\Y \vdash (\emptyset,\emptyset,\{\emptyset\},\emptyset,\emptyset,\emptyset)$.
            By Property~\ref{def:L1L2}, $\Y$ is an $(\emptyset,\emptyset)$-\bs in $G_{X_0}$, implying that $\Y$ is a $\bs$ of $G$.
            Let us now prove that $\Y$ is minimal. Let $v \in \Y$. As $v \in \Y \setminus \Y_0$, by Property~\ref{def:mina}, there exists $Z \in \L_1$ such that $\Y \setminus \{v\}$ is not an $(X_0,Z)$-\bs in $G$.
            As $\L_1 = \{\emptyset\}$, we obtain that $\Y \setminus \{v\}$ is not an $(\emptyset,\emptyset)$-\bs in $G$, and thus not a $\bs$ in $G$.

            Suppose now that $\Y$ is an \mbs of $G$.
            Property~\ref{def:L1L2} is satisfied as $\Y$ is a $\bs$ in $G$.
            Let us now prove Property~\ref{def:mina}. Let $v \in \Y \setminus \Y_0$. As $\Y$ is minimal,
            $\Y \setminus \{v\}$ is not a $\bs$ in $G$, and thus not an $(\emptyset,\emptyset)$-\bs  in $G_{X_0}=G$,
            where $\emptyset \in \L_1$.
          \end{proof}

          The following proposition is now immediate.

          \begin{proposition}\label{prop:solvingPiOK}
            Given an $n$-vertex graph $G$ with treewidth $\tw(G)=t$, if
            \begin{itemize}
              \item[$\bullet$] $t_1(n,t)$ is the time to compute a nice tree decomposition $\D$ of $G$ of width $t$, and
              \item[$\bullet$] $t_A(n,t)$ is the time to compute an optimal solution of problem $\Pi$, 
            \end{itemize}
            then one can compute $\mmbs(G)$ in time $\O(t_1(n,t)+ t_A(n,t))$.
          \end{proposition}


          In what follows, namely in Sections~\ref{sec:join},~\ref{sec:intro}, and~\ref{sec:forget}, we fix an input graph  $G$ and a nice tree decomposition $\D=(\B,T)$ of $G$ of width $t$.


          \subsection{Join node}
          \label{sec:join}

          Before proving Lemma~\ref{lem:join} corresponding to the join case, let us first prove the following technical lemmas.

          \begin{lemma}\label{prop:p1}
            For every $I \subseteq V(G_X)$ and every $Z \subseteq X$ where $Z$ is an \is of $G_X$, $I$ is a maximum $(X,Z)$-\is in $G_X$ if and only if $I^L=I \cap V(G_{X^L})$ is a maximum $(X^L,Z)$-\is in $G_{X^L}$ and $I^R=I \cap V(G_{X^R})$ is a maximum $(X^R,Z)$-\is in $G_{X^R}$.
          \end{lemma}
          \begin{proof}
            For the forward implication, suppose $I$ is a maximum $(X,Z)$-\is in $G_X$.
            Let $I_*^L$ be a maximum $(X^L,Z)$-\is in $G_{X^L}$. Then, $(I \setminus I^L) \cup I_*^L$ is still an \is as
            there is no edge between $I_*^L \cap X$ and $I^R$ (as $I_*^L \cap X = I^L \cap X = Z$),
            and there is no edge between $I_*^L \setminus X$ and $I^R \setminus X$ as, by the properties of a tree decomposition, there is no edge even between $V(G_{X_L}) \setminus X$ and $V(G_{X_R}) \setminus X$.
            As $((I \setminus I^L) \cup I_*^L) \cap X = Z$, this implies that $(I \setminus I^L) \cup I_*^L$ is an $(X,Z)$-\is in $G_X$ and that $|I| \ge (I \setminus I^L) \cup I_*^L$.
            As $I \cap I^L = I \cap I_*^L = Z$, we get $|I^L| \ge |I_*^L|$. As $I^L \cap X^L = Z$, we obtain that $I^L$ is a maximum
            $(X^L,Z)$-\is in $G_{X^L}$. The same arguments hold for $I^R$.

            For the backward implication, suppose $I^L$ is a maximum  $(X^L,Z)$-\is in $G_{X^L}$ and $I^R$ is a maximum $(X^R,Z)$-\is in $G_{X^R}$. Observe first that $I^L \cup I^R$
            is an $(X,Z)$-\is in $G_X$, as there is no edge between $V(G_{X_L}) \setminus X$ and $V(G_{X_R}) \setminus X$.
            Let $I^*$ be a maximum $(X,Z)$-\is in $G_X$. Note that $I_*^L := I^* \cap V(G_{X^L})$ is an $(X^L,Z)$-\is  in $G_{X^L}$, and symmetrically that $I_*^R := I^* \cap V(G_{X^R})$ is an $(X^R,Z)$-\is  in $G_{X^R}$. This implies that $|I_*^L| \le |I^L|$ and $|I_*^R| \le |I^R|$. As $I_*^L \cap I_*^R = I^{L} \cap I^{R} = Z$, the previous inequalities imply $|I^*| \le |I|$, meaning that $I$ is a maximum $(X,Z)$-\is in $G_X$.
          \end{proof}

          \mbvfinale{A boring problem with the notation. In a join node, we say that $X = X_L = X_R$, but it is not true that $G_X = G_{X^L} = G_{X^R}$. It may be confusing.}

          \begin{lemma}\label{prop:p2}
            Let $Z \subseteq X$. For every $\Y \subseteq V(G_X)$, $\Y$ is an $(X,Z)$-\bs in $G_X$ if and only if
            $\Y^L = \Y \cap V(G_{X^L})$ is an $(X^L,Z)$-\bs in $G_{X^L}$ or $\Y^R = \Y \cap V(G_{X^R})$ is an $(X^R,Z)$-\bs in $G_{X^R}$.
          \end{lemma}
          \begin{proof}
            For the forward implication, suppose $\Y$ is an $(X,Z)$-\bs in $G_X$.
            Suppose by contradiction that there exists a maximum $(X^L,Z)$-\is $I^L$ in $G_{X^L}$ such that $I^L \cap \Y^L = \emptyset$,
            and a maximum $(X^R,Z)$-\is $I^R$ in $G_{X^R}$ such that $I^R \cap \Y^R = \emptyset$. Let $I = I^L \cup I^R$.
            By Lemma~\ref{prop:p1}, $I$ is a maximum $(X,Z)$-\is in $G_X$. As $I^L \cap \Y = I^L \cap \Y^L = \emptyset$, and also $I^R \cap \Y = I^R \cap \Y^R = \emptyset$, we get $I \cap (\Y^L \cup \Y^R) = I \cap \Y=\emptyset$, a contradiction to the hypothesis that $\Y$ is an $(X,Z)$-\bs in $G_X$.

            For the backward implication, suppose $\Y^L$ is an $(X^L,Z)$-\bs in $G_{X^L}$ or $\Y^R$ is an $(X^R,Z)$-\bs in $G_{X^R}$.
            Suppose by contradiction that there exists a maximum $(X,Z)$-\is $I$ in $G_X$ such that $I \cap \Y = \emptyset$.
            By Lemma~\ref{prop:p1}, $I^L=I \cap V(G_{X^L})$ is a maximum $(X^L,Z)$-\is in $G_{X^L}$ and $I^R=I \cap V(G_{X^R})$ is a maximum $(X^R,Z)$-\is in $G_{X^R}$.
            As $I^L \cap \Y = I^R \cap \Y = \emptyset$, we obtain that $\Y^L$ is not an $(X^L,Z)$-\bs in $G_{X^L}$ and that $\Y^R$ is not an $(X^R,Z)$-\bs in $G_{X^R}$, a contradiction.
          \end{proof}

          We are now ready to state the main lemma of this section.

          \begin{lemma}\label{lem:join}
            Let $(X,\Y_0,\L_1,\L_2,f,\S) \in \E(G,\D)$ where $X \in \B$ is a join node and $X^L,X^R$ are the children of $X$ (with $X=X^L=X^R$
            ).
            For every $\Y \subseteq V(G_X)$, it holds that
            $\Y \vdash (X,\Y_0,\L_1,\L_2,f,\S)$
            if and only if \fixmeperso{avanton avait either $\Y=\Y_0$ or}
            there exist sets $\Y^L, \Y^R, \L_1^A, \L_1^B, \L_1^C,\L_2^A, \L_2^B, \L_2^C$ such that the following properties hold:

            \begin{enumerate}
              \item \label{join:1} $\Y = \Y^L \cup \Y^R$, 


              \item \label{join:2} $\L_1 = \L_1^A \uplus \L_1^B \uplus \L_1^C$ and $\L_2=\L_2^A \uplus \L_2^B \uplus \L_2^C$,

              \item \label{join:2bis} for every $v \in B_0$, $f(v) \in \L_2^B$,




              \item \label{join:3} $\Y^L \vdash (X^L,\Y_0,\L_1^L,\L_2^L,f^L,\S^L)$, where
              \begin{itemize}
                \item[$\bullet$]$\L_1^L = \L_1^A$,
                \item[$\bullet$]$\L_2^L = \L_1^B \cup \L_2^B\cup \L_2^A$,
                \item[$\bullet$]$f^L = f$, and
                \item[$\bullet$]$\S^L = \S \cup \L_1^C \cup \L_2^C$;
              \end{itemize}
              \medskip
              and $\Y^R \vdash (X^R,\Y_0,\L_1^R,\L_2^R,f^R,\S^R)$, where
              \medskip
              \begin{itemize}
                \item[$\bullet$]$\L_1^R = \L_1^C$,
                \item[$\bullet$]$\L_2^R = \L_1^B\cup \L_2^B \cup \L_2^C$,
                \item[$\bullet$]$f^R =  f$, and
                \item[$\bullet$]$\S^R =  \S \cup \L_1^A \cup \L_2^A$.
              \end{itemize}

            \end{enumerate}

          \end{lemma}



          \begin{proof}
            For the forward implication, suppose first that $\Y \vdash (X,\Y_0,\L_1,\L_2,f,\S)$.
            Let $\Y^L = \Y \cap V(G_{X^L})$ and $\Y^R = \Y \cap V(G_{X^R})$, satisfying Property~\ref{join:1} of the lemma.
            For $i \in [2]$, let

            \begin{itemize}
              \item[$\bullet$]$\L_i^A = \{Z \in \L_i \mid $  $\Y^L$ is an $(X^L,Z)$-\bs in $G_{X^L}$ and $\Y^R$ is not an $(X^R,Z)$-\bs in $G_{X^R} \}$,
              \item[$\bullet$]$\L_i^C = \{Z \in \L_i \mid $  $\Y^L$ is not an $(X^L,Z)$-\bs in $G_{X^R}$ and $\Y^R$ is an $(X^R,Z)$-\bs in $G_{X^R} \}$, and
              \item[$\bullet$]$\L_i^B = \{Z \in \L_i \mid $  $\Y^L$ is an $(X^L,Z)$-\bs in $G_{X^L}$ and $\Y^R$ is an $(X^R,Z)$-\bs in $G_{X^R} \}$.
            \end{itemize}

            By Definition~\ref{def:tchack}, for every $Z \in \L_i$, as $\Y$ is an $(X,Z)$-\bs in $G_X$. By Lemma~\ref{prop:p2}, we obtain that $\Y^L$ is an $(X^L,Z)$-\bs in $G_{X^L}$ or $\Y^R$ is an $(X^R,Z)$-\bs in $G_{X^R}$. This
            implies that $\L_i = \L_i^A \uplus \L_i^B \uplus \L_i^C$, and thus Property~\ref{join:2} is satisfied.

            For Property~\ref{join:3}, let us only prove that $\Y^L \vdash (X^L,\Y_0,\L_1^L,\L_2^L,f^L,\S^L)$, as the proof for $\Y^R$ follows the same arguments. We verify that each of the (non-trivial) properties of Definition~\ref{def:tchack} is satisfied.

            \emph{Property~\ref{def:L1L2}}.
            We need to prove that $\Y^L$ is an $(X^L,Z)$-\bs in $G_{X^L}$ for every $Z\in \L_1^L\cup \L_2^L$. This follows from the definition  of the sets $\L^L_i$,
            and by the hypothesis that $\Y$ is an $(X,Z)$-\bs in $G_X$ for every $Z\in \L_1\cup \L_2$ (since $\Y \vdash (X,\Y_0,\L_1,\L_2,f,\S)$).

            \emph{Property~\ref{def:S}}.
            Let us prove that $\Y^L$ is not an $(X^L,Z)$-\bs in $G_{X^L}$, for every $Z\in \S^L = \S\cup \L_1^C\cup \L_2^C$.
            Let $Z \in \S^L$. If $Z \in \S$, then since $\Y$ is not an $(X,Z)$-\bs in $G_X$, because $\Y \vdash (X,\Y_0,\L_1,\L_2,f,\S)$, we have by Lemma~\ref{prop:p2} that $\Y^L$ is
            not an $(X^L,Z)$-\bs in $G_{X^L}$. If $Z \in \L_1^C \cup \L_2^C$, then the result follows from definition of $\L_i^C$.

            \emph{Property~\ref{def:mina}}. We have to prove that $\forall v\in \Y^L\setminus \Y_0$, $\exists Z\in \L_1^L = \L_1^A$ such that $\Y^L\setminus\{v\}$ is not an $(X^L,Z)$-\bs in $G_{X^L}$. If $\Y^L = \Y_0$, the statement trivially holds. Otherwise, let
            $v \in \Y^L \setminus \Y_0$. As $\Y \vdash (X,\Y_0,\L_1,\L_2,f,\S)$, there exists $Z \in \L_1$
            such that $\Y \setminus \{v\}$ is not an $(X,Z)$-\bs in $G_X$, by Definition~\ref{def:tchack}. This implies,  by Lemma~\ref{prop:p2}, that $\Y^L \setminus \{v\}$ is not an $(X^L,Z)$-\bs in $G_{X^L}$.
            As $\Y \setminus \{v\} \supseteq \Y^R$, we get that $\Y^R$ is not an $(X,Z)$-\bs in $G_X$, and thus by Lemma~\ref{prop:p2} that $\Y^R$ is not an $(X^R,Z)$-\bs in $G_{X^R}$, implying that $Z \in \L_1^A$.

            \emph{Property~\ref{def:minb}}. We finally have to prove that $\forall v\in \Y_0$, $\Y^L\setminus\{v\}$ is not an $(X_L, f^L(v))$-\bs in $G_{X^L}$. Let $v \in \Y_0$. Let $Z = f^L(v)=f(v)$, where $Z \in \L_2$.
            As $\Y \vdash (X,\Y_0,\L_1,\L_2,f,\S)$, $\Y \setminus \{v\}$ is not an $(X,Z)$-\bs in $G_X$, implying by Lemma~\ref{prop:p2} that
            $\Y^L \setminus \{v\}$ is not an $(X^L,Z)$-\bs in $G_{X^L}$.
            Moreover, as $\Y$ is an $(X,Z)$-\bs in $G_X$ and $\Y \setminus \{v\}$ is not an $(X,Z)$-\bs in $G_X$, we deduce that $v \in Z$, implying that $Z \in \L_2^B \subseteq \L_2^L$. Note also that $f^L(v) \in  \L_2^B$, as required by Property~\ref{join:2bis}.



            \medskip

            For the backward implication, suppose  that there exist $\Y^L,\Y^R,\L_1^A, \L_1^B, \L_1^C,\L_2^A, \L_2^B, \L_2^C$ satisfying the lemma's conditions.
            Let us prove that $\Y \vdash (X,\Y_0,\L_1,\L_2,f,\S)$, by verifying again that each of the (non-trivial) properties of Definition~\ref{def:tchack} is satisfied.

            \emph{Property~\ref{def:L1L2}}. 
            We have to prove that $\Y=\Y^L\cup \Y^R$ is an $(X,Z)$-\bs in $G_X$, for every $Z\in \L_1\cup\L_2$. By hypothesis, we know that: $\Y^L \vdash (X^L,\Y_0,\L_1^L,\L_2^L,f^L,\S^L)$ and $\Y^R \vdash (X^R,\Y_0,\L_1^R,\L_2^R,f^R,\S^R)$. By Definition~\ref{def:tchack}, we deduce that $\Y^L$ is an $(X^L,Z)$-\bs in $G_{X^L}$, for every $Z\in \L_1^L\cup \L_2^L = \L_1^A\cup \L_2^A\cup \L_1^B\cup \L_2^B$. By Lemma~\ref{prop:p2}, $\Y$ is an $(X,Z)$-\bs in $G_X$, for every $Z\in \L_1^A\cup \L_2^A\cup \L_1^B\cup \L_2^B$. Analogously, one may deduce that $\Y$ is an $(X,Z)$-\bs in $G_X$, for every $Z\in \L_1^C\cup \L_2^C\cup \L_1^B\cup \L_2^B$. Thus, $\Y$ is an $(X,Z)$-\bs in $G_X$, for every $Z\in \L_1\cup \L_2$.

            \emph{Property~\ref{def:S}}.
            We have to prove that $\Y$ is not an $(X,Z)$-\bs in $G_X$, for every $Z\in \S$.
            Let $Z \in \S$. Since $\Y^L \vdash (X^L,\Y_0,\L_1^L,\L_2^L,f^L,\S^L)$ and $\S\subseteq \S^L$, we have that $\Y^L$ is not an $(X^L,Z)$-\bs in $G_{X^L}$, by Definition~\ref{def:tchack}. As $\Y^L$ is not an $(X^L,Z)$-\bs in $G_{X^L}$ and, analogously, $\Y^R$ is not an $(X^R,Z)$-\bs in $G_{X^R}$,
            it implies by Lemma~\ref{prop:p2} that $\Y$ is not an $(X,Z)$-\bs in $G_{X}$.

            \emph{Property~\ref{def:mina}}. Let us now prove that for every $v\in \Y\setminus\Y_0$, there is $Z\in \L_1=\L_1^A\uplus\L_1^B\uplus\L_1^C$ such that $\Y\setminus\{v\}$ is not an $(X,Z)$-\bs in $G_X$. If $\Y = \Y_0$, then there is nothing to prove. Otherwise, let $v \in \Y \setminus \Y_0$, and
            suppose without loss of generality that $v \in \Y^L \setminus \Y_0$.
            As  $\Y^L \vdash (X^L,\Y_0,\L_1^L,\L_2^L,f^L,\S^L)$, and as $\L_1^L = \L_1^A$, there exists $Z \in \L_1^A$  such that $\Y^L \setminus \{v\}$ is not an $(X^L,Z)$-\bs in $G_{X^L}$.
            As $Z \in \L_1^A$, $L_1^A \subseteq \S^R$, and $\Y^R \vdash (X^R,\Y_0,\L_1^R,\L_2^R,f^R,\S^R)$,  $\Y^R$ is not an $(X^R,Z)$-\bs in $G_{X^R}$.
            Thus, by Lemma~\ref{prop:p2}, $\Y \setminus \{v\}$ is not an $(X,Z)$-\bs in $G_X$.

            \emph{Property~\ref{def:minb}}. Let us finally prove that for each $v\in \Y_0$, $\Y\setminus\{v\}$ is not an $(X,f(v))$-\bs in $G_X$. Let $v \in \Y_0$ and let $Z = f(v)$. By Property~\ref{join:2bis}, we know that $Z \in \L_2^B$, i.e. $Z$ is both in $\L_2^L$ and $\L_2^R$. Then, as $\Y^L \vdash (X^L,\Y_0,\L_1^L,\L_2^L,f^L,\S^L)$ and $f^L = f$, by Property~\ref{def:minb} we get that $\Y^L \setminus \{v\}$ is not an $(X^L,Z)$-\bs in $G_{X^L}$. Using the same arguments for $\Y^R$, we get that $\Y^R \setminus \{v\}$ is not an $(X^R,Z)$-\bs in $G_{X^R}$.
            By Lemma~\ref{prop:p2}, we obtain that $\Y \setminus \{v\}$ is not an $(X,Z)$-\bs in $G_X$.
          \end{proof}


          \subsection{Introduce node}
          \label{sec:intro}

          \begin{definition}\label{def:critical}
            Let $G$ be a graph, $X \subseteq V(G)$, $v \in X$, and $\R \subseteq 2^X$.
            We denote
            \begin{itemize}
              \item[$\bullet$]$\R(v) = \{Z \in \R \mid v \in Z \}$,
              \item[$\bullet$]$\R(\bar{v}) = \{Z \in \R \mid v \notin Z\}$, and
              \item[$\bullet$]$r_v(\R) = \{Z \setminus \{v\} \mid Z \in \R\}$.
            \end{itemize}
          \end{definition}

                Before proving Lemma~\ref{lem:intro} corresponding to the introduce case, let us first prove the following lemmas
                where we assume that $X \in \B$ is an introduce node and that $X^C$ is the child of $X$ with $X^C=X \setminus \{v\}$ for some vertex $v \in X$.

          \begin{lemma}\label{lemma:intro:I1}
            Let $Z \subseteq X$ such that $Z$ is an \is with $v \in Z$. For every $I \subseteq V(G)$ such that $v \in I$, $I$ is a maximum $(X,Z)$-\is in $G_X$ if and only if $I \setminus \{v\}$ is a maximum $(X^C,Z \setminus \{v\})$-\is in $G_{X^C}$.
          \end{lemma}
          \begin{proof}
            For the forward implication, suppose that $I$ is a maximum $(X,Z)$-\is in $G_X$ such that $v\in I$. Note that $I \setminus \{v\}$ is an $(X^C,Z \setminus \{v\})$-\is in $G_{X^C}$.
            Let $I'$ be a maximum $(X^C,Z \setminus \{v\})$-\is in $G_{X^C}$.
            As $Z$ is an \is of $G$, and $N_{G_X}(v) \subseteq X$ by the properties of a tree decomposition, $I' \cup \{v\}$ is an $(X,Z)$-\is in $G_X$,
            implying $|I' \cup \{v\}| \le |I|$. Therefore $|I \setminus \{v\}| \geq |I'|$, hence $I \setminus \{v\}$ is a maximum $(X^C,Z \setminus \{v\})$-\is in $G_{X^C}$.

            For the backward implication, suppose that $I \setminus \{v\}$ is a maximum $(X^C,Z \setminus \{v\})$-\is in $G_{X^C}$. Note that $I$ is an $(X,Z)$-\is in $G_{X}$.
            Let $I'$ be a maximum $(X,Z)$-\is in $G_{X}$.
            As $I' \setminus \{v\}$ is an $(X^C,Z \setminus \{v\})$-\is in $G_{X^C}$, we get $|I' \setminus \{v\}| \le |I \setminus \{v\}|$
            and therefore, since both $I$ and $I'$ contain $v$, $|I'| \le |I|$ and the lemma follows.
          \end{proof}

          \begin{lemma}\label{lemma:intro:I2}
            Let $Z \subseteq X$ such that $Z$ is an \is with $v \notin Z$. For every $I \subseteq V(G)$ such that $v \notin I$, $I$ is a maximum $(X,Z)$-\is in $G_X$ if and only if $I$ is a maximum $(X^C,Z)$-\is in $G_{X^C}$.
          \end{lemma}
          \begin{proof}
            For the forward implication, suppose that $I$ is a maximum $(X,Z)$-\is in $G_X$ such that $v\notin I$. Note that $I$ is an $(X^C,Z)$-\is in $G_{X^C}$.
            Let $I'$ be a maximum $(X^C,Z)$-\is in $G_{X^C}$.
            As $I$' is an $(X,Z)$-\is in $G_X$, $|I'| \le |I|$, leading to the desired result.

            For the backward implication, consider that $I$ is a maximum $(X^C,Z)$-\is in $G_{X^C}$. Note that $I$ is an $(X,Z)$-\is in $G_{X}$.
            Let $I'$ be a maximum $(X,Z)$-\is in $G_{X}$.
            As $v \notin Z$, $v \notin I'$, and $I'$ is an $(X^C,Z)$-\is in $G_{X^C}$, implying $|I'| \le |I|$.
          \end{proof}

          \begin{lemma}\label{lemma:intro:alphaprime}
            Let $Z \subseteq X$ such that $Z$ is an \is.
            For every $\Y \subseteq V(G_X)$ such that $v \notin \Y$, $\Y$ is an $(X,Z)$-\bs in $G_X$ if and only if $\Y$ is an $(X^C,Z \setminus \{v\})$-\bs in $G_{X^C}$.
          \end{lemma}
          \begin{proof}
            For the forward implication, assume that $\Y$ is an $(X,Z)$-\bs in $G_X$ such that $v \notin \Y$. Let $I$ be a maximum $(X^C,Z \setminus \{v\})$-\is in $G_{X^C}$.
            Suppose first that $v \in Z$.
            By Lemma~\ref{lemma:intro:I1}, we get that $I \cup \{v\}$ is a maximum $(X,Z)$-\is in $G_X$,
            implying that $\Y \cap (I \cup \{v\}) \neq \emptyset$. As $v \notin \Y$, we get $\Y \cap I \neq \emptyset$.
            Suppose now that $v \notin Z$.
            By Lemma~\ref{lemma:intro:I2}, we get that $I$ is a maximum $(X,Z)$-\is in $G_X$, implying $\Y \cap I \neq \emptyset$.

            For the backward implication, suppose that $\Y$ is an $(X^C,Z \setminus \{v\})$-\bs in $G_{X^C}$. Let $I$ be a maximum $(X,Z)$-\is in $G_X$.
            Suppose first that $v \in Z$.
            By Lemma~\ref{lemma:intro:I1}, we get that $I \setminus \{v\}$ is a maximum $(X^C,Z \setminus \{v\})$-\is in $G_{X^C}$,
            implying that $\Y \cap (I \setminus \{v\}) \neq \emptyset$.
            Suppose now that $v \notin Z$.
            By Lemma~\ref{lemma:intro:I2}, we get that $I$ is a maximum $(X^C,Z)$-\is in $G_{X^C}$, implying $\Y \cap I \neq \emptyset$.
          \end{proof}

          \begin{lemma}\label{lemma:intro:gamma}
            Let $Z \subseteq X$ such that $Z$ is an \is with $v \notin Z$. For every $\Y \subseteq V(G_X)$ such that $v \in \Y$, $\Y$ is an $(X,Z)$-\bs in $G_X$ if and only if $\Y \setminus \{v\}$ is an $(X^C,Z)$-\bs in $G_{X^C}$.
          \end{lemma}
          \begin{proof}
            For the forward implication, suppose that $\Y$ is an $(X,Z)$-\bs in $G_X$ such that $v \in \Y$.
            Let $I$ be a maximum $(X^C,Z)$-\is in $G_{X^C}$.
            By Lemma~\ref{lemma:intro:I2}, $I$ is a maximum $(X,Z)$-\is in $G_X$,
            implying that $\Y \cap I \neq \emptyset$. As $v \notin I$, we get $(\Y \setminus \{v\}) \cap I \neq \emptyset$.

            For the backward implication, suppose that $\Y \setminus \{v\}$ is an $(X^C,Z)$-\bs in $G_{X^C}$. Let $I$ be a maximum $(X,Z)$-\is in $G_X$.
            By Lemma~\ref{lemma:intro:I2}, $I$ is a maximum $(X^C,Z)$-\is in $G_{X^C}$,
            implying that $(\Y \setminus \{v\}) \cap I \neq \emptyset$.
          \end{proof}

          We are now ready to state the main lemma of this section.
          Let us recall that given a function $f: A \to B$ and a subset $A' \subseteq A$, we denote by $f_{|A'}$ the restriction of $f$ to $A'$.

          \begin{lemma}\label{lem:intro}
            Let $(X,\Y_0,\L_1,\L_2,f,\S) \in \E(G,\D)$ where $X \in \B$ is an introduce node and $X^C$ is the child of $X$ with $X^C=X \setminus \{v\}$.
            For every $\Y \subseteq V(G_X)$,
            $\Y \vdash (X,\Y_0,\L_1,\L_2,f,\S)$ if and only if one of the following two cases holds:\\


            \noindent\textbf{Case 1}: $v \in \Y$ and there exist $\L_2^{A}, \L_2^{B}$ such that
            \begin{enumerate}
              \item \label{intro:1} $\L_2(v)=\L_2^A \uplus \L_2^B$,
              \item \label{intro:2} $f(v) \in \L_2^A$,
              \item \label{intro:3} for every $Z \in \S$, $v \notin Z$, and
              \item \label{intro:4} $\Y \setminus \{v\} \vdash (X \setminus \{v\},\Y_0 \setminus \{v\},\L_1^C,\L_2^C,f^C,\S^C)$, where
              \begin{itemize}
                \item[$\bullet$]$\L_1^C = \L_1(\bar{v})$,
                \item[$\bullet$]$\L_2^C = \L_2(\bar{v})$,
                \item[$\bullet$]$f^C = f_{|\Y\setminus \{v\}}$, and
                \item[$\bullet$]$\S^C = \S \cup r_v(\L_2^A)$.

              \end{itemize}
            \end{enumerate}

            \noindent\textbf{Case 2}: $v \notin \Y$ and
            \label{intro:1bis} $\Y \vdash (X \setminus \{v\},\Y_0 , \L_1^C,\L_2^C,f^C,\S^C)$, where
            \begin{itemize}
              \item[$\bullet$]$\L_1^C = r_v(\L_1)$,
              \item[$\bullet$]$\L_2^C = r_v(\L_2)$,
              \item[$\bullet$]$f^C(v') = f(v') \setminus \{v\}$ for every $v'  \in \Y_0$, and
              \item[$\bullet$]$\S^C = r_v(\S)$.
            \end{itemize}
          \end{lemma}

          \begin{proof}
            For the forward implication, suppose that $\Y \subseteq V(G_X)$ is such that $\Y \vdash (X,\Y_0,\L_1,\L_2,f,\S)$. We distinguish the two cases considered in Lemma~\ref{lem:intro}. In both cases, we verify that each of the corresponding properties is satisfied.

            \medskip

            \textbf{Case 1}. Suppose that $v \in \Y$, and thus $v\in \Y_0$, as $v\in X$ and $\Y_0=\Y\cap X$.
            Let $\L_2^A = \{ Z \in \L_2 |$ $\Y \setminus \{v\}$ is not an $(X,Z)$-\bs in $G_X \}$.
            By Property~\ref{def:minb} applied to $v$, we get that $f(v) \in \L_2^A$, implying Property~\ref{intro:2}.
            Moreover, for every $Z \in \L_2^A$, there exists a maximum $(X,Z)$-\is $I$ in $G_X$ such that $I \cap (\Y \setminus \{v\}) = \emptyset$,
            and thus if we had $v \notin Z$, then $v \notin I$ and $I \cap \Y = \emptyset$, a contradiction.
            This implies that $\L_2^A \subseteq \L_2(v)$, and we define $\L_2^B = \L_2(v) \setminus \L_2^A$, implying Property~\ref{intro:1}.
            By Property~\ref{def:S}, as $\Y$ is not an $(X,Z)$-\bs in $G_X$ for every $Z \in \S$ and $v \in \Y$, we get Property~\ref{intro:3}.
            Let us now prove Property~\ref{intro:4}, by verifying each of the non-trivial properties of Definition~\ref{def:tchack} applied to $\Y \setminus \{v\}$.

            \emph{Property~\ref{def:L1L2}}. Recall that $X^C = X\setminus \{v\}$. We need to prove that $\Y\setminus\{v\}$ is an $(X^C, Z)$-\bs in $G_{X^C}$, for every $Z\in \L_1^C \cup \L_2^C = \L_1(\bar{v})\cup \L_2(\bar{v})$.
            Let $Z \in \L_1^C \cup \L_2^C$. As $v \notin Z$, Lemma~\ref{lemma:intro:gamma} implies that $\Y \setminus \{v\}$ is an $(X^C,Z)$-\bs in $G_{X^C}$.

            \emph{Property~\ref{def:S}}. We must prove that $\Y\setminus\{v\}$ is not an $(X^C,Z)$-\bs in $G_{X^C}$, for every $Z\in \S^C = \S\cup r_v(\L_2^A)$.
            Let $Z \in \S^C$.
            If $Z \in \S$, as $v \notin Z$ (which we know from Property~\ref{intro:3}) and $\Y$ is not an $(X,Z)$-\bs in $G_X$ (since $\Y \vdash (X,\Y_0,\L_1,\L_2,f,\S)$), Lemma~\ref{lemma:intro:gamma} implies that $\Y \setminus \{v\}$ is not an $(X^C,Z)$-\bs in $G_{X^C}$.
            If $Z \in r_v(\L_2^A)$, then let $Z'$ be such that $Z = Z' \setminus \{v\}$. We know that $\Y \setminus \{v\}$ is not an $(X,Z')$-\bs in $G_X$.
            By Lemma~\ref{lemma:intro:alphaprime}, we get that $\Y \setminus \{v\}$ is not an $(X^C,Z)$-\bs in $G_{X^C}$.

            \emph{Property~\ref{def:mina}}. Let us now prove that, for every $v'\in (\Y\setminus\{v\})\setminus(\Y_0\setminus\{v\})$ there is $Z\in \L_1^C = \L_1(\bar{v})$ such that $(\Y\setminus\{v\})\setminus\{v'\}$ is not an $(X^C,Z)$-\bs in $G_{X^C}$. Since $v\in \Y_0$, let $v'\in \Y\setminus\Y_0$.
            Since $\Y \vdash (X,\Y_0,\L_1,\L_2,f,\S)$, there exists $Z \in \L_1$ such that $\Y \setminus \{v'\}$ is not an $(X,Z)$-\bs in $G_X$.
            As $v \in \Y \setminus \{v'\}$, this implies that $Z \in \L_1(\bar{v})$.
            As $v \notin Z$, from Lemma~\ref{lemma:intro:gamma} we get that $(\Y \setminus \{v\}) \setminus \{v'\}$ is not an $(X^C,Z)$-\bs in $G_{X^C}$.

            \emph{Property~\ref{def:minb}}. We now have to prove that for every $v'\in \Y_0\setminus\{v\}$, $(\Y\setminus\{v\})\setminus\{v'\}$ is not an $(X^C, f^C(v'))$-\bs in $G_{X^C}$. If $\Y_0 \setminus \{v\}=\emptyset$, we have nothing to prove. Otherwise, let $v' \in \Y_0 \setminus \{v\}$. Since $\Y \vdash (X,\Y_0,\L_1,\L_2,f,\S)$, there exists $Z \in \L_2$ such that $Z = f^C(v')=f(v')$
             and $\Y \setminus \{v'\}$ is not an $(X,Z)$-\bs in $G_X$.
            As $v \in \Y \setminus \{v'\}$, this implies that $Z \in \L_2(\bar{v})$.
            As $v \notin Z$, Lemma~\ref{lemma:intro:gamma} implies that $(\Y \setminus \{v\}) \setminus \{v'\}$ is not an $(X^C,Z)$-\bs in $G_{X^C}$.

            \medskip

            \textbf{Case 2}. Suppose that $v \notin \Y$. Let us prove that $\Y \vdash (X \setminus \{v\},\Y_0 , r_v(\L_1)$ ,$r_v(\L_2),f^C,r_v(\S))$, where $f^C(v') = f(v') \setminus \{v\}$ for every $v'  \in \Y_0$.

            \emph{Property~\ref{def:L1L2}}. We first prove that $B$ is an $(X^C,Z)$-set in $G_{X^C}$, for every $Z\in \L_1^C\cup \L_2^C=r_v(\L_1)\cup r_v(\L_2)$.
            Let $Z \in r_v(\L_1)\cup r_v(\L_2)$ where $Z = Z' \setminus \{v\}$ and $Z' \in \L_1 \cup \L_2$.
            As $v \notin \Y$ and $\Y$ is an $(X,Z')$-\bs in $G_X$, Lemma~\ref{lemma:intro:alphaprime} implies that $\Y$ is an $(X^C,Z)$-\bs in $G_{X^C}$.

            \emph{Property~\ref{def:S}}.
            Let us prove that $\Y$ is not an $(X^C,Z)$-\bs in $G_{X^C}$, for every $Z\in \S^C= r_v(\S)$.
            Let $Z \in r_v(\S)$, where $Z = Z' \setminus \{v\}$ and $Z' \in \S$.
            As $v \notin \Y$, and as $\Y$ is not an $(X,Z')$-\bs in $G_X$, Lemma~\ref{lemma:intro:alphaprime} implies that $\Y$ is not an $(X^C,Z)$-\bs in $G_{X^C}$.

            \emph{Property~\ref{def:mina}}.
            We now prove that, for every $v'\in \Y\setminus \Y_0$, there exists $Z\in \L_1^C = r_v(\L_1)$ such that $\Y\setminus\{v'\}$ is not an $(X^C,Z)$-\bs in $G_{X^C}$. Suppose that $\Y\setminus\Y_0 \neq \emptyset$, as otherwise the statement trivially holds.
            Let $v' \in \Y \setminus \Y_0$.
            Since $\Y \vdash (X,\Y_0,\L_1,\L_2,f,\S)$, there exists $Z \in \L_1$ such that $\Y \setminus \{v'\}$ is not an $(X,Z)$-\bs in $G_X$.
            As $v \notin \Y \setminus \{v'\}$,  Lemma~\ref{lemma:intro:alphaprime} implies that $\Y \setminus \{v'\}$ is not an $(X^C,Z \setminus \{v\})$-\bs in $G_{X^C}$,
            and $Z \setminus \{v\} \in \L_1^C$.

            \emph{Property~\ref{def:minb}}. We must finally prove that, for every $v'\in \Y_0$, $\Y\setminus\{v'\}$ is not an $(X^C,f^C(v'))$-\bs in $G_{X^C}$. Let $v' \in \Y_0$. Since $\Y \vdash (X,\Y_0,\L_1,\L_2,f,\S)$, let $Z =f(v')$, where $Z \in \L_2$, such that $\Y \setminus \{v'\}$ is not an $(X,Z)$-\bs in $G_X$ (by Property~\ref{def:minb}).
            As $v \notin \Y \setminus \{v'\}$,  Lemma~\ref{lemma:intro:alphaprime} implies that $\Y \setminus \{v'\}$ is  not an $(X^C,Z \setminus \{v\})$-\bs in $G_{X^C}$,
            and $Z \setminus \{v\} = f^C(v')$.

            \medskip

            We now focus on the backward implication, and we distinguish again the two cases according to the possible hypothesis. In both cases, remind that our goal is to prove that $\Y \vdash (X,\Y_0,\L_1,\L_2,f,\S)$.

            \medskip

            \textbf{Case 1}.  Let $\Y$ with $v \in \Y$ and suppose that there exist $\L_2^{A}, \L_2^{B}$ satisfying the statement of the lemma.

            \emph{Property~\ref{def:L1L2}}. Let us prove that $\Y$ is an $(X,Z)$-\bs in $G_X$, for every $Z\in\L_1\cup \L_2$.
            Let $Z \in \L_1 \cup \L_2$. If $v\in Z$, then as $v \in \Y$ it follows that $\Y$ is an $(X,Z)$-\bs in $G_{X}$.
            Otherwise, by Property~\ref{def:L1L2}, we get that $\Y \setminus \{v\}$ is an $(X^C,Z)$-\bs in $G_{X^C}$.
            As $v \notin Z$, Lemma~\ref{lemma:intro:gamma} implies that $\Y$ is an $(X,Z)$-\bs in $G_{X}$.

            \emph{Property~\ref{def:S}}.
            We now prove that $\Y$ is not an $(X,Z)$-\bs in $G_X$, for every $Z\in \S$.
            Let $Z \in \S$. We know that $\Y \setminus \{v\}$ is not an $(X^C,Z)$-\bs in $G_{X^C}$.
            By hypothesis, $v \notin Z$, and thus Lemma~\ref{lemma:intro:gamma} implies that $\Y$ is not an $(X,Z)$-\bs in $G_{X}$.

            \emph{Property~\ref{def:mina}}.
            We have to show that, for every $v'\in \Y\setminus\Y_0$, there exists $Z\in \L_1$ such that $B\setminus \{v'\}$ is not an $(X,Z)$-\bs in $G_X$. Assume that $\Y\setminus\Y_0\neq\emptyset.$
            Let $v' \in \Y \setminus \Y_0$.
            By hypothesis, there exists $Z \in \L_1^C=\L_1(\bar{v})$ such that $(\Y \setminus \{v\}) \setminus \{v'\}$ is not an $(X^C,Z)$-\bs in $G_{X^C}$.
            As $v \notin Z$, Lemma~\ref{lemma:intro:gamma} implies that $(\Y \setminus \{v'\})$ is not an $(X,Z)$-\bs in $G_X$.

            \emph{Property~\ref{def:minb}}. We finally prove that, for every $v'\in \Y_0$, $\Y\setminus\{v'\}$ is not an $(X,f(v'))$-\bs in $G_X$. Let $v' \in \Y_0$ and let $Z = f(v')$.
            Suppose first that $v' = v$.  By Property~\ref{intro:2} we know that $Z \in \L_2^A$, and by Property~\ref{intro:1}
            we know that $v \in Z$, implying then that there exists $Z'$ such that $Z = Z' \cup \{v\}$. As $r_v(\L_2^A) \subseteq \S^C$, it follows that $\Y \setminus \{v\}$ is not an $(X^C,Z')$-\bs in $G_{X^C}$.
            By Lemma~\ref{lemma:intro:alphaprime}, $\Y \setminus \{v\}$ is not an $(X,Z)$-\bs in $G_{X}$.
            Suppose now that $v' \neq v$.
            By hypothesis, $(\Y \setminus \{v\}) \setminus \{v'\}$ is not an $(X^C,Z)$-\bs in $G_{X^C}$, and $Z \in \L_2(\bar{v})$.
            As $v \notin Z$, by Lemma~\ref{lemma:intro:gamma} $\Y \setminus \{v'\}$ is not an $(X,Z)$-\bs in $G_{X}$.

            \medskip

            \textbf{Case 2}. Let $\Y$ with $v \notin \Y$ and $\Y \vdash (X \setminus \{v\},\Y_0 , r_v(\L_1),r_v(\L_2),f^C,r_v(\S))$, where $f^C(v') = f(v') \setminus \{v\}$ for every $v'  \in \Y_0$.

            \emph{Property~\ref{def:L1L2}}.
            Let us prove that $\Y$ is an $(X,Z)$-\bs in $G_X$, for every $Z\in\L_1\cup \L_2$.
            Let $Z \in \L_1 \cup \L_2$.
            As $v \notin \Y$ and as $\Y$ is an $(X^C,Z \setminus \{v\})$-\bs in $G_{X^C}$, Lemma~\ref{lemma:intro:alphaprime} implies that $\Y$ is an $(X,Z)$-\bs in $G_X$.

            \emph{Property~\ref{def:S}}.
            We now prove that $\Y$ is not an $(X,Z)$-\bs in $G_X$, for every $Z\in \S$.
            Let $Z \in \S$.
            As $v \notin \Y$ and as $\Y$ is not an $(X^C,Z \setminus \{v\})$-\bs in $G_{X^C}$, Lemma~\ref{lemma:intro:alphaprime} implies that $\Y$ is not an $(X,Z)$-\bs in $G_{X}$.

            \emph{Property~\ref{def:mina}}.
            We have to show that, for every $v'\in \Y\setminus\Y_0$, there exists $Z\in \L_1$ such that $B\setminus \{v'\}$ is not an $(X,Z)$-\bs in $G_X$. If $\Y\setminus\Y_0=\emptyset$, then we have nothing to prove. Let $v' \in \Y \setminus \Y_0$.
            By hypothesis, there exists $Z \in \L^C_1$ such that $\Y \setminus \{v'\}$ is not an $(X^C,Z)$-\bs in $G_{X^C}$.
            Let $Z' \in \L_1$ such that $Z = Z' \setminus \{v\}$.
            As $v \notin \Y \setminus \{v'\}$,  Lemma~\ref{lemma:intro:alphaprime} implies that $\Y \setminus \{v'\}$ is not an $(X,Z')$-\bs in $G_X$.

            \emph{Property~\ref{def:minb}}.
            We finally prove that, for every $v'\in \Y_0$, $\Y\setminus\{v'\}$ is not an $(X,f(x))$-\bs in $G_X$.
            Let $v' \in \Y_0$ and let $Z =f(v')$, where $Z \in \L_2$. Recall that $f^C(v')=Z \setminus \{v\}$.
            By hypothesis, $\Y \setminus \{v'\}$ is not an $(X^C,Z \setminus \{v\})$-\bs in $G_{X^C}$.
            As $v \notin \Y \setminus \{v'\}$,  Lemma~\ref{lemma:intro:alphaprime} implies that $\Y \setminus \{v'\}$ is not an $(X,Z)$-\bs in $G_{X}$.
          \end{proof}

          \subsection{Forget node}
          \label{sec:forget}

          \mbvfinale{please check this case as in case 1 I modified the def of function f, and added condition \ref{forget1:1bisbis}}

          Let us start with some preliminaries related to the notion of \emph{criticality}.

          \begin{definition}\label{def:critical}
            Let $G$ be a graph, $X \subseteq V(G)$, $Z \subseteq X$ such that $Z$ is an \is, and $v \in V(G)$.
            We say that $(X,Z)$ is
            \begin{itemize}
              \item[$\bullet$]\emph{$v$-critical in $G$} if for every maximum $(X,Z)$-\is $I$ in $G$, $v \in I$,
              \item[$\bullet$]\emph{$\bar{v}$-critical in $G$} if for every maximum $(X,Z)$-\is $I$ in $G$, $v \notin I$, and
              \item[$\bullet$]\emph{$v$-mixed in $G$} if there exists a maximum $(X,Z)$-\is $I$ in $G$ with $v \in I$ and there exists a maximum $(X,Z)$-\is $I'$ in $G$ with $v \notin I$.
            \end{itemize}

            Given $v \in V(G)$ and a set $\R \subseteq 2^X$ such that for each $Z \in \R$,  $Z$ is an \is, we denote

            \begin{itemize}
              \item[$\bullet$]$\R^{(v,X)} = \{Z \in \R \mid $ $(X,Z)$ is $v$-critical in $G_X \}$,
              \item[$\bullet$]$\R^{(\bar{v},X)} = \{Z \in \R \mid $ $(X,Z)$ is $\bar{v}$-critical in $G_X \}$,
              \item[$\bullet$]$\R^{(*v,X)} = \{Z \in \R \mid $ $(X,Z)$ is $v$-mixed in $G_X \}$, and
              \item[$\bullet$]$a_v(\R) = \{Z \cup \{v\} \mid Z \in \R\}$.
            \end{itemize}
          \end{definition}

          \begin{lemma}\label{prop:critical}
            Let $G$ be a graph, $X \subseteq V(G)$, and $Z \subseteq X$ such that $Z$ is an \is and $v \in V(G)$.
            Deciding whether $(X,Z)$ is $v$-critical in $G$, $\bar{v}$-critical in $G$ or $v$-mixed in $G$ can be done in time $\O^*(2^{\tw(G)})$.
          \end{lemma}
          \begin{proof}
            If $v \in Z$, then $(X,Z)$ is by definition $v$-critical in $G$, and if $v \in X \setminus Z$ then $(X,Z)$ is by definition $\bar{v}$-critical $G$.
            Suppose now that $v \notin X$.
            Then, observe that
            \begin{itemize}
              \item[$\bullet$] $(X,Z)$ is $v$-critical in $G$ if and only if $\alpha_{(X \cup \{v\},Z)}(G) < \alpha_{(X,Z)}(G)$,
              \item[$\bullet$] $(X,Z)$ is $\bar{v}$-critical in $G$ if and only if $\alpha_{(X \cup \{v\},Z \cup \{v\})}(G) < \alpha_{(X,Z)}(G)$, and
              \item[$\bullet$] $(X,Z)$ is $v$-mixed in $G$ if and only if $\alpha_{(X \cup \{v\},Z \cup \{v\})}(G) = \alpha_{(X \cup \{v\},Z)}(G) = \alpha_{(X,Z)}(G)$.
            \end{itemize}
            Let us now prove that, for every $(X', Z')$ such that $Z'\subseteq X'$ and $Z'$ is an \is, $\alpha_{(X',Z')}(G)$ can be computed in time $\O^*(2^{\tw(G)})$.
            Indeed, observe that $\alpha_{(X',Z')}(G) = |Z'|+\alpha(G \setminus (X' \cup N(Z')))$. As $\tw(G \setminus (X' \cup N(Z'))) \le \tw(G')$ and $\alpha(G)$ can be computed in $\O^*(2^{\tw(G)})$~\cite{CyganFKLMPPS15}, we get the desired result.
          \end{proof}


          Before proving Lemma~\ref{lem:forget} corresponding to the forget case, let us first prove the following lemmas, where we assume that $X \in \B$ is a forget node and $X^C$ is the child of $X$ with $X^C=X \cup \{v\}$, for some vertex $v \in X_C$.

          \begin{lemma}\label{lemma:I1}
            Let $I \subseteq V(G)$ such that $v \in I$ and let $Z \subseteq X$ such that $Z$ is an \is. The following claims hold:
            \begin{itemize}
              \item[$\bullet$] If $I$ is a maximum $(X,Z)$-\is in $G_X$, then $I$ is a maximum $(X^C,Z \cup \{v\})$-\is in $G_{X^C}$.
              \item[$\bullet$] If $(X,Z)$ is not $\bar{v}$-critical in $G_X$, then $I$ is a maximum $(X,Z)$-\is in $G_X$ if and only if $I$ is a maximum $(X^C,Z \cup \{v\})$-\is in $G_{X^C}$.
            \end{itemize}
          \end{lemma}
          \begin{proof}
            For the first item, let $I$ be a maximum $(X,Z)$-\is in $G_X$. As $v \in I$, $I \cap X^C = Z \cup \{v\}$, and $I$ is also an $(X^C,Z \cup \{v\})$-\is in $G_{X^C}$.
            Let $I'$ be a maximum $(X^C,Z \cup \{v\})$-\is in $G_{X^C}$. As $I'$ is also an $(X,Z)$-\is in $G_X$, $|I| \ge |I'|$, and thus $I$ is a maximum  $(X^C,Z \cup \{v\})$-\is in $G_{X^C}$.

            For the second item, the sufficiency is already proved in the first item. For the backward implication, let $I$ be a maximum $(X^C,Z \cup \{v\})$-\is in $G_{X^C}$.
            Observe first that $I$ is an $(X,Z)$-\is in $G_X$.
            Let $I'$ be a maximum $(X,Z)$-\is in $G_X$ such that $v \in I'$, which exists as $(X,Z)$ is not $\bar{v}$-critical in $G_X$. By the previous property, $I'$ is a (maximum) $(X^C,Z \cup \{v\})$-\is in $G_{X^C}$,
            implying $|I| \ge |I'|$ and the desired result.
          \end{proof}

          \begin{lemma}\label{lemma:I2}
            Let $I \subseteq V(G)$ such that $v \notin I$ and let $Z \subseteq X$ such that $Z$ is an \is. The following claims hold:
            \begin{itemize}
              \item[$\bullet$] If $I$ is a maximum $(X,Z)$-\is in $G_X$, then $I$ is a maximum $(X^C,Z)$-\is in $G_{X^C}$.
              \item[$\bullet$] If $(X,Z)$ is not $v$-critical in $G_X$, then  $I$ is a maximum $(X,Z)$-\is in $G_X$ if and only if $I$ is a maximum $(X^C,Z)$-\is in $G_{X^C}$.
            \end{itemize}
          \end{lemma}
          \begin{proof}
            For the first item, let $I$ be a maximum $(X,Z)$-\is in $G_X$. As $v \notin I$, $I \cap X^C = Z$, hence $I$ is also  an $(X^C,Z)$-\is in $G_{X^C}$.
            Let $I'$ be a maximum $(X^C,Z)$-\is in $G_{X^C}$. As $I'$ is also an $(X,Z)$-\is in $G_X$, $|I| \ge |I'|$, and thus $I$ is a maximum  $(X^C,Z)$-\is in $G_{X^C}$.

            For the second item, again we only need to prove the backward implication.
            Let $I$ be a maximum $(X^C,Z)$-\is in $G_{X^C}$.
            Observe first that $I$ is an $(X,Z)$-\is in $G_X$.
            Let $I'$ be a maximum $(X,Z)$-\is in $G_X$ such that $v \notin I'$, which exists as $(X,Z)$ is not $v$-critical in $G_X$. By the first item, $I'$ is a (maximum) $(X^C,Z)$-\is in $G_{X^C}$,
            implying $|I| \ge |I'|$ and the desired result.
          \end{proof}

          \begin{lemma}\label{lemma:Y}
            Let $Z \subseteq X$ where $Z$ is an \is. The following claims hold:
            \begin{itemize}
              \item[$\bullet$] If $(X,Z)$ is $v$-critical in $G_X$, then for each $\Y \subseteq V(G_X)$, $\Y$ is an $(X,Z)$-\bs in $G_X$ if and only if $\Y$ is an $(X^C,Z \cup \{v\})$-\bs in $G_{X^C}$.
              \item[$\bullet$] If $(X,Z)$ is $\bar{v}$-critical in $G_X$, then for each $\Y \subseteq V(G_X)$, $\Y$ is an $(X,Z)$-\bs in $G_X$ if and only if $\Y$ is an $(X^C,Z)$-\bs in $G_{X^C}$.
              \item[$\bullet$] If $(X,Z)$ is $v$-mixed in $G_X$, then for each $\Y \subseteq V(G_X)$, $\Y$ is an $(X,Z)$-\bs in $G_X$ if and only if
              $\Y$ is an $(X^C,Z \cup \{v\})$-\bs in $G_{X^C}$ and $\Y$ is an $(X^C,Z)$-\bs in $G_{X^C}$.
            \end{itemize}
          \end{lemma}
          \begin{proof}
            For the first item, let $Z \subseteq X$ such that $(X,Z)$ is $v$-critical in $G_X$.

            \smallskip

            For the forward implication, suppose that $\Y$ is an $(X,Z)$-\bs in $G_X$. Let $I$ be a maximum $(X^C,Z \cup \{v\})$-\is in $G_{X^C}$.
            As $v \in I$ and as $(X,Z)$ is not $\bar{v}$-critical in $G_X$, by Lemma~\ref{lemma:I1}, $I$ is also a maximum $(X,Z)$-\is in $G_X$, implying that $I \cap \Y \neq \emptyset$.      

            For the backward implication, let $\Y$ be an $(X^C,Z \cup \{v\})$-\bs $G_{X^C}$. Let $I$ be a maximum $(X,Z)$-\is in $G_{X}$.
            As $(X,Z)$ is $v$-critical in $G_X$, we know that $v \in I$. By Lemma~\ref{lemma:I1}, $I$ is also a maximum $(X^C,Z \cup \{v\})$-\is in $G_{X^C}$, implying that $I \cap \Y \neq \emptyset$.

            \medskip

            For the second item, let $Z \subseteq X$ such that $(X,Z)$ is $\bar{v}$-critical in $G_X$.
            \smallskip

            For the forward implication, suppose that $\Y$ is an $(X,Z)$-\bs in $G_X$. Let $I$ be a maximum $(X^C,Z)$-\is in $G_{X^C}$.
            As $v \notin I$ and as $(X,Z)$ is not $v$-critical in $G_X$, by Lemma~\ref{lemma:I2}, $I$ is also a maximum $(X,Z)$-\is in $G_X$, implying that $I \cap \Y \neq \emptyset$.

            For the backward implication, let $\Y$ be an $(X^C,Z)$-\bs $G_{X^C}$. Let $I$ be a maximum $(X,Z)$-\is in $G_{X}$.
            As $(X,Z)$ is $\bar{v}$-critical in $G_X$, we know that $v \notin I$. By Lemma~\ref{lemma:I2}, $I$ is also a maximum $(X^C,Z)$-\is in $G_{X^C}$, implying that $I \cap \Y \neq \emptyset$.

            \medskip

            For the third item, let $Z \subseteq X$ such that $(X,Z)$ is $v$-mixed in $G_X$.
            \smallskip

            For the forward implication, assume that $\Y$ is an $(X,Z)$-\bs in $G_X$. Let $I_1$ be a maximum $(X^C,Z \cup \{v\})$-\is in $G_{X^C}$ and $I_2$ be a maximum $(X^C,Z)$-\is in $G_{X^C}$.
            As $(X,Z)$ is both not $v$-critical and not $\bar{v}$-critical in $G_X$, by Lemmas~\ref{lemma:I1} and~\ref{lemma:I2} we now that both $I_1$ and $I_2$
            are maximum $(X,Z)$-\is in $G_X$, implying $\Y \cap I_1 \neq \emptyset$ and $\Y \cap I_2 \neq \emptyset$.         

            For the backward implication, let finally $\Y$ be an $(X^C,Z \cup \{v\})$-\bs in $G_{X^C}$ and an $(X^C,Z)$-\bs in $G_{X^C}$.
            Let $I$ be a maximum $(X^C,Z)$-\is in $G_{X^C}$. If $v \in I$, by Lemma~\ref{lemma:I1}, $I$ is also a maximum $(X^C,Z \cup \{v\})$-\is in $G_{X^C}$,
            implying $I \cap \Y \neq \emptyset$, and if $v \notin I$, by Lemma~\ref{lemma:I2}, $I$ is also a maximum $(X^C,Z)$-\is in $G_{X^C}$,
            implying $I \cap \Y \neq \emptyset$ as well.
          \end{proof}

          \begin{lemma}\label{lemma:liste}
            Let $\L \subseteq 2^X$ where for each $Z \in \L$, $Z$ is an \is. For every $\Y \subseteq G_X$, $\Y$ is an $(X,Z)$-\bs in $G_X$ for every
            $Z \in \L$ if and only if  $\Y$ is an $(X^C,Z)$-\bs in $G_{X^C}$ for every $Z \in a_v(\L^{(v,X)}) \cup \L^{(\bar{v},X)} \cup a_v(\L^{(*v,X)}) \cup \L^{(*v,X)}$.
          \end{lemma}
          \begin{proof}
            For the forward implication, suppose that $\Y$ is an $(X,Z)$-\bs in $G_X$, for every
            $Z \in \L$.
              Let $Z \in a_v(\L^{(v,X)})$ (resp. $Z \in a_v(\L^{(*v,X)})$), implying that $Z = Z' \cup \{v\}$ with $Z' \in \L^{(v,X)}$ (resp. $Z' \in \L^{(*v,X)}$).
              Observe that for every $Z \in \L$, $v \notin Z$, implying that $v \notin Z'$ and thus that we also have $Z' = Z \setminus \{v\}$.
              \mbvfinale{il faudra regarder si, soit on bien que $v$ pas dans les $Z'$ partout, ou sinon, bien voir si on suppose pas que $Z' = Z \setminus \{v\}$ (qui n'est pas
                vrai quand $v \in Z'$. Ou plus facile : normalement, partout ou on applique ce lemme, c'est avec une liste L dont aucun ensemble ne contient v} This implies that $(X,Z')$ is $v$-critical (resp. $v$-mixed) in $G_X$.
            By hypothesis, $\Y$ is an $(X,Z')$-\bs in $G_X$, implying, as $(X,Z')$ is $v$-critical (resp. $v$-mixed) in $G_X$,  that $\Y$ is an $(X^C,Z)$-\bs in $G_{X^C}$ by Lemma~\ref{lemma:Y}.
            Let now $Z \in \L^{(\bar{v},X)}$ (resp. $Z \in \L^{(*v,X)}$), implying that $(X,Z)$ is $\bar{v}$-critical (resp. $v$-mixed) in $G_X$.
            By hypothesis, $\Y$ is an $(X,Z)$-\bs in $G_X$, implying, as $(X,Z)$ is $\bar{v}$-critical (resp. $v$-mixed) in $G_X$, that $\Y$ is an $(X^C,Z)$-\bs in $G_{X^C}$  by Lemma~\ref{lemma:Y}.

            For the backward implication, suppose that $\Y$ is an $(X^C,Z)$-\bs in $G_{X^C}$ for every $Z \in a_v(\L^{(v,X)}) \cup \L^{(\bar{v},X)} \cup a_v(\L^{(*v,X)}) \cup \L^{(*v,X)}$.
            Let $Z \in \L$. If $Z \in \L^{(v,X)}$, then there exists $Z' \in a_v(\L^{(v,X)})$ such that $Z' =Z \cup \{v\}$.
            By hypothesis, $\Y$ is an $(X^C,Z')$-\bs in $G_{X^C}$, implying, as $(X,Z)$ is $v$-critical in $G_X$, that $\Y$ is an $(X,Z)$-\bs in $G_X$  by Lemma~\ref{lemma:Y}.
            If $Z \in \L^{(\bar{v},X)}$, then by hypothesis, $\Y$ is an $(X^C,Z)$-\bs in $G_{X^C}$, implying, as $(X,Z)$ is $\bar{v}$-critical in $G_X$,  that $\Y$ is an $(X,Z)$-\bs in $G_X$  by Lemma~\ref{lemma:Y}.
            If $Z \in \L^{(*v,X)}$, then by hypothesis $\Y$ is an $(X^C,Z \cup \{v\})$-\bs in $G_{X^C}$ and $\Y$ is an $(X^C,Z)$-\bs in $G_{X^C}$,
            implying, as $(X,Z)$ is $v$-mixed in $G_X$,  that $\Y$ is an $(X,Z)$-\bs in $G_X$  by Lemma~\ref{lemma:Y}.
          \end{proof}

          We are now ready to state the main lemma of this section.

          \begin{lemma}\label{lem:forget}
            Let $(X,\Y_0,\L_1,\L_2,f,\S) \in \E$ where $X \in \B$ is a forget node and $X^C$ is the child of $X$ with $X^C=X \cup \{v\}$.
            For each $\Y \subseteq V(G_X)$,
            $\Y \vdash (X,\Y_0,\L_1,\L_2,f,\S)$ if and only if one of the following two cases holds:\\

            \noindent\textbf{Case 1}: there exists $Z^* \in \L_1$ such that
            \begin{enumerate}
              \item \label{forget1:1} $(X,Z^*)$ is not $\bar{v}$-critical in $G_X$,
              \item \label{forget1:1bis} for each $Z \in \S$, $(X,Z)$ is not $v$-critical in $G_X$,
              \item \label{forget1:1bisbis} for each $v' \in \Y_0$, $(X,f(v'))$ is not $v$-critical in $G_X$, and
              \item \label{forget1:2} $\Y \vdash (X \cup \{v\},\Y_0 \cup \{v\},\L_1^C,\L_2^C,f^C,\S^C)$, where

              \begin{itemize}
                \item[$\bullet$]$\L_1^C = a_v(\L_1^{(v,X)}) \cup \L_1^{(\bar{v},X)} \cup a_v(\L_1^{(*v,X)}) \cup \L_1^{(*v,X)}$,
                \item[$\bullet$]$\L_2^C = a_v(\L_2^{(v,X)}) \cup \L_2^{(\bar{v},X)} \cup a_v(\L_2^{(*v,X)}) \cup \L_2^{(*v,X)} \cup \{Z^* \cup \{v\}\}$,
                \item[$\bullet$]\label{forget1:3} $f^C:\Y_0\cup\{v\} \to \L_2^C$ is such that


                $$f^C(v')=\begin{cases}
                  Z^*\cup\{v\} &\mbox{, if }v'= v,
                  \\f(v')&\mbox{, otherwise, and}
                \end{cases}$$

                \item[$\bullet$]$\S^C = \S$.

              \end{itemize}

            \end{enumerate}

            \noindent\textbf{Case 2}: there exist $\S^A, \S^B$, and $f^C$ such that
            \begin{enumerate}
              \item $\S^{(*v,X)} = \S^A \uplus \S^B$ and
              \item \label{forget2:1} $\Y \vdash (X \cup \{v\},\Y_0,\L_1^C,\L_2^C,f^C,\S^C)$, where

              \begin{itemize}
                \item[$\bullet$]$\L_1^C = a_v(\L_1^{(v,X)}) \cup \L_1^{(\bar{v},X)} \cup a_v(\L_1^{(*v,X)}) \cup \L_1^{(*v,X)}$,
                \item[$\bullet$]$\L_2^C = a_v(\L_2^{(v,X)}) \cup \L_2^{(\bar{v},X)} \cup a_v(\L_2^{(*v,X)}) \cup \L_2^{(*v,X)}$,
                \item[$\bullet$]\label{forget2:2} $f^C: B_0 \to \L_2^C$ is such that


                $$f^C(v')=\begin{cases}
                  f(v') \cup \{v\}&\mbox{, if } f(v') \in \L_2^{(v,X)},\\f(v')&\mbox{, if }f(v') \in \L_2^{(\bar{v},X)},
                \end{cases}$$
                otherwise $f^C(v') \in \{f(v'),f(v') \cup \{v\} \}$, and

                \item[$\bullet$]$\S^C = a_v(\S^{(v,X)}) \cup \S^{(\bar{v},X)} \cup a_v(\S^A) \cup \S^B$.
              \end{itemize}

            \end{enumerate}

          \end{lemma}

          \begin{proof}
            Observe first that in both cases, for every $Z \in \L_1^C \cup \L_2^C \cup \S^C$, $Z$ is an \is as required in the definition of $\E$.
          Indeed, for each $Z \in \L_1 \cup L_2 \cup \S$, we only add $Z \cup \{v\}$ to $\L_1^C \cup \L_2^C \cup \S^C$ when $Z$ is not $\bar{v}$-critical in $G_X$,
          implying that $Z \cup \{v\}$ is an \is.
          
          \smallskip
          
            For the forward implication, 
            suppose  that $\Y \subseteq V(G_X)$ is such that $\Y \vdash (X,\Y_0,\L_1,\L_2,f,\S)$, and let us distinguish two cases.

            \medskip

            Suppose first that $v \in \Y$. In this case, we will prove that all statements corresponding to Case 1 hold. Recall that $X=X^C\setminus \{v\}$ and $\Y_0=\Y\cap X$. Thus $v\in \Y\setminus \Y_0$.
            Since $\Y \vdash (X,\Y_0,\L_1,\L_2,f,\S)$, Property~\ref{def:mina} in Definition~\ref{def:tchack} implies that there exists $Z^* \in \L_1$ such that $\Y \setminus \{v\}$ is not
            an $(X,Z^*)$-\bs in $G_X$, implying that there exists a maximum $(X,Z^*)$-\is $I^*$ in $G_X$ such that $I^* \cap (\Y \setminus \{v\}) = \emptyset$.
            Moreover, $v \in I^*$ as otherwise, $I^* \cap \Y = \emptyset$, contradicting the fact that $\Y$ is an $(X,Z^*)$-\bs in $G_X$.
            This implies Property~\ref{forget1:1} of Case 1, i.e. $(X,Z^*)$ is not $\bar{v}$-critical in $G_X$.

            Let $v' \in \Y_0$, $Z = f(v')$, and $\Y' = \Y \setminus \{v'\}$. As $\Y \vdash (X,\Y_0,\L_1,\L_2,f,\S)$, Property~\ref{def:minb} implies that $\Y'$ is not an $(X,Z)$-\bs in $G_x$.
            Thus, there exists a maximum $(X,Z)$-\is $I$ in $G_X$ such that $I\cap B'=\emptyset$. As $v \in \Y'$, we deduce that
            $(X,Z)$ is not $v$-critical in $G_X$, implying Property~\ref{forget1:1bisbis} of Case 1.
            Let us now prove Property~\ref{forget1:2} and along the proof we will verify that Property~\ref{forget1:1bis} is also satisfied.

            Thus, let us now check all properties of Definition~\ref{def:tchack} to prove that $\Y \vdash (X \cup \{v\},\Y_0 \cup \{v\},\L_1^C,\L_2^C,f^C,\S^C)$, where $\L_1^C$, $\L_2^C$,$f^C$, and $\S^C$ are defined as in Case 1.

            \emph{Property~\ref{def:L1L2}}.
            Let us prove that $\Y$ is an $(X^C,Z)$-\bs in $G_{X^C}$, for each $Z\in \L_1^C\cup\L_2^C$.
            By Lemma~\ref{lemma:liste}, for each $Z \in \L_1^C \cup (\L_2^C \setminus \{Z^* \cup \{v\}\})$,
            we know that $\Y$ is an $(X^C,Z)$-\bs in $G_{X^C}$.
            Moreover, recall that $Z^*\in \L_1$ and that $\Y \vdash (X,\Y_0,\L_1,\L_2,f,\S)$. Thus, $\Y$ is an $(X,Z^*)$-\bs in $G_X$ and, as $(X,Z^*)$ is either $v$-critical or $v$-mixed in $G_X$,
            this implies by Lemma~\ref{lemma:Y} that $\Y$ is an $(X^C,Z^* \cup \{v\})$-\bs in $G_{X^C}$.

            \emph{Property~\ref{def:S}}.
            We now prove that $\Y$ is not an $(X^C,Z)$-\bs in $G_{X^C}$, for each $Z\in \S^C = \S$.
            Let $Z \in \S$. As $\Y \vdash (X,\Y_0,\L_1,\L_2,f,\S)$, $\Y$ is not an $(X,Z)$-\bs in $G_X$. Since $v \in \Y$, it follows that $(X,Z)$ is not $v$-critical in $G_X$, implying Property~\ref{forget1:1bis} of Case 1.
            If $(X,Z)$ is $\bar{v}$-critical in $G_X$, then by Lemma~\ref{lemma:Y}, we get that $\Y$ is not an $(X^C,Z)$-\bs in $G_{X^C}$.
            If $(X,Z)$ is $v$-mixed in $G_X$, since $\Y$ is not an $(X,Z)$-\bs in $G_X$, by Lemma~\ref{lemma:Y} we get that $\Y$ is not an $(X^C,Z)$-\bs in $G_{X^C}$
            or $\Y$ is not an $(X^C,Z \cup \{v\})$-\bs in $G_{X^C}$. The latter case is not possible as $v \in \Y$, and thus we get the desired property.

            \emph{Property~\ref{def:mina}}.
            Let us prove that, for each $v'\in \Y\setminus (\Y_0\cup \{v\})$, there is $Z\in \L_1^C$ such that $\Y\setminus \{v'\}$ is not an $(X^C,Z)$-\bs in $G_{X^C}$.
            Let $v' \in \Y \setminus (\Y_0 \cup \{v\})$.
            As  $\Y \vdash (X,\Y_0,\L_1,\L_2,f,\S)$, there exists $Z \in \L_1$ such that $\Y' = \Y \setminus \{v'\}$ is not an $(X,Z)$-\bs in $G_X$.
            Notice that as $v \in \Y'$, $(X,Z)$ is not $v$-critical in $G_X$. If $(X,Z)$ is $\bar{v}$-critical in $G_X$, then by Lemma~\ref{lemma:Y},
            $\Y'$ is not an $(X^C,Z)$-\bs in $G_{X^C}$, and we are done as $Z \in \L_1^{(\bar{v},X)} \subseteq \L_1^C$. If $(X,Z)$ is $v$-mixed in $G_X$, then by Lemma~\ref{lemma:Y},
            $\Y'$ is not an $(X^C,Z)$-\bs in $G_{X^C}$ or $\Y'$ is not an $(X^C,Z \cup \{v\})$-\bs in $G_{X^C}$. Again, this latter case is not possible as $v \in \Y'$. Thus, we conclude the proof as
            $Z \in \L_1^{(*v,X)}\subseteq \L_1^C$.

            \emph{Property~\ref{def:minb}}.
            To finish this case, let us prove that for each $v'\in \Y_0\cup\{v\}$, $\Y\setminus\{v'\}$ is not an $(X^C,f^C(v'))$-\bs in $G_{X^C}$.
            Let $v' \in \Y_0^C = \Y_0\cup\{v\}$. Let $\Y' = \Y \setminus \{v'\}$.
            Let us first consider the case $v' = v$. In this case, remind that $f^C(v')= Z^*\cup \{v\}$, where, as chosen above, $Z^* \in \L_1$ such that $\Y \setminus \{v\}$ is not
            an $(X,Z^*)$-\bs in $G_X$. Then let us consider $I^*$ defined above, i.e. a maximum $(X,Z^*)$-\is in $G_X$ such that $I^* \cap (\Y \setminus \{v\}) = \emptyset$.
            As $v \in I^*$, according to Lemma~\ref{lemma:I1}, $I^*$ is a maximum $(X^C, Z^* \cup \{v\})$-\is in $G_{X^C}$, and $\Y' \cap I^* = \emptyset$.
            Suppose now that $v' \neq v$ and remind that, in this case, $f^C(v')=f(v')$. Then, since $\Y \vdash (X,\Y_0,\L_1,\L_2,f,\S)$ and $v'\in \Y_0$, we have that $\Y'$ is not an $(X,Z)$-\bs in $G_X$, where $Z = f(v') \in \L_2$.
            As $v \in \Y'$, we deduce that $(X,Z)$ is not $v$-critical in $G_x$.
            If $(X,Z)$ is $v$-mixed in $G_X$, then by Lemma~\ref{lemma:Y},
            $\Y'$ is not an $(X^C,Z)$-\bs in $G_{X^C}$ or $\Y'$ is not an $(X^C,Z \cup \{v\})$-\bs in $G_{X^C}$. This last case is again not possible as $v \in \Y'$.
            Thus, we deduce that $\Y'$ is not an $(X^C,Z)$-\bs.
            If $(X,Z)$ is $\bar{v}$-critical in $G_X$, then by Lemma~\ref{lemma:Y},
            we also get that $\Y'$ is not an $(X^C,Z)$-\bs in $G_{X^C}$.

            \medskip

            Suppose now $v \notin \Y$. We now prove that Case 2 of lemma's statement holds.
            Since $\Y \vdash (X,\Y_0,\L_1,\L_2,f,\S)$, we have that, for each $Z \in \S^{(*v,X)}$, $\Y$ is not an $(X,Z)$-\bs in $G_X$. By Lemma~\ref{lemma:Y} we get
            that either $\Y$ is not an $(X^C,Z)$-\bs in $G_{X^C}$, in which case we add $Z$ to $\S^B$, and if it is not the case
            then $\Y$ is not an $(X^C,Z \cup \{v\})$-\bs in $G_{X^C}$, in which case we add $Z$ to $\S^A$.
            It remains to define the function $f^C$ for $v'\neq v$ such that $f(v') \in \L_2^{(*v,X)}$.
            Let $v' \neq v$ such that $f(v') \in \L_2^{(*v,X)}$, and let $Z = f(v')$.
            Since $\Y \vdash (X,\Y_0,\L_1,\L_2,f,\S)$, $\Y \setminus \{v'\}$ is not an $(X,Z)$-\bs in $G_X$. By Lemma~\ref{lemma:Y}, as $(X,Z)$ is mixed in $G_X$,
            we get that either $\Y \setminus \{v'\}$ is not an $(X^C,Z)$-\bs in $G_{X^C}$, in which case we define $f^C(v')=Z$,
            and if it not the case then $\Y \setminus \{v'\}$ is not an $(X^C,Z \cup \{v\})$-\bs in $G_{X^C}$, in which case we define $f^C(v')=Z \cup \{v\}$.
            \mbvfinale{on pourrait enelver le il exists $f^C$ car $f^C$ est determinee par le reste, mais ca alourdit l'ennonce du lemme}
            Let us now prove that $\Y \vdash (X^C,\Y_0,\L_1^C,\L_2^C,f^C,\S^C)$ where $\L_1^C$, $\L_2^C$, and $\S^C$ are defined as in the Case 2 of the lemma's statement, by verifying that the required properties in Definition~\ref{def:tchack} are satisfied. To prove Property~\ref{def:L1L2}, one should argue that $\Y$ is an $(X^C,Z)$-\bs in $G_{X^C}$, for each $Z\in\L_1^C\cup\L_2^C$ where $\L_1^C = a_v(\L_1^{(v,X)}) \cup\L_1^{(\bar{v},X)} \cup a_v(\L_1^{(*v,X)}) \cup \L_1^{(*v,X)}$ and $\L_2^C = a_v(\L_2^{(v,X)}) \cup \L_2^{(\bar{v},X)} \cup a_v(\L_2^{(*v,X)}) \cup \L_2^{(*v,X)}$. It is immediate using Lemma~\ref{lemma:liste} and the hypothesis that $\Y \vdash (X,\Y_0,\L_1,\L_2,f,\S)$.

            \emph{Property~\ref{def:S}}. Let us now prove that $\Y$ is not an $(X^C,Z)$-\bs in $G_{X^C}$, for each $Z\in \S^C = a_v(\S^{(v,X)})\cup \S^{(\bar{v},X)}\cup a_v(S^A)\cup a_v(S^B)$.
            Let $Z \in \S^C$. If $Z \in a_v(\S^{(v,X)})$ then $Z = Z' \cup \{v\}$ where $Z' \in \S^{(v,X)}$.
            Since $\Y \vdash (X^C,\Y_0,\L_1^C,\L_2^C,f^C,\S^C)$, we know that $\Y$ is not an $(X,Z')$-\bs in $G_X$. As $(X,Z')$ is $v$-critical in $G_X$, by Lemma~\ref{lemma:Y}
            we know that $\Y$ is not an $(X^C,Z' \cup \{v\})$-\bs in $G_{X^C}$.
            If $Z \in \S^{(\bar{v},X)}$ then by the hypothesis $\Y \vdash (X^C,\Y_0,\L_1^C,\L_2^C,f^C,\S^C)$, we know that $\Y$ is not an $(X,Z)$-\bs in $G_X$. As $(X,Z)$ is $\bar{v}$-critical in $G_X$, by Lemma~\ref{lemma:Y} we know that $\Y$ is not an $(X^C,Z)$-\bs in $G_{X^C}$.
            If $Z \in a_v(\S^A) \cup \S^B$, then by definition of $\S^A$ and $\S^B$, $\Y$ is not an $(X^C,Z)$-\bs in $G_{X^C}$.

            \emph{Property~\ref{def:mina}}.
            We must now prove that, for each $v'\in\Y\setminus \Y_0$, there is $Z\in \L_1^C = a_v(\L_1^{(v,X)}) \cup\L_1^{(\bar{v},X)} \cup a_v(\L_1^{(*v,X)}) \cup \L_1^{(*v,X)}$
            such that $\Y\setminus\{v'\}$ is not an $(X^C,Z)$-\bs in $G_{X^C}$.
            Recall that $v\notin \Y$ and let $v' \in \Y \setminus (\Y_0 \cup \{v\})$.
            As  $\Y \vdash (X,\Y_0,\L_1,\L_2,f,\S)$, there exists $Z \in \L_1$ such that $\Y' = \Y \setminus \{v'\}$ is not an $(X,Z)$-\bs in $G_X$.
            By Lemma~\ref{lemma:liste} with the list $\L = \{Z\}$, there exists $Z' \in \L_1^C$ such that $\Y'$ is not an $(X^C,Z')$-\bs in $G_{X^C}$.

            \emph{Property~\ref{def:minb}}.
            Let us finally prove that, for each $v'\in \Y_0^C=\Y_0$, $\Y\setminus\{v'\}$ is not an $(X^C,f^C(v'))$-\bs in $G_{X^C}$.
            Let $v' \in \Y_0^C$. As $v' \in \Y_0$, by the hypothesis $\Y \vdash (X,\Y_0,\L_1,\L_2,f,\S)$ we know that $\Y'=\Y \setminus \{v'\}$ is not an $(X,Z)$-\bs in $G_X$, where $Z = f(v')$ with $Z \in \L_2$.
            If $(X,Z)$ is $\bar{v}$-critical in $G_X$, then by Lemma~\ref{lemma:Y},
            we also get that $\Y'$ is not an $(X^C,Z)$-\bs in $G_{X^C}$, and we are done as $Z=f^C(v')$.
            If $(X,Z)$ is $v$-critical in $G_X$, then by Lemma~\ref{lemma:Y},
            we also get that $\Y'$ is not an $(X^C,Z \cup \{v\})$-\bs in $G_{X^C}$, and we are done as $Z \cup \{v\}=f^C(v')$.
            Finally, $(X,Z)$ is $v$-mixed then $Z \in \L_2^{(*v,X)}$ and by the definition of $f^C(v')$ we get that $\Y'$ is not an $(X^C,f(v'))$-\bs in $G_{X^C}$.
            
            \medskip

            For the backward implication, let us prove that $\Y \vdash (X,\Y_0,\L_1,\L_2,f,\S)$, by distinguishing again both cases in the statement of the lemma.

            \medskip
            \noindent \textbf{Case 1}. Suppose that there exist $Z^* \in \L_1$ as required in Case 1.
            Property~\ref{def:L1L2} follows directly from Lemma~\ref{lemma:liste}. Let us verify that the other properties of Definition~\ref{def:tchack} are also verified.

            \emph{Property~\ref{def:S}}.
            Let us prove that $\Y$ is not $(X,Z)$-\bs in $G_X$, for each $Z\in\S=\S^C$.
            Let $Z \in \S$. By hypothesis, $\Y \vdash (X^C,\Y_0^C,\L_1^C,\L_2^C,f^C,\S^C)$ and thus $\Y$ is not an $(X^C,Z)$-\bs in $G_{X^C}$. Consequently, there exists a maximum $(X^C,Z)$-\is $I$ in $G_{X^C}$ such that
            $I \cap \Y = \emptyset$. In addition, since $\Y_0^C = \Y_0\cup\{v\}\subseteq B$, we have that $v \notin I$. As $(X,Z)$ is not $v$-critical in $G_X$,
            we get by Lemma~\ref{lemma:I2} that $I$ is a  maximum $(X,Z)$-\is in $G_{X}$.

            \emph{Property~\ref{def:mina}}.
            We now argue that, for each $v'\in \Y\setminus \Y_0$, there is $Z\in \L_1$ such that $\Y\setminus\{v'\}$ is not an $(X,Z)$-\bs in $G_X$.
            Let $v' \in \Y \setminus \Y_0$.
            If $v'=v$, then as $v \in \Y_0^C$, by definition of $f^C$ we get that $\Y \setminus \{v\}$ is not an $(X^C,Z^* \cup \{v\})$-\bs in $G_{X^C}$.
            As $(X,Z^*)$ is either $v$-critical or $v$-mixed in $G_X$, in both cases by Lemma~\ref{lemma:Y} we get that $\Y \setminus \{v\}$ is not an $(X,Z^*)$-\bs in $G_X$. As by hypothesis $Z^* \in \L_1$, this implies Property~\ref{def:mina}.
            Suppose now $v' \neq v$. Since $\Y \vdash (X^C,\Y_0^C,\L_1^C,\L_2^C,f^C,\S^C)$, we know that there exists $Z \in \L_1^C$ such that $\Y \setminus \{v'\}$ is not an $(X^C,Z)$-\bs in $G_{X^C}$. By Lemma~\ref{lemma:liste}, there exists $Z \in \L_1$ such that $\Y \setminus \{v'\}$ is not an $(X,Z)$-\bs in $G_X$.

            \emph{Property~\ref{def:minb}}.
            Finally, we prove that for each $v'\in \Y_0$, $\Y\setminus\{v'\}$ is not an $(X,f(v'))$-\bs in $G_X$.
            Let $v' \in \Y_0$ and let $Z = f(v')$. Recall that $\Y_0= \Y\cap X$ and thus $v\notin \Y_0$, implying that $v\neq v'$. By definition, we thus have $f^C(v')=f(v')$.
            Since $\Y \vdash (X^C,\Y_0^C,\L_1^C,\L_2^C,f^C,\S^C)$, we know that $\Y'=\Y \setminus \{v'\}$ is not an $(X^C,Z)$-\bs in $G_{X^C}$.
            By Property~\ref{forget1:1bisbis}, $(X,Z)$ is not $v$-critical in $G_{X}$.
            As $(X,Z)$ is $\bar{v}$-critical or $v$-mixed in $G_{X}$, then by Lemma~\ref{lemma:Y}, $\Y'$ is not an $(X,Z)$-\bs in $G_{X}$.

            \medskip
            \noindent \textbf{Case 2}. Suppose that there exist $\S^A, \S^B$, and $f^C$ as required in Case 2. Property~\ref{def:L1L2} follows again directly from Lemma~\ref{lemma:liste}.

            \emph{Property~\ref{def:S}}.
            Let us prove that $\Y$ is not $(X,Z)$-\bs in $G_X$, for each $Z\in\S$.
            Let $Z \in \S$. Recall that $\S^C = a_v(\S^{(v,X)})\cup \S^{(\bar{v},X)}\cup a_v(\S^A)\cup \S^B$ and that $\Y \vdash (X^C,\Y_0,\L_1^C,\L_2^C,f^C,\S^C)$ as claimed in Case 2. If $Z \in \S^{(v,X)}$, then $\Y$ is not an $(X^C,Z \cup \{v\})$-\bs in $G_{X^C}$, and by Lemma~\ref{lemma:Y},
            as $(X,Z)$ is $v$-critical in $G_X$,  $\Y$ is not an $(X,Z)$-\bs in $G_{X}$.
            If $Z \in \S^{(\bar{v},X)}$, then $\Y$ is not an $(X^C,Z)$-\bs in $G_{X^C}$, and by Lemma~\ref{lemma:Y},
            as $(X,Z)$ is $\bar{v}$-critical in $G_X$,  $\Y$ is not an $(X,Z)$-\bs in $G_{X}$.
            It remains to treat the case where $Z \in \S^{(*v,X)} = \S^A \cup \S^B$. In this case, if $Z \in \S^A$ then $\Y$ is not an $(X^C,Z \cup \{v\})$-\bs in $G_{X^C}$,
            or ($Z \in \S^B$) $\Y$ is not an $(X^C,Z)$-\bs in $G_{X^C}$. In both cases, as $(X,Z)$ is $v$-mixed in $G_X$, by Lemma~\ref{lemma:Y} $\Y$ is not an $(X,Z)$-\bs in $G_{X}$.

            \emph{Property~\ref{def:mina}}.
            We now show that, for each $v'\in \Y\setminus \Y_0$, there is $Z\in \L_1$ such that $\Y\setminus\{v'\}$ is not an $(X,Z)$-\bs in $G_X$.
            Let $v' \in \Y \setminus \Y_0$. Recall that, since $\Y \vdash (X^C,\Y_0,\L_1^C,\L_2^C,f^C,\S^C)$, $\Y\cap X^C = B_0\subseteq X$ and $\{v\} = X^C\setminus X$.
            As $v \notin \Y$, we deduce $v' \neq v$. By hypothesis, we know that there exists $Z \in \L_1^C$ such that $\Y \setminus \{v'\}$ is not an $(X^C,Z)$-\bs in $G_{X^C}$. By Lemma~\ref{lemma:liste}, there exists $Z \in \L_1$ such that $\Y \setminus \{v'\}$ is not an $(X,Z)$-\bs in $G_X$.

            \emph{Property~\ref{def:minb}}.
            Finally, we prove that for each $v'\in \Y_0$, $\Y\setminus\{v'\}$ is not an $(X,f(v'))$-\bs in $G_X$.
            Let $v' \in \Y_0$, $Z=f(v')$ and $Z' = f^C(v')$. By hypothesis, we know that $\Y'=\Y \setminus \{v'\}$ is not an $(X^C,Z')$-\bs in $G_{X^C}$.
            If $Z \in \L_2^{(v,X)}$, then by definition of $f^C$ we have $Z' = Z \cup \{v\}$, and by Lemma~\ref{lemma:Y} we get that $\Y'$ is not an $(X,Z)$-\bs in $G_{X}$.
            If $Z \in \L_2^{(\bar{v},X)}$, then by definition of $f^C$ we have $Z' = Z$, and by Lemma~\ref{lemma:Y} we get that $\Y'$ is not an $(X,Z)$-\bs in $G_{X}$.
            Finally, if $Z \in \L_2^{(*v,X)}$, then by definition of $f^C$ we have $Z' \in \{Z,Z \cup \{v\}\}$, and again by Lemma~\ref{lemma:Y} we get that $\Y'$ is not an $(X,Z)$-\bs in $G_{X}$.
          \end{proof}

          \subsection{Putting pieces together}
          \label{sec:together}

          Let us now assume once again that the input graph $G$ and the nice tree decomposition $\D$ of $G$ are provided, and let us describe a recursive algorithm $\A$
          that solves problem $\Pi$.
          We distinguish several cases as follows. In each case, namely join, introduce, or forget, we use the notations introduced in the corresponding lemma, namely Lemma~\ref{lem:join}, Lemma~\ref{lem:intro}, or Lemma~\ref{lem:forget}, respectively, and define Algorithm~$\A$ as follows.

           \renewcommand{\P}{{\cal P}}

          \begin{definition}\label{def:algo-join}
            Suppose we are given an instance $I=(X,\Y_0,\L_1,\L_2,f,\S) \in \E$ of problem $\Pi$ such that $X$ is a join node with children $X^L = X^R = X$.
            For each collection $\P = \{\L_1^A, \L_1^B, \L_1^C,\L_2^A, \L_2^B, \L_2^C\}$ such that
            \begin{itemize}
              \item $\L_1=\L_1^A \uplus \L_1^B \uplus \L_1^C$ and
              \item $\L_2=\L_2^A \uplus \L_2^B \uplus \L_2^C$,
            \end{itemize}
            we denote by $I^L(\P) = (X^L,\Y_0,\L_1^L,\L_2^L,f^L,\S^L)$ and
            $I^R(\P) = $ $(X^R,\Y_0,$ $\L_1^R,\L_2^R,f^R,\S^R)$ as defined by Lemma~\ref{lem:join}.

            In the join case, Algorithm~$\A$ enumerates all such collections, and returns $\A(I^L(\P)) \cup \A(I^R(\P))$, where $\P$ maximizes $|A(I^L(\P))|+|A(I^R(\P))|$.
          \end{definition}

          \begin{definition}\label{def:algo-intro}
            Suppose we are given an instance $I=(X,\Y_0,\L_1,\L_2,f,\S) \in \E$ of problem $\Pi$ such that $X$ is an introduce node with child $X^C = X \setminus \{v\}$.
            For each collection $\P = \{\L_2^A, \L_2^B\}$ of $\L_2 $ as required in Case $1$ of Lemma~\ref{lem:intro}, 
            we denote by  $I^1(\P)$ the instance defined in Case $1$ of Lemma~\ref{lem:intro}, and we denote by $I^2$ the instance defined in Case $2$ of Lemma~\ref{lem:intro}.

            In the introduce case, if $v \in \Y_0$ then Algorithm~$\A$ enumerates all such collections and returns $\{v\} \cup \A(I^1(\P))$, where $\P$ maximizes $|\A(I^1(\P))|$,
            and if $v \notin \Y_0$ then Algorithm~$\A$ returns $\A(I^2)$.
          \end{definition}

          \begin{definition}\label{def:algo-forget}
            Suppose we are given an instance $I=(X,\Y_0,\L_1,\L_2,f,\S) \in \E$ of problem $\Pi$ such that $X$ is a forget node with child $X^C = X \cup \{v\}$.
            For each $Z^* \in \L_1$  as required in Case $1$ of Lemma~\ref{lem:forget}, we denote by $I^1(Z^*)$ the instance defined in Case $1$ of Lemma~\ref{lem:forget},
            and for each partition $\P = \{\S^A,\S^B\}$ of $\S^{(*v,X)}$ and function $f_2^C$ as required in Case $2$ of Lemma~\ref{lem:forget}, we denote by $I^2(\P,f_2^C)$
            the instance defined in Case $2$ of Lemma~\ref{lem:forget}.

            In the forget case, Algorithm~$\A$ enumerates all sets $Z^* \in \L_1$ as required in Case $1$, and computes $\Y^1 = \A(I^1(Z^*))$
            where $Z^*$ maximizes $|I^1(Z^*)|$. Then, Algorithm~$\A$ enumerates all sets $\S^A,\S^B$, and functions $f_2^C$ as required in Case $2$, and computes $\Y^2 = \A(I^2(\P,f_2^C))$
            where $\P,f_2^C$ maximizes $|I^2(\P,f_2^C)|$. Finally, Algorithm~$\A$ returns the largest solution among all $\Y^1$ and $\Y^2$.
          \end{definition}

          In any of the three cases (join, forget, introduce), $\A$ returns $\max_{x \in E}\A(x)$ for some appropriate set $E$, and we point out that
          it may be the case that $E=\emptyset$, when  none of the enumerated parameters respect the required conditions of Lemma~\ref{lem:join}, Lemma~\ref{lem:intro}, and Lemma~\ref{lem:forget}, and in this case
          Algorithm~$\A$ returns $-\infty$ instead of a solution.
          Concerning the base case, we can always assume that the underlying tree decomposition has leaves where $X=\emptyset$. On such a leaf, $\emptyset$ is the only candidate solution,
          and thus $\A$ returns $\emptyset$ if it is a valid solution, or $-\infty$ otherwise.

          \begin{lemma}\label{lem:algo:correct}
            $\A$ solves $\Pi$ optimally: for every instance $I \in \E$, if $I$ is feasible then $A(I)$ returns an optimal solution. Otherwise, it returns $- \infty$.
          \end{lemma}

          \begin{proof}
            The proof is by induction on the number  of  remaining bags  in $G_X$ in the provided nice tree decomposition $\D$.
            Let $\Y$ be the solution returned by $\A(I)$, and let $\Y^*$ be an optimal solution of $I$. We distinguish the different types of nodes in the nice tree decomposition $\D$ of $G$.
            In the three types of nodes, if $I$ is not feasible, then by Lemma~\ref{lem:join}, Lemma~\ref{lem:intro}, and Lemma~\ref{lem:forget}, and by the inductive hypothesis, any of the recursive calls
            will output $- \infty$, and thus $\A(I)$ will return $-\infty$ as well. We suppose now that $\I$ is feasible, and let $\Y = \A(I)$ and let $\Y^*$ an optimal solution.
            \smallskip

            \textbf{Join node}. By Lemma~\ref{lem:join}, there exists a collection $\P^*$ as defined in Definition~\ref{def:algo-join} and sets $\Y^{*L},\Y^{*R}$ such that $\Y^{*L} \vdash I^L(\P^*)$ and $\Y^{*R} \vdash I^R(\P^*)$. Let $\P$ be the collection chosen by $\A$.
            We have
            $|\Y| = |A(I^L(\P))|+|A(I^L(\P))| - |\Y_0| \ge |A(I^L(\P^*))|+|A(I^L(\P^*))| - |\Y_0| \ge |\Y^{*L}|+|\Y^{*R}| - |\Y_0| = |\Y^*|$.
            Moreover, by Lemma~\ref{lem:join}, $\Y \vdash I$.

            \smallskip

            \textbf{Introduce node}.
            If $v \in \Y^*$, then according to Case 1 of Lemma~\ref{lem:intro} there exists a collection $\P^*$ as defined in Definition~\ref{def:algo-intro} such that $\Y^* \setminus \{v\} \vdash I^1(\P^*)$.
            Let $\P$ be the collection chosen by $\A$.
            As in this case we have $v \in \Y_0$, we have $|\Y| = 1+|A(I^1(\P))| \ge 1+|A(I^1(\P^*))| \ge 1+ |\Y^* \setminus \{v\}| = |\Y^*|$.
            Moreover, by Lemma~\ref{lem:intro}, $\Y \vdash I$.
            If $v \notin \Y^*$, then according to Case 2 of Lemma~\ref{lem:intro}, $\Y^* \vdash I^2$.
            As in this case we have $v \notin \Y_0$, we have $|\Y| = |A(I^2)| \ge |\Y^*|$.
            Moreover, according to Lemma~\ref{lem:intro}, $\Y \vdash I$.

            \smallskip

            \textbf{Forget node}.   If $v \in \Y^*$, then according to Case 1 of Lemma~\ref{lem:forget} there exist $Z^{**} \in \L_1$ such that $\Y^* \vdash I^1(Z^{**})$.
            Let $Z^*$ be the element chosen by $\A$ for the first case, and $(\P,f_2^C)$ the elements chosen for the second case.
            We have $|\Y| \ge |A(I^1(Z^*))| \ge  |A(I^1(Z^{**}))| \ge |\Y^*|$.
            If $v \notin \Y^*$, then according to Case 2 of Lemma~\ref{lem:forget} there exist $\P^*$ and $f_2^{*C}$ such that $\Y^* \vdash I^2(\P^*,f_2^{*C})$.
            We have $|\Y| \ge |A(I^2(\P,f_2^C))| \ge |A(I^2(\P^*,f_2^{*C}))|  \ge |\Y^*|$.
          \end{proof}

          \begin{lemma}\label{lem:algo:complexity}
            Algorithm~$\A$ runs in time $\O^*(2^{\O(2^{t})})$, where $t$ is the width of the given nice tree decomposition of the input graph.
          \end{lemma}
          \begin{proof}
            The time complexity of Algorithm~$\A$ is $\O^*(x_1 \cdot x_2)$, where $x_1=|\E|$ is the number of possible inputs of $\Pi$ and $x_2$ is the maximum time necessary to compute
            $\A(I)$ for each $I \in \E$. We denote $n=|V(G)|$.
            From the definitions of the corresponding objects, it can be routinely verified that $x_1 \le n \cdot 2^t \cdot 2^t \cdot 2^{2^t} \cdot 2^{2^t} \cdot (2^t)^t \cdot 2^t \le 2^{3t+2^{t+1}+t^2} = 2^{\O(2^t)}$.

            Let us now bound $x_2$. To that end, let $\theta_1$ be an upper bound on the number of enumerated subinstances made in any of the three cases (join, introduce, or forget) and let  by $\theta_2$ be an upper bound on the 
             time complexity related to all operations like taking the minimum, and verifying that each enumerated subinstance verifies required properties
            (like for example, in Case 1 of Lemma~\ref{lem:forget}).
            In the join case, the number of subinstances is at most $3^{|\L_1|} \cdot 3^{|\L_2|} \le 3^{2^{t+1}}$ (to consider all $\L_1^A, \L_1^B, \L_1^C, \L_2^A, \L_2^B, \L_2^C$), and for each subinstance
            the complexity is polynomial in $n$.
            In the introduce case, the number of subinstances is $2^{2^t}+1$ (to consider all $\L_2^A$), and for each subinstance
            the complexity is polynomial in $n$.
            In the forget case, the number of subinstances is at most $2^{2^t}+2^{2^t} \cdot 2^{2^t}$ (to consider all $Z^* \in \L_1$ in Case 1 and all $\S^A$ and $f^C$ in Case 2), and for each subinstance
            the complexity is in $\O^*(2^t)$ as in Case 1, for every $Z \in \S$ (resp. every $v' \in \Y_0$,) we must verify if $(X,Z)$ (resp. $(X,f(v'))$) is not $v$-critical
            in $G_X$, and this verification can be done in time $\O^*(2^t)$ according to Lemma~\ref{prop:critical}.
            All in all, we can choose $\theta_1=3^{2^{t+2}}=2^{\O(2^t)}$, $\theta_2 = \O^*(2^t)$, and the lemma follows.\end{proof}

          As according to~\cite{DBLP:journals/siamcomp/Bodlaender96} we can determine whether $\tw(G) \le t$ in time $\O^*(t^{\O(t^3)})$, and construct the corresponding (nice) tree decomposition of $G$ in case of a positive answer, from Proposition~\ref{prop:solvingPiOK}
          and Lemma~\ref{lem:algo:complexity} the following theorem is now immediate.

          \begin{theorem}\label{thm:tw}
          The \MMBS/\tw problem is \FPT. More precisely, it  can we solved in time $\O^*(2^{\O(2^{\tw(G)})})$.
          \end{theorem}

\section{Further research}
\label{sec:conc}
\label{section:conclusion}

We presented a number of negative and positive results for the \MMBS and \MMHS problems. Several interesting questions remain open. Concerning \MMBS, even if seems implausible that the problem could be expressed in monadic second-order logic, it would be nice to prove it formally. For that, one may try to use the framework introduced by van Bevern et al.~\cite{BevernDFGR15}. Simplifying the dynamic programming algorithm behind Theorem~\ref{thm:tw} is also worth trying.

As for \MMHS, we believe that the main challenge is trying to get a algorithm parameterized by $\alpha+\beta$ running faster than $\O^*(2^{\alpha\beta})$ (Theorem~\ref{thm:fptv2}).
Let us consider the case $\alpha=2$, corresponding to the \MMVC problem.
The parameterized complexity of \MMVC has received some attention recently, with results concerning \FPT algorithms in time $\O^*(2^\beta)$~\cite{FernauHDR}, and even in time $\O^*(c^\beta)$ for $c < 1.54$~\cite{BORIA201562},
\FPT algorithms for structural parameterizations~\cite{Zehavi17}, and kernelization~\cite{mmvc2021}.
This motivates the  problem of trying to improve the running time $\O^*(2^{\alpha \beta})$ of Theorem~\ref{thm:fptv2}, for example by solving \MMHS in
time $\O^*(\alpha^\beta)$, or even  $\O^*(\alpha^{\O(\beta)})$.
Recall that the algorithm of
 Proposition~\ref{prop:fixedalpha_beta_fpt} runs in time $\O^*(\alpha^\beta)$ for {\sl fixed $\alpha$}, and that  it hides a term $|V(\H)|^{f(\alpha)}$ for some function $f$.

 
Achieving a running time of $\O^*(\alpha^\beta)$ might be typically done by  guessing, at each step, only which vertex of a given hyperedge should be added to the solution.
However, guessing only a vertex $v$ and applying recursion on a remaining instance $(\H',\beta-1)$, where $\H'$ is defined by removing $v$ and all hyperedges containing $v$, is not correct.
Indeed, $(\H',\beta-1)$ being a \yes-instance, certified by a solution $S'$, does not imply that $(\H,\beta)$ is also a \yes-instance, as $S' \cup \{v\}$ may not be minimal anymore.
Thus, we believe that, in order to solve \MMHS in time $\O^*(\alpha^\beta)$ or even $\O^*(\alpha^{\O(\beta)})$, a significantly new approach should be devised.

\medskip
\noindent \textbf{Acknowledgement}. We would like to thank Mamadou Moustapha Kanté for helpful suggestions concerning the non-expressibility of problems in monadic second-order logic.


\bibliography{biblio}

\end{document}